\documentclass[11pt, letterpaper]{article}
\usepackage{fullpage}
\usepackage[T1]{fontenc}
\usepackage{amsmath,amsfonts,amsthm}
\usepackage{algorithm}
\usepackage[noend]{algpseudocode}
\usepackage[colorlinks=true,pdfpagemode=none,linkcolor=blue,citecolor=blue]{hyperref}
\usepackage{color}
\usepackage{multirow}
\usepackage{enumerate}
\usepackage{graphicx}
\usepackage{comment}
\hypersetup{pdftex,colorlinks=true,allcolors=blue}

\usepackage{hyperref}
\hypersetup{
    bookmarksnumbered=true,     
    bookmarksopen=true,         
    bookmarksopenlevel=1,       
    colorlinks=true,            
    pdfstartview=Fit,           
    pdfpagemode=UseOutlines
}

\newtheorem{theorem}{Theorem}[section]

\newtheorem{lemma}[theorem]{Lemma}

\newtheorem{definition}[theorem]{Definition}

\newtheorem{fact}[theorem]{Fact}

\newcommand{\cL}{\mathcal{L}}
\newcommand{\cC}{\mathcal{C}}
\newcommand{\cP}{\mathcal{P}}
\newcommand{\DP}{\textsc{DP}}

\newcommand{\onev}{\mathbf{1}}
\newcommand{\zerov}{\mathbf{0}}

\newcommand{\tO}[1]{\widetilde{O}\left(#1\right)}
\newcommand{\tOm}[1]{\widetilde{\Omega}\left(#1\right)}

\newcommand{\tSC}{\widetilde{SC}}

\newcommand{\hG}{\widehat{G}}

\newcommand{\cV}{\mathcal{V}}

\newcommand{\hT}{\widehat{T}}

\newcommand{\eps}{\varepsilon}

\newcommand{\tf}{\widetilde{f}}

\newcommand{\of}{\bar{f}}

\newcommand{\os}{\bar{s}}

\newcommand{\hr}{\widehat{r}}

\newcommand{\hmu}{\hat{\mu}}

\newcommand{\heps}{\widehat{\eps}}

\newcommand{\vphi}{\boldsymbol{\mathit{\phi}}}
\newcommand{\vpsi}{\boldsymbol{\mathit{\psi}}}

\newcommand{\ovphi}{\boldsymbol{\mathit{\bar{\phi}}}}

\newcommand{\vrho}{\boldsymbol{\mathit{\rho}}}
\newcommand{\vDelta}{\boldsymbol{\mathit{\Delta}}}
\newcommand{\tvDelta}{\boldsymbol{\mathit{\widetilde{\Delta}}}}

\newcommand{\vpi}{\boldsymbol{\mathit{\pi}}}
\newcommand{\vPi}{\boldsymbol{\mathit{\Pi}}}

\newcommand{\tvpi}{\widetilde{\boldsymbol{\mathit{\pi}}}}
\newcommand{\tpi}{\widetilde{\mathit{\pi}}}

\newcommand{\hp}{\widehat{\mathit{p}}}

\newcommand{\ophi}{\bar{\phi}}

\newcommand{\cc}{\boldsymbol{\mathit{c}}}
\newcommand{\dd}{\boldsymbol{\mathit{d}}}

\newcommand{\ff}{\boldsymbol{\mathit{f}}}
\newcommand{\tff}{\boldsymbol{\mathit{\widetilde{f}}}}
\newcommand{\tR}{{\widetilde{R}}}
\newcommand{\hR}{{\widehat{R}}}

\newcommand{\off}{\boldsymbol{\mathit{\bar{f}}}}

\newcommand{\hh}{\boldsymbol{\mathit{h}}}

\newcommand{\pp}{\boldsymbol{\mathit{p}}}
\newcommand{\qq}{\boldsymbol{\mathit{q}}}
\newcommand{\rr}{\boldsymbol{\mathit{r}}}
\newcommand{\hrr}{\boldsymbol{\mathit{\widehat{r}}}}
\newcommand{\orr}{\boldsymbol{\mathit{\bar{r}}}}

\renewcommand{\ss}{\boldsymbol{\mathit{s}}}
\newcommand{\oss}{\boldsymbol{\bar{\mathit{s}}}}
\newcommand{\uu}{\boldsymbol{\mathit{u}}}

\newcommand{\vv}{\boldsymbol{\mathit{v}}}

\newcommand{\xx}{\boldsymbol{\mathit{x}}}

\newcommand{\BB}{\boldsymbol{\mathit{B}}}
\newcommand{\CC}{\boldsymbol{\mathit{C}}}

\newcommand{\GG}{\boldsymbol{\mathit{G}}}

\newcommand{\oHH}{\boldsymbol{\overline{\mathit{H}}}}
\newcommand{\II}{\boldsymbol{\mathit{I}}}

\newcommand{\LL}{\boldsymbol{\mathit{L}}}

\newcommand{\QQ}{\boldsymbol{\mathit{Q}}}
\newcommand{\RR}{\boldsymbol{\mathit{R}}}

\renewcommand{\epsilon}{\varepsilon}

\newcommand{\ha}{{\mathit{\widehat{a}}}}

\usepackage[dvipsnames]{xcolor}

\title{Faster Sparse Minimum Cost Flow by Electrical Flow Localization}
\author{
Kyriakos Axiotis\thanks{MIT, \tt{kaxiotis@mit.edu}}
\and
Aleksander M\k{a}dry\thanks{MIT, {\tt madry@mit.edu}} 
\and
Adrian Vladu\thanks{CNRS \& IRIF, Universit\'{e} de Paris, \tt{vladu@irif.fr}}
}
\newcommand{\e}{\epsilon_{\mathrm{step}}}
\newcommand{\es}{\epsilon_{\mathrm{solve}}}

\newcommand{\bc}{\beta_{\textsc{Checker}}}

\makeatletter
\newtheorem*{rep@theorem}{\rep@title}
\newcommand{\newreptheorem}[2]{%
\newenvironment{rep#1}[1]{%
 \def\rep@title{#2 \ref{##1}}%
 \begin{rep@theorem}}%
 {\end{rep@theorem}}}
\makeatother

\newreptheorem{theorem}{Theorem}
\newreptheorem{lemma}{Lemma}

\date{}
\begin{document}
\maketitle

\begin{abstract}
We give an $\widetilde{O}({m^{3/2 - 1/762} \log (U+W))}$ time algorithm for minimum cost flow
with capacities bounded by $U$ and costs bounded by $W$. For sparse graphs with general capacities, this is the first algorithm
 to improve over the $\widetilde{O}({m^{3/2} \log^{O(1)} (U+W)})$ running time 
obtained by an appropriate instantiation of an interior point method [Daitch-Spielman, 2008].

Our approach is extending the framework put forth in [Gao-Liu-Peng, 2021] for computing the maximum flow in graphs with large capacities and, in particular, demonstrates how to reduce the problem of computing an electrical flow with general demands to the same problem on a sublinear-sized set of vertices---even if the demand is supported on the entire graph.
Along the way, we develop new machinery to assess the importance of the graph's edges at each phase of the interior point method optimization process. This capability relies on establishing a new connections between the electrical flows arising inside that optimization process and vertex distances in the corresponding effective resistance metric.
\end{abstract}

\newpage
\tableofcontents

\section{Introduction}
In the last decade, continuous optimization has proved to be an invaluable tool for designing graph algorithms, often leading to significant
improvements over the best known combinatorial algorithms. This has been particularly true in the context of flow problems---arguably, some of the most prominent graph problems~\cite{daitch2008faster, christiano2011electrical,lee2013new, madry2013navigating, sherman2013nearly, kelner2014almost, lee2014path,peng2016approximate, madry2016computing, cmsv17,sherman2017generalized, sherman2017area, sidford2018coordinate, liu2019faster, liu2020faster, axiotis2020circulation, van2020bipartite, van2021minimum}. %
Indeed, these developments have brought a host of remarkable improvements in a variety of regimes, such as when seeking only approximate solutions, or when the underlying graph is dense. However, most of these improvements did not fully address the challenge of seeking \emph{exact} solutions in \emph{sparse} graphs. Fortunately, the improvements for that  regime eventually emerged~\cite{madry2013navigating,madry2016computing,cmsv17,liu2019faster,liu2020faster,axiotis2020circulation}. They still suffered though from an important shortcoming: they all had a polynomial running time dependency in the graph's capacities, and hence---in contrast to the classical combinatorial algorithms---they did not yield efficient algorithms in the presence of arbitrary capacities.
Recently, Gao, Liu and Peng \cite{gao2021fully} have finally changed this state of affairs by providing the first \emph{exact} maximum flow algorithm to break the $\tO{m^{3/2}\log^{O(1)} U}$ barrier for sparse graphs with \emph{general} capacities (bounded by $U$). Their approach, however, crucially relies on a preconditioning technique that is specific to the maximum flow problem and, in particular, having an $s$-$t$ demand---rather than a general one. As a result, the corresponding improvement held only for that particular problem. 

In this paper, we demonstrate how to circumvent these limitations and provide the first algorithm that breaks the 
$\tO{m^{3/2}\log^{O(1)} (U+W)}$ barrier for the \emph{minimum cost flow} problem
in \emph{sparse} graphs with general demands, capacities (bounded by $U$), and costs (bounded by $W$).
This algorithm runs in time $\tO{m^{3/2 - 1/762}\log (U+W)}$.

\subsection{Previous work}
In 2013, M\k{a}dry~\cite{madry2013navigating}
presented the first running time improvement to the maximum flow problem
since the $\tO{m\sqrt{n} \log U}$ algorithm of~\cite{goldberg1998beyond} in the regime of sparse graphs with small capacities. To this end, he presented an algorithm that runs in time
$\tO{m^{10/7}\mathrm{poly}(U)}$,
where $U$ is a bound on edge capacities, breaking past the $\tO{m^{3/2}}$ running time barrier that has for decades resisted improvement attempts. The main idea in that work was to use an interior point method with
an improved number of iterations guarantee that was delivered via use of an adaptive re-weighting of 
the central path and careful perturbations of the problem instance. 
Building on this framework, a series of subsequent works~\cite{madry2016computing,liu2019faster,liu2020faster}
has brought the runtime of sparse max flow down to $\tO{m^{4/3}\mathrm{poly}(U)}$. (With the most recent of these works crucially relying on nearly-linear time $\ell_p$ flows~\cite{kyng2019flows}.)
In parallel~\cite{cmsv17,axiotis2020circulation}, the running time of the 
more general minimum cost flow problem was reduced to $\tO{m^{4/3}\mathrm{poly}(U)\log W}$,
where $W$ is a bound on edge costs.

However, even though these algorithms offer a significant improvement when
$U$ is relatively small, the question of whether there exists an algorithm
faster than $\tO{m^{3/2} \log^{O(1)} U}$ for sparse graphs with general  capacities remained open. In fact,
a polynomial dependence on capacities or costs seems inherent in the 
central path re-weighting technique used in all the aforementioned works.
Recently, \cite{gao2021fully} finally made progress on this question by developing an algorithm for the maximum flow problem that
runs in time $\tO{m^{3/2 - 1/328}\log U}$. The source of improvement here was different
from previous works, in the sense that it was not based on decreasing the number of iterations
of the interior point method. Instead,
it was based on devising a data structure 
to solve the dynamically changing Laplacian system required by the interior point method in sublinear
time per iteration.

The new approach put forth by \cite{gao2021fully}, despite being quite different to the prior ones, still leaned on 
the preconditioning approach of~\cite{madry2016computing}, as well as on other properties that
are specific to the maximum flow problem. For this reason, this improvement did not extend
to the minimum cost flow problem with general capacities, for which the fastest known runtime
was still
$\tO{m \log(U+W)+ n^{1.5} \log^2 (U+ W)}$~\cite{van2021minimum} and $\widetilde{O}({m^{3/2} \log^{O(1)} (U+W)})$~\cite{daitch2008faster} in the sparse regime.

\subsection{Our result}

In this work, we give an algorithm for the minimum cost flow problem with a running time of $\tO{m^{3/2 - 1/762}\log (U+W)}$. This is the first improvement for sparse graphs with general capacities
over~\cite{daitch2008faster},
which runs in time $\tO{m^{3/2}\log^{O(1)} (U+W)}$.
Specifically, we prove that:

\begin{theorem}
\sloppy Given a graph $G(V,E)$ with edge costs $\cc\in\mathbb{Z}_{[-W,W]}^m$,
a demand $\dd\in\mathbb{R}^n$, and capacities $\uu\in\mathbb{Z}_{(0,U]}^m$,
there exists an algorithm that with high probability
runs in time $\widetilde{O}\left(m^{3/2-1/762}\log (U+W)\right)$ and
returns a flow $\ff\in[\zerov,\uu]$ in $G$ such that
$\ff$ routes the demand $\dd$
and the cost $\langle \cc, \ff\rangle$ is minimized.

\label{thm:main}
\end{theorem}

\subsection{High level overview of our approach}
As we build on the approach presented in~\cite{gao2021fully},
we first briefly overview some of the key ideas introduced there that will also be relevant for our discussion.
The maximum flow interior point method by \cite{madry2016computing} works by, repeatedly
over $\tO{\sqrt{m}}$ steps,
taking an
electrical flow step that is a multiple of 
\[ \tff = \RR^{-1} \BB \LL^+ \BB^\top \onev_{st}\,, \]
where $\LL = \BB^\top \RR^{-1} \BB$ is a Laplacian matrix and $\rr$ are resistances that
change per step. 
However, $\tff$ has $m$ entries and takes $\tO{m}$ to compute, which gives the standard $\tO{m^{3/2}}$
bound. To go beyond this,
\cite{gao2021fully} show that it suffices to compute $\tff$ for only a \emph{sublinear} number 
of high-congestion entries of $\tff$, where congestion is defined as $\vrho=\sqrt{\rr}\tff$. By known 
linear sketching
results, these edges can be detected by computing the inner product $\langle\qq,\vrho\rangle$
for a small number of randomly chosen vectors $\qq\in\mathbb{R}^m$. Crucially, 
given a vertex subset $C\subseteq V$ of sublinear size
that contains $s$ and $t$, 
this inner product can be equivalently 
written as the following sublinear-sized inner product
\begin{align}
\left\langle \qq,\vrho\right\rangle = \left\langle \vpi^C\left(\BB^\top \frac{\qq}{\sqrt{\rr}}\right), 
SC^+ \dd\right\rangle \,, \label{eq:identity}
\end{align}
where $SC := SC(G,C)$ is the \emph{Schur complement} of $G$ onto $C$,
$\dd$ is equal to $\BB^\top \onev_{st}$, and 
$\vpi^C\left(\cdot\right)$ is a \emph{demand projection} onto $C$.
Therefore, the problem is reduced to maintaining two
quantities: $\vpi^C\left(\BB^\top \frac{\qq}{\sqrt{\rr}}\right)$
and $(SC(G,C))^+ \dd$ in sublinear time per operation. The latter is computed by
using the dynamic Schur complement data structure of~\cite{durfee2019fully},
and the former can be maintained by a careful use of random walks.

We now describe our approach. Instead of using the 
interior point method formulation of~\cite{madry2016computing} which only applies to the 
maximum flow problem,
we use the one by~\cite{axiotis2020circulation} for the, more general, minimum cost flow problem.

There are now several obstacles to making this approach work
by maintaining the quantity $\langle \qq,\vrho\rangle$:

\paragraph{Preconditioning}
A significant difference between~\cite{madry2016computing} and~\cite{axiotis2020circulation} is that while the former is able to guarantee 
that the magnitude of the electrical potentials computed in each step is inversely proportional to the duality gap,
meaning that a large duality gap implies potential embeddings of low stretch,
no such preconditioning method is known for minimum cost flow. In fact, \cite{axiotis2020circulation} used demand perturbations to show that
a \emph{weaker} bound on the potentials can be achieved, which was still sufficient for their purposes.
Unfortunately, this bound is not strong enough to be used in the analysis of~\cite{gao2021fully}.

In order to alleviate this issue, we completely remove preconditioning from the picture by only requiring a bound on the \emph{energy} of the
electrical potentials (instead of their magnitude). 
In particular, given an approximate demand projection ${\tvpi}^{C}\left(\BB^\top \frac{\qq}{\sqrt{\rr}}\right)$,
identity
\eqref{eq:identity} is used to detect congested edges.
In~\cite{gao2021fully}, there is a uniform upper bound on
the entries of the potential embedding 
$\vphi = SC^+ \dd$
because of preconditioning, thus the error in $(\ref{eq:identity})$ can be bounded by
\begin{align*}
\left\|\tvpi^C\left(\BB^\top \frac{\qq}{\sqrt{\rr}}\right) - \vpi^C\left(\BB^\top \frac{\qq}{\sqrt{\rr}}\right)\right\|_1 \left\|\vphi\right\|_\infty\,.
\end{align*}
As we do not have a good bound on $\left\|\vphi\right\|_\infty$,
we instead use an alternative upper bound on the error:
\begin{align*}
\sqrt{\mathcal{E}_{\rr}\left(\tvpi^C\left(\BB^\top \frac{\qq}{\sqrt{\rr}}\right) - \vpi^C\left(\BB^\top \frac{\qq}{\sqrt{\rr}}\right)\right)} \sqrt{E_{\rr}\left(\vphi\right)}\,,
\end{align*}
where $\mathcal{E}_{\rr}(\cdot)$ gives the energy to route a demand with resistances $\rr$,
and $E_{\rr}(\cdot)$ gives the energy of a potential embedding with resistances $\rr$.
As the standard interior point method step
satisfies $E_{\rr}(\vphi) \leq 1$, all our efforts focus on ensuring that
\begin{align}
\sqrt{\mathcal{E}_{\rr}\left(\tvpi^C\left(\BB^\top \frac{\qq}{\sqrt{\rr}}\right) - \vpi^C\left(\BB^\top \frac{\qq}{\sqrt{\rr}}\right)\right)} \leq \eps
\label{eq:energy_error}
\end{align}
for some error parameter $\eps$.
One issue is the fact that the energy depends on the current resistances, 
therefore even if at some point the error of the demand projection is low,
after a few iterations it might increase because of resistance changes.
We deal with this issue by taking the stability of resistances along the central path
into account.
This allows us to upper bound how much this error increases after a number of iterations. 
The resistance stability lemma is a generalization of the one used in~\cite{gao2021fully}.

Unfortunately, even though (\ref{eq:energy_error}) seems like the right type of guarantee,
it is unclear how to ensure that it is always true.
Specifically, it involves efficiently computing the hitting probabilities from some vertex
$v$ to $C$ in an appropriate norm, which ends up being non-trivial.
Instead, we show that the following \emph{weaker} error bound can be ensured with high probability:
\begin{align}
\left|\left\langle \tvpi^C\left(\BB^\top \frac{\qq}{\sqrt{\rr}}\right) - \vpi^C\left(\BB^\top \frac{\qq}{\sqrt{\rr}}\right), \vphi\right\rangle\right| \leq \eps\,,
\label{eq:potential_error}
\end{align}
where $\vphi$ is a \emph{fixed} potential vector with $E_{\rr}(\vphi) \leq 1$.
Interestingly, this guarantee is still sufficient for our purposes.

\paragraph{Costs and general demand}
There is a fundamental obstacle to using the approach of~\cite{gao2021fully} once edge costs are introduced. In particular, for the maximum flow
problem, the demand pushed by the electrical flow in each iteration is an $s$-$t$ demand, so---up to scaling---it is always constant.
In minimum cost flow on the other hand, the augmenting flow is a multiple of 
$\cc - \RR^{-1} \BB \LL^+ \BB^\top \cc$. Here it is not possible to locate a sublinear number of congested edges just by looking at
the electrical flow term $\RR^{-1} \BB \LL^+ \BB^\top \cc$, as there might be significant cancellations with $\cc$. 
We instead use the following equivalent form:
$\frac{\frac{1}{\ss^+} - \frac{1}{\ss^-}}{\rr} - \RR^{-1} \BB \LL^+ \BB^\top \frac{\frac{1}{\ss^+} - \frac{1}{\ss^-}}{\rr}$,
which allows us to ignore the first term 
because it is small and concentrate on the electrical flow term.
One issue that arises is the fact that the demand vector $\BB^\top \frac{\frac{1}{\ss^+}-\frac{1}{\ss^-}}{\rr}$ now depends on slacks, and as a result changes
throughout the interior point method. This issue can be handled relatively easily. 

A more significant issue concerns the vertex sparsifier. In fact, the vertex sparsifier framework around which~\cite{gao2021fully} is based only 
accepts demands that are supported on the vertex set $C$ of the sparsifier. As $|C|$ is sublinear in $n$, this only captures demands with sublinear support,
one such example being max flow with support $2$. However, our demand vector $\BB^\top \frac{\frac{1}{\ss^+} - \frac{1}{\ss^-}}{\rr}$ in general will be supported
on $n$ vertices. Even though it might seem impossible to get around this issue, we show that the special structure of $C$ allows us to 
push the demand to a small number of vertices. More specifically, we show that
if one projects all of the demand onto $C$, the flow induced by this new demand will not differ much from the one with the original demand.
Concretely, given a Laplacian system $\LL \vphi = \dd$, we decompose it into two systems
$\LL \vphi^{(1)} = \vpi^C(\dd)$ and
$\LL \vphi^{(2)} = \dd - \vpi^C(\dd)$, where $\vpi^C(\dd)$ is the projection of $\dd$ onto $C$.
Intuitively, the latter system computes the electrical flow to push all demands to $C$, and the former to serve this $C$-supported demand. 
We show that, as long as $C$ is a \emph{congestion reduction subset} (as it is also the case in~\cite{gao2021fully}),
$\vphi^{(2)}$ has negligible contribution in the electrical flow, thus it can be ignored.
More specifically, in Section~\ref{sec:Fsystem} we prove the following lemma:
\begin{replemma}{lem:non-projected-demand-contrib}
Consider a graph $G(V,E)$ with resistances $\rr$ and Laplacian $\LL$, a $\beta$-congestion reduction subset $C$,
and a demand $\dd = \delta \BB^\top \frac{\qq}{\sqrt{\rr}}$ for some $\delta > 0$ and
$\qq \in[-1,1]^m$.
Then, the potential embedding defined as
\begin{align*}
\vphi = \LL^+ \left(\dd - \vpi^C\left(\dd\right)\right)
\end{align*}
has congestion $\delta \cdot \tO{1/\beta^2}$, i.e. %
$\left\|\frac{\BB \vphi}{\sqrt{\rr}}\right\|_\infty \leq \delta \cdot \tO{1/\beta^2}$. %
\end{replemma}
Now, for computing $\vphi^{(1)}$, we need to get an approximate estimate of $\vpi^C(\dd)$.
Even though the most natural approach would be to try to maintain $\vpi^C(\dd)$ under 
vertex insertions to $C$, this approach has issues related to the fact
that our error guarantee is based
on a \emph{fixed} potential vector. In particular, if we used an estimate of $\vpi^C(\dd)$, then
the potential vector in (\ref{eq:potential_error}) would depend on the randomness of this
estimate, and as a result the high probability guarantee would not work.

Instead, we show that it is not even neccessary to maintain $\vpi^C(\dd)$ very accurately.
In fact, it suffices to \emph{exactly} compute it only every few iterations of the algorithm,
and use this estimate for the calculation. What allows us to do this is the following lemma,
which bounds the change of $\vpi^C(\dd)$ measured in energy, after a sequence of 
vertex insertions and resistance changes.
\begin{replemma}{lem:old_projection_approximate}
Consider a graph $G(V,E)$ with resistances $\rr^0$, $\qq^0\in[-1,1]^m$, a $\beta$-congestion
reduction subset $C^0$, %
and a fixed sequence of updates,
where the $i$-th update $i\in\{0,T-1\}$ is of the following form:
\begin{itemize}
    \item {\textsc{AddTerminal}($v^i$): Set $C^{i+1} = C^{i} \cup \{v^i\}$ for some $v^i\in V\backslash C^i$, $q_e^{i+1} = q_e^{i}, r_e^{i+1} = r_e^{i}$}
    \item {\textsc{Update}($e^i,\qq,\rr$): Set $C^{i+1} = C^{i}$, $q_e^{i+1} = q_e$ $r_e^{i+1} = r_e$, where $e^i\in E(C^{i})$}
\end{itemize}
Then, with high probability,
\[
\sqrt{\mathcal{E}_{\rr^T} \left(\vpi^{C^0,\rr^0}\left(\BB^\top \frac{\qq_S^0}{\sqrt{\rr^0}}\right) - \vpi^{C^T,\rr^T}\left(\BB^\top \frac{\qq_S^T}{\sqrt{\rr^T}}\right)\right)}
\leq \tO{\max_{i\in\{0,\dots,T-1\}} \left\|\frac{\rr^T}{\rr^i}\right\|_\infty^{1/2} \beta^{-2}} \cdot T\,.
\]
\end{replemma}

If we call this demand projection estimate $\vpi_{old}$, 
the quantity that we would like to maintain (\ref{eq:identity})
now becomes
\begin{align*}
\left\langle\qq,\vrho\right\rangle \approx \left\langle \vpi^C\left(\BB^\top \frac{\qq}{\sqrt{\rr}}\right), SC^+ \vpi_{old}\right\rangle \,.
\end{align*}
Therefore all that's left is to efficiently maintain approximations to demand projections of the form
$\vpi^C\left(\BB^\top \frac{\qq}{\sqrt{\rr}}\right)$.

\paragraph{Bounding demand projections.}
An important component for showing that demand projections can be updated efficiently
is bounding the magnitude of an entry 
$\pi_v^{C\cup\{v\}}(\BB^\top \frac{\onev_{e}}{\sqrt{r_e}})$ 
of the projection, for some fixed edge $e=(u,w)$.
This is apparent in the following identity which shows how a demand 
projection changes after inserting a vertex:
 \begin{align}
 \vpi^{C\cup\{v\}}\left(\dd\right) = \vpi^C\left(\dd\right) + \pi_v^{C\cup\{v\}}\left(\dd\right)\cdot \left(\onev_v - \vpi^C(\onev_v)\right)\,.
 \label{eq:addonevertex}
 \end{align}
In~\cite{gao2021fully} 
this projection entry is upper bounded by
$(p_v^{C\cup\{v\}}(u)+p_v^{C\cup\{v\}}(w)) \cdot \frac{1}{\sqrt{r_e}}$, where
$p_v^{C\cup\{v\}}(u)$ is the probability that a random walk starting at $u$ hits 
$v$ before $C$.
This bound can be very bad as $r_e$ can be arbitrarily small, although in the particular
case of max flow it is possible to show that such low-resistance edges 
cannot get congested and thus are not of interest.

In order to overcome this issue, we provide a different bound, which
in contrast works best when $r_e$ is small.

\begin{replemma}{st_projection1}
Consider a graph $G(V,E)$ with resistances $\rr$ and a subset of vertices $C\subseteq V$.
For any vertex $v\in V\backslash C$
we have that
\begin{align*}
\left|\pi_v^{C\cup\{v\}}\left(\BB^\top \frac{\onev_e}{\sqrt{\rr}}\right)\right| \leq (p_v^{C\cup\{v\}}(u) + p_v^{C\cup\{v\}}(w)) \cdot \frac{\sqrt{r_e}}{R_{eff}(v,e)} \,.
\end{align*}
\end{replemma}

Here $R_{eff}(v,e)$ is the effective resistance between $v$ and $e$.
In fact, together with the other upper bound mentioned above, this implies that 
\[
\left|\pi_v^{C\cup\{v\}}\left(\BB^\top \frac{\onev_e}{\sqrt{\rr}}\right)\right| \leq (p_v^{C\cup\{v\}}(u) + p_v^{C\cup\{v\}}(w)) \cdot \frac{1}{\sqrt{R_{eff}(v,e)}} \,,
\] 
which no longer depends on the value of the resistance $r_e$.

As we will see, it suffices to approximate
$\pi_v^{C\cup\{v\}}\left(\BB^\top \frac{\onev_e}{\sqrt{\rr}}\right)$ up
to additive accuracy roughly 
$\heps \cdot (p_v^{C\cup\{v\}}(u) + p_v^{C\cup\{v\}}(w)) / \sqrt{R_{eff}(C,v)}$
for some error parameter $\heps > 0$.
Thus, Lemma~\ref{st_projection1} immediately implies that for any edge $e$ such that
$R_{eff}(v,e) \gg R_{eff}(C,v)$, this term is small enough to begin with,
and thus can be ignored.

\paragraph{Important edges.}

In order to ensure that the demand projection can be updated efficiently, we focus
only on the demand coming from a special set of edges, which we call
\emph{important}. These are the edges
that are close (in effective resistance metric) to $C$ relative to their own resistance $r_e$.
In fact, the farther an edge is from $C$ in this sense, 
the smaller its worst-case congestion, and so non-important edges
do not influence the set of congested edges that we are looking for.
At a high level, %
this is because parts of the graph that are very far 
in the potential embedding have minimal interactions with each other.

\begin{definition}[Important edges]
An edge $e\in E$ is called $\eps$-\emph{important} (or just \emph{important}) if 
$R_{eff}(C,e) \leq r_e / \eps^2$.
\end{definition}
Based on the above discussion, we seek to find congested edges \emph{only} among important edges.
\begin{replemma}{lem:important_edges}[Localization lemma]
Let $\vphi^*$ be any solution of
\begin{align*}
\LL \vphi^* = \delta \cdot \vpi^C\left(\BB^\top \frac{\pp}{\sqrt{\rr}} \right)\,,
\end{align*}
where $\rr$ are any resistances, $\pp\in[-1,1]^m$, and $C\subseteq V$.
Then,
for any $e\in E$ that is not $\eps$-important
we have
$\left|\frac{\BB\vphi^*}{\sqrt{\rr}}\right|_e \leq 12\eps$.
\end{replemma}
One issue is that the set of important edges changes whenever $C$ changes. However, we show that,
because of the stability of resistances along the central path, the set of important edges only
needs to be updated once every few iterations.

\section{Preliminaries}
\subsection{General}
For any $k\in \mathbb{Z}_{\geq 0}$, we denote
$[k] = \{1,2,\dots,k\}$.

For any $x\in\mathbb{R}^n$ and $C\subseteq [n]$, we 
denote by $x_C\in\mathbb{R}^{|C|}$ the restriction of $x$
to the entries in $C$.
Similarly for a matrix $A$, subset of rows $C$, and subset of columns $F$,
we denote by $A_{CF}$ the submatrix that results
from keeping the rows in $C$ and the columns in $F$.

When not ambiguous, we use the corresponding uppercase symbol 
to a symbol denoting a vector,
to denote the diagonal matrix of that vector. In other
words $\RR = \mathrm{diag}(\rr)$.

Given $x,y\in\mathbb{R}$ and $\alpha\in\mathbb{R}_{\geq 1}$, 
we say that $x$ and $y$ $\alpha$-approximate each other 
and write $x \approx_{\alpha} y$
if
$\alpha^{-1} \leq x/y \leq \alpha$.

When a graph $G(V,E)$ is clear from context, we will
use $n = |V|$ and $m=|E|$.
We use $\BB\in\mathbb{R}^{m\times n}$ to denote the 
edge-vertex incidence matrix of $G$ and, given some resistances
$\rr$, we use $\LL\in\mathbb{R}^{n\times n}$
to denote the Laplacian $\LL = \BB^\top \RR^{-1} \BB$.

\subsection{Minimum cost flow}

Given a directed graph $G(V,E)$ with costs $\cc\in\mathbb{R}^m$, demands $\dd\in\mathbb{R}^n$, and capacities 
$\uu\in\mathbb{R}_{> 0}^m$, 
the \emph{minimum cost flow problem} asks to compute a flow $\ff$ that
\begin{itemize}
\item{routes the demand: $\BB^\top \ff = \dd$}
\item{respects the capacities: $\zerov \leq \ff \leq \uu$, and}
\item{minimizes the cost: $\langle \cc, \ff\rangle$.}
\end{itemize}
We will denote such an instance of the minimum cost flow by the tuple 
$(G(V,E),\cc,\dd,\uu)$.

\subsection{Electrical flows}
\begin{definition}[Energy of a potential embedding]
Consider a graph $G(V,E)$ with resistances $\rr$
and a potential embedding $\vphi$. We denote
by 
\[ E_{\rr}(\vphi) = \sum\limits_{e\in E} \frac{(B\vphi)_e^2}{r_e} \] 
the total energy of the electrical flow induced by $\vphi$. %
\end{definition}
\begin{definition}[Energy to route a demand]
Consider a graph $G(V,E)$ with resistances $\rr$, and a vector $\dd\in\mathbb{R}^n$. %
If $\dd$ is a demand ($\langle 1,d\rangle = 0$), we denote
by 
\[ \mathcal{E}_{\rr}(\dd) = \underset{\vphi:\, B^\top \frac{B\vphi}{\rr} = \dd}{\min}\, E_{\rr}(\vphi) \]
the total energy %
that is required to route the demand $\dd$ with resistances $\rr$.
We extend this definition for a $\dd$ that is not a demand vector ($\langle\onev, \dd\rangle \neq 0$), as
$\mathcal{E}_{\rr}(\dd) = \mathcal{E}_{\rr}\left(\dd - \frac{\langle \onev, \dd\rangle}{n} \cdot \onev\right)$. 
\end{definition}
\begin{fact}[Energy statements]
Consider a graph $G(V,E)$ with resistances $\rr$.
\begin{itemize}
\item {For any $x,y\in \mathbb{R}^n$, we have $\sqrt{\mathcal{E}_{\rr}(x+y)}
\leq \sqrt{\mathcal{E}_{\rr}(x)}
+ \sqrt{\mathcal{E}_{\rr}(y)}
$.}
\item{
and a vector $\dd\in\mathbb{R}^n$.
Then,
\[ \underset{\phi: E_{\rr}(\vphi) \leq 1}{\max}\, \langle \dd, \vphi\rangle = 
\sqrt{\mathcal{E}_{\rr}(\dd)}\,. \]}
\item{
For any resistances $\rr' \leq \alpha \rr$ for some $\alpha \geq 1$
and any $\dd\in\mathbb{R}^n$, we have
$\mathcal{E}_{\rr'}(\dd) \leq \alpha \cdot \mathcal{E}_{\rr}(\dd)$
}
\end{itemize}
\end{fact}
\begin{proof}
We let $\LL = \BB^\top \RR^{-1} \BB$ be the Laplacian of $G$ with resistances $\rr$.
For the first one,
we have
$\sqrt{\mathcal{E}_{\rr}(x+y)}
= \left\|x+y\right\|_{\LL^+}
\leq \left\|x\right\|_{\LL^+}+\left\|y\right\|_{\LL^+}
= \sqrt{\mathcal{E}_{\rr}(x)}
+ \sqrt{\mathcal{E}_{\rr}(y)}
$,
where we used the triangle inequality.

The second one
follows since 
$E_{\rr}(\vphi) = \left\|\vphi\right\|_{\LL}^2$
and $\mathcal{E}_{\rr}(\dd) = \left\|\dd\right\|_{\LL^+}^2$
and the norms 
$\left\|\cdot\right\|_{\LL}$ and
$\left\|\cdot\right\|_{\LL^+}$ are dual.

For the third one, we note that
$\left(\BB^\top \RR'^{-1} \BB\right)^+ \preceq \alpha \cdot \left(\BB^\top \RR^{-1} \BB\right)^+$,
and so 
$\mathcal{E}_{\rr'}(\dd)
 = \left\|\dd\right\|_{(\BB^\top \RR'^{-1} \BB)^+}^2
 \leq \alpha \left\|\dd\right\|_{(\BB^\top \RR^{-1} \BB)^+}^2
 = \alpha \cdot \mathcal{E}_{\rr}(\dd)$.
\end{proof}

\begin{definition}[Effective resistances]
Consider a graph $G(V,E)$ with resistances $\rr$
and any pair of vertices $u,v\in V$. We denote by
$R_{eff}(u,v)$ the energy required to route $1$ unit of flow
from $u$ to $v$, 
i.e. $R_{eff}(u,v) = \mathcal{E}_{\rr}\left(\onev_u - \onev_v\right)$.
This is called the \emph{effective resistance between $u$ and $v$}.
We extend this definition to work with vertex subsets $X,Y\subseteq V$, such that
$R_{eff}(X,Y)$ is the effective resistance between the vertices $x,y$ that result from
contracting $X$ and $Y$. When used as an argument of $R_{eff}$,
an edge $e=(u,v)\in E$ is treated as the vertex subset $\{u,v\}$.
\end{definition}

\begin{definition}[Schur complement]
Given a graph $G(V,E)$ with Laplacian $\LL\in\mathbb{R}^{n\times n}$ and a vertex subset $C\subseteq V$
as well as $F = V\backslash C$, 
$SC(G,C) := \LL_{CC} - \LL_{CF}\LL_{FF}^{-1}\LL_{FC}$ (or just $SC$) is called
the \emph{Schur complement of $G$ onto $C$}.
\end{definition}

\begin{fact}[Cholesky factorization]
Given a matrix $\LL\in\mathbb{R}^{n\times n}$, a subset $C\subseteq[n]$,
and $F = [n] \backslash C$, we have
\begin{align*}
\LL^+ = 
\begin{pmatrix}
\II & -\LL_{FF}^{-1} \LL_{FC}\\
\zerov & \II
\end{pmatrix}
\begin{pmatrix}
\LL_{FF}^{-1} & \zerov\\
\zerov & SC(\LL,C)^+
\end{pmatrix}
\begin{pmatrix}
\II & \zerov\\
-\LL_{CF} \LL_{FF}^{-1} & \II
\end{pmatrix}\,.
\end{align*}
\end{fact}

\subsection{Random walks}

\begin{definition}[Hitting probabilities]
Consider a graph $G(V,E)$ with resistances $\rr$. For any 
$u,v\in V$, $C\subseteq V$, we denote
by $p_v^{C,\rr}(u)$ the probability that for random walk that 
starts from $u$ and uses edges with probability proportional to $\frac{1}{\rr}$,
the first vertex of $C$ to be visited is $v$.
When not ambiguous, we will use the notation $p_v^C(u)$.
\label{def:hitting_probabilities}
\end{definition}

\begin{definition}[Demand projection]
Consider a graph $G(V,E)$ and a demand vector $\dd$. 
For any $v\in V$, $C\subseteq V$, we 
define $\pi_v^{C,\rr}(\dd) = \sum\limits_{u\in V} d_u p_v^{C,\rr}(u)$ and call 
the resulting vector $\vpi^{C,\rr}(\dd)\in\mathbb{R}^{n}$ the 
\emph{demand projection of $\dd$ onto $C$}.
When not ambiguous, we will use the notation $\vpi^C(\dd)$.
\label{def:demand_projection}
\end{definition}

For convenience, when we write $\vpi^C(\dd)$ we might also 
refer to the restriction of this vector to $C$. This will be
clear from the context, and, as $\pi_v^C(\dd) = 0$ for any $v\notin C$,
no ambiguity is introduced.

\begin{fact}[\cite{gao2021fully}]
Given a graph $G(V,E)$ with Laplacian $\LL$, 
a vertex subset $C\subseteq V$, and
$\dd\in\mathbb{R}^n$, we have
\[ \vpi^C(\dd) = \dd_{C} - \LL_{CF}\LL_{FF}^{-1} \dd_{F}\\
\,. \]
Additionally,
\[ \left[\LL^+\dd\right]_C = 
 SC^+\vpi^C(\dd) \,,\]
 where $SC$ is the Schur complement of $G$ onto $C$.
 \label{fact:demand_proj}
\end{fact}

An important property of the demand projection is that the energy required to route it is upper bounded by the energy required to route the original demand. The proof can be found in Section~\ref{sec:aux}.
\begin{lemma}\label{lem:sc-energy-bd}
Let $\dd$ be a demand vector,  let $\rr$ be resistances, and let $C \subseteq V$ be a subset of vertices.
Then 
\[
\mathcal{E}_{\rr}\left( \vpi^C(\dd) \right) \leq \mathcal{E}_{\rr}\left( \dd \right) \,.
\]
\end{lemma}

The following lemma relates 
the effective resistance between a vertex and a vertex set,
to the energy to route a particular demand, based on a demand
projection.
\begin{lemma}[Effective resistance and hitting probabilities]
Given a graph $G(V,E)$ with resistances $\rr$, any vertex set $A\subseteq V$ and vertex $u\in V\backslash A$, we have
$R_{eff}(u,A) = \mathcal{E}_{\rr}(\onev_u - \vpi^A(\onev_u))$.
\label{lem:effective_hitting}
\end{lemma}
\begin{proof}
Let $\LL$ be the Laplacian of $G$ with resistances $\rr$ and $F = V\backslash A$.
We first prove that 
\[ \mathcal{E}_{\rr}(\onev_u - \vpi^A(\onev_u)) = \onev_u^\top \LL_{FF}^{-1}\onev_u \,.\]
This is because
\begin{align*}
& \mathcal{E}_{\rr}(\onev_u - \vpi^A(\onev_u)) \\
& = \langle \onev_u - \vpi^A(\onev_u), \LL^+ (\onev_u - \vpi^A(\onev_u)) \rangle\\
& = \bigg\langle \onev_u - \begin{pmatrix}\zerov & \zerov\\-\LL_{AF} \LL_{FF}^{-1} & \II\end{pmatrix} \onev_u, 
\LL^+ \left(\onev_u - \begin{pmatrix}\zerov & \zerov \\ -\LL_{AF}\LL_{FF}^{-1} & \II \end{pmatrix} \onev_u\right) \bigg\rangle\\
& = \bigg\langle\begin{pmatrix}\II & \zerov\\-\LL_{AF}\LL_{FF}^{-1} & \II\end{pmatrix} \left(\onev_u + \begin{pmatrix}\zerov\\\LL_{AF} \LL_{FF}^{-1} \onev_u\end{pmatrix}\right),\\
&\quad\quad\,\begin{pmatrix}\LL_{FF}^{-1} & \zerov \\ \zerov & SC(\LL,A)^+\end{pmatrix}
\begin{pmatrix}\II & \zerov\\-\LL_{AF}\LL_{FF}^{-1} & \II\end{pmatrix} \left(\onev_u + \begin{pmatrix}\zerov\\\LL_{AF}\LL_{FF}^{-1} \onev_u \end{pmatrix}\right) \bigg\rangle\\
& = \left\langle\onev_u,
\begin{pmatrix}\LL_{FF}^{-1} & \zerov \\ \zerov & SC(\LL,A)^+\end{pmatrix} \onev_u\right\rangle\\
&=\langle \onev_u, \LL_{FF}^{-1} \onev_u\rangle \,.
\end{align*}
On the other hand, note that $R_{eff}(u,A) = \hR_{eff}(u,\ha)$, where $\hR$ are the effective resistances in a graph $\hG$ that results after contracting
$A$ to a new vertex $\ha$. It is easy to see that the Laplacian of this new graph is
\begin{align*}
\widehat{\LL} = \begin{pmatrix}
\LL_{FF} & \LL_{FA} \onev \\
\onev^\top \LL_{AF} & \onev^\top \LL_{FA} \onev
\end{pmatrix}\,.
\end{align*}
We look at the system 
$\widehat{\LL} \begin{pmatrix}x\\a\end{pmatrix} = \onev_u - \onev_{\ha}$, where $a$ is a scalar. The solution is given by 
\begin{align*}
& \xx = \LL_{FF}^{-1} \left(\onev_u - a\cdot \LL_{FA} \onev\right)\,.
\end{align*}
However, as $\onev\in \ker(\widehat{\LL})$ by the fact that it is a Laplacian, we can assume that $a=0$ by shifting.
Therefore $\xx = \LL_{FF}^{-1} \onev_u$, and so we can conclude that
\begin{align*}
R_{eff}(u,A) 
& = \langle\onev_u - \onev_{\ha}, \widehat{\LL}^+ (\onev_u - \onev_{\ha})\rangle\\
& = \langle\onev_u, \LL_{FF}^{-1} \onev_u\rangle \,.
\end{align*}
So we have proved that $R_{eff}(u,A) = \mathcal{E}_{\rr}(\onev_u - \vpi^A(\onev_u))$ and we are done.
\end{proof}

Finally, the following lemma relates
the effective resistance between a vertex and a vertex set, to
the effective resistance between vertices.
\begin{lemma}
Given a graph $G(V,E)$ with resistances $\rr$, any vertex set $A\subseteq V$ and vertex $u\in V\backslash A$, we have
\[ \frac{1}{|A|} \cdot \underset{v\in A}{\min}\, R_{eff}(u,v)
\leq R_{eff}(u,A) \leq \underset{v\in A}{\min}\, R_{eff}(u,v)\,.\]
\label{lem:multi_effective_resistance}
\end{lemma}
\begin{proof}
Let $\LL$ be the Laplacian of $G$ with resistances $\rr$, and note that $R_{eff}(u,v) = \left\|\LL^{+/2} (\onev_u - \onev_v)\right\|_2^2$
and, by Lemma~\ref{lem:effective_hitting},
$R_{eff}(u,A) = \left\|\LL^{+/2} (\onev_u - \vpi^A(\onev_u))\right\|_2^2$.
Expanding the latter, we have
\begin{align*}
R_{eff}(u,A) 
& = 
\left\|\sum\limits_{v\in A} \pi_v^A(\onev_u)\cdot\LL^{+/2} (\onev_u -  \onev_v)\right\|_2^2\\
& = \sum\limits_{v\in A}\left(\pi_v^A(\onev_u)\right)^2 \left\| \LL^{+/2} (\onev_u - \onev_v)\right\|_2^2
+ \sum\limits_{v\in A}
 \sum\limits_{\substack{v'\in A\\v'\neq v}} \pi_v^A(\onev_u) \pi_{v'}^A(\onev_u) \langle \onev_u - \onev_{v'}, \LL^{+} (\onev_u - \onev_v)\rangle\,.
\end{align*}
Now, note that $\pi_v^A(\onev_u),\pi_{v'}^A(\onev_u) \geq 0$. Additionally, let 
$\vphi = \LL^{+} (\onev_u - \onev_v)$ be the potential embedding that induces a $1$-unit electrical flow from $v$ to $u$.
As the potential embedding stretches between $\phi_v$ and $\phi_u$, we have that
$\phi_{v'} \leq \phi_u$, so $\langle \onev_u - \onev_{v'}, \LL^+ (\onev_u - \onev_v)\rangle = \phi_u- \phi_{v'} \geq 0$. Therefore,
\begin{align*}
R_{eff}(u,A) 
& \geq
\sum\limits_{v\in A}\left(\pi_v^A(\onev_u)\right)^2 \left\| \LL^{+/2} (\onev_u - \onev_v)\right\|_2^2\\
& \geq \frac{1}{|A|} \sum\limits_{v\in A} \pi_v^A(\onev_u) \cdot R_{eff}(u,v)\\
& \geq \frac{1}{|A|}\underset{v\in A}{\min}\, R_{eff}(u,v)\,,
\end{align*}
where we used the Cauchy-Schwarz inequality and the fact that $\sum\limits_{v\in A} \pi_v^A(\onev_u) = 1$.
\end{proof}

\section{Interior Point Method with Dynamic Data Structures}
\label{sec:ipm}
The goal of this section is to show that, given a data structure for approximating
electrical flows in sublinear time,
we can execute the min cost flow interior point method with total runtime faster than
$\tO{m^{3/2}}$.

\subsection{LP formulation and background}
We present the interior point method setup that we will use, which is from~\cite{axiotis2020circulation}.
Our goal is to solve the following minimum cost flow
linear program:
\begin{align*}
& \min\, \left\langle \cc, \CC\xx\right\rangle\\
& \zerov \leq \ff^0 + \CC \xx \leq \uu\,,
\end{align*}
where $\ff^0$ is a flow with $\BB^\top \ff^0= \dd$ and
$\CC$ is an $m\times (m-n+1)$ matrix whose image is the set of circulations in $G$.

In order to use an interior point method,
the following log barrier objective is defined:
\begin{align}
\underset{\xx}{\min}\, F_{\mu}(\xx) = \left\langle \frac{\cc}{\mu}, \CC\xx \right\rangle -\sum\limits_{e\in E}\left(\log\left(\ff^0 + \CC\xx\right)_e
+ \log\left(\uu - (\ff^0 + \CC\xx)\right)_e
\right)\,.
\label{eq:logbarrier}
\end{align}
For any parameter $\mu > 0$,
the optimality condition of (\ref{eq:logbarrier}) is called the \emph{centrality condition} 
and is given by
\begin{align}
\CC^\top \left(\frac{\cc}{\mu} + \frac{1}{\ss^+} - \frac{1}{\ss^-}\right) = \zerov\,,
\label{eq:centrality}
\end{align}
where $\ff = \ff^0 + \CC\xx$, and
$\ss^+ = \uu - \ff$, $\ss^- = \ff$ are called the \emph{positive} and \emph{negative} slacks of $\ff$
respectively.

This leads us to the following definitions.
\begin{definition}[$\mu$-central flow]
Given a minimum cost flow instance with costs $\cc$, demands $\dd$ and capacities $\uu$,
as well as a parameter $\mu > 0$, 
we will say that a flow $\ff$ (and its corresponding slacks $\ss$ and resistances $\rr$)
is $\mu$-central if 
$\BB^\top \ff = \dd$, $\ss > \zerov$, and
it satisfies the centrality condition (\ref{eq:centrality}), i.e.
\begin{align*}
\CC^\top \left(\frac{\cc}{\mu} + \frac{1}{\ss^+} - \frac{1}{\ss^-}\right) = \zerov\,.
\end{align*}
Additionally, we will denote such flow by $\ff(\mu)$ (and its corresponding slacks 
and resistances by $\ss(\mu)$ and $\rr(\mu)$, respectively).
\end{definition}
\begin{definition}[$(\mu,\alpha)$-central flow]
Given parameters $\mu > 0$ and $\alpha\geq 1$,
we will say that a flow $\ff$ with resistances $\rr > \zerov$ 
is $(\mu,\alpha)$-central
if $\rr \approx_{\alpha} \rr(\mu)$. We will also call its corresponding slacks $\ss$ and
resistances $\rr$ $(\mu,\alpha)$-central.
\end{definition}

Given a $\mu$-central flow $\ff$
and some step size $\delta > 0$, the standard (Newton) step 
to obtain an approximately $\mu/(1+\delta)$-central flow $\ff' = \ff + \tff$
is given by
\begin{align*}
\tff = 
&
-\frac{\delta}{\mu} \frac{\cc}{\rr} + \frac{\delta}{\mu} \RR^{-1}\BB\LL^+ \BB^\top \frac{\cc}{\rr}\\
& = 
\delta \frac{\frac{1}{\ss^+} - \frac{1}{\ss^-}}{\rr} 
- \delta \RR^{-1}\BB\LL^+ \BB^\top \frac{\frac{1}{\ss^+} - \frac{1}{\ss^-}}{\rr}\\
& = 
\delta \cdot g(\ss) 
- \delta \RR^{-1}\BB\LL^+ \BB^\top g(\ss)
\end{align*}
where $\rr = \frac{1}{(\ss^+)^2} + \frac{1}{(\ss^-)^2}$ and 
we have denoted $g(\ss) = \frac{\frac{1}{\ss^+} - \frac{1}{\ss^-}}{\rr}$.

\begin{fact}
Using known scaling arguments, can assume that costs and capacities
are bounded by $\mathrm{poly}(m)$, while only incurring an extra logarithmic dependence in the largest network parameter~\cite{gabow1983scaling}.
\end{fact}

We also use the fact that the resistances in the interior point
method are never too large, which is proved in
Appendix~\ref{sec:aux}.
\begin{fact}
For any $\mu\in(1/\mathrm{poly}(m),\mathrm{poly}(m))$, we have 
$\left\|\rr(\mu)\right\|_\infty \leq m^{\tO{\log m}}$.
\end{fact}

\subsection{Making progress with approximate electrical flows}

The following lemma shows that we can make $k$ 
steps of the interior point method by 
computing $O(k^4)$ $(1+O(k^{-6}))$-\emph{approximate} electrical flows.
The proof is essentially the same as in~\cite{gao2021fully}, but we provide it for
completeness in Appendix~\ref{sec:proof_lem_approx_central}.

\begin{lemma}
Let $\ff^1,\dots,\ff^{T+1}$ be flows with
slacks $\ss^t$ and resistances $\rr^t$ for $t\in[T+1]$, 
where $T = \frac{k}{\e}$ for some $k\leq \sqrt{m}/10$
and $\e = 10^{-5}k^{-3}$, %
such that
\begin{itemize}
\item $\ff^1$ is $(\mu,1+\es/8)$-central for $\es = 10^{-5} k^{-3}$ %
\item {For all $t\in[T]$ and $e\in E$,
$f_e^{t+1} = \begin{cases}
f_e(\mu) + \e \sum\limits_{i=1}^{t} \tf_e^i & \text{if $\exists i\in[t]:\tf_e^{i} \neq 0$}\\
f_e^1 & \text{otherwise}
\end{cases}$, where
\[ 
\tff^{*t} = \delta g(\ss^t) - \delta (\RR^t)^{-1} \BB\left(\BB^\top (\RR^t)^{-1} \BB\right)^+ \BB^\top g(\ss^t) \] 
for $\delta = \frac{1}{\sqrt{m}}$ and
\[ \left\|\sqrt{\rr^t}\left(\tff^{*t} - \tff^{t}\right) \right\|_\infty \leq \eps \] 
for $\eps = 10^{-6} k^{-6}$.
}
\end{itemize}
Then, setting $\e = \es = 10^{-5} k^{-3}$ and
$\eps = 10^{-6} k^{-6}$
we get that
$\ss^{T+1} \approx_{1.1} \ss\left(\mu/(1+\e\delta)^{k\e^{-1}}\right)$.
\label{lem:approx_central}
\end{lemma}
From now and for the rest of Section~\ref{sec:ipm}
we fix the values of $\e, \es, \eps$ based on this lemma.
Using this lemma together with the following recentering procedure also used in~\cite{gao2021fully}, we can 
exactly compute a $\left({\mu}/{(1+\e/\sqrt{m})^{k\e^{-1}}}\right)$-central flow.
\begin{lemma}
Given a flow $\ff$ with slacks $\ss$
such that $\ss\approx_{1.1} \ss(\mu)$ for some $\mu > 0$,
we can compute $\ff(\mu)$ in $\tO{m}$.
\label{lem:recenter}
\end{lemma}

\subsection{The \textsc{Locator} data structure}
\label{sec:locator_def}

From the previous lemma it becomes obvious that the only thing left is to maintain
in sublinear time an approximation to
\[ \delta g(\ss^t) - \delta (\RR^t)^{-1} \BB(\BB^\top (\RR^t)^{-1} \BB)^+ \BB^\top g(\ss^t)\,. \] 
for $\delta = 1 / \sqrt{m}$.
This is the job of the $(\alpha,\beta,\eps)$-\textsc{Locator}
data structure, which computes all the entries of this vector that have magnitude $\geq \eps$.
We note that the guarantees of this data structure
are similar to the ones in~\cite{gao2021fully}, but our locator requires 
an extra parameter $\alpha$ which is a measure of how much resistances
can deviate before a full recomputation has to be made.

\begin{definition}[$(\alpha,\beta,\eps)$-\textsc{Locator}]
An $(\alpha,\beta,\eps)$-\textsc{Locator} is a data structure that maintains valid slacks $\ss$
and resistances $\rr$, and can support the following operations against oblivious
adversaries with high probability:
\begin{itemize}
\item{$\textsc{Initialize}(\ff)$:
Set $\ss^+ = \uu-\ff$, $\ss^-=\ff$,
$\rr = \frac{1}{(\ss^+)^2} + \frac{1}{(\ss^-)^2}$.
}
\item{$\textsc{Update}(e,\ff)$: 
Set $s_e^+ = u_e-f_e$, $s_e^-=f_e$,
$r_e = \frac{1}{(s_e^+)^2} + \frac{1}{(s_e^-)^2}$.
Works under the condition 
that 
\[ r_e^{\max} / \alpha \leq r_e \leq \alpha \cdot r_e^{\min}\,, \]
where $r_e^{\max}$ and $r_e^{\min}$ are the maximum and minimum
resistance values that edge $e$ has had since the last call to 
$\textsc{BatchUpdate}$.}
\item{$\textsc{BatchUpdate}(Z,\ff)$: 
Set $s_e^+ = u_e - f_e, s_e^- = f_e, r_e = \frac{1}{(s_e^+)^2} + \frac{1}{(s_e^-)^2}$ for all $e \in Z$.

}
\item{$\textsc{Solve}()$:
Let
\begin{align}
\tff^* = \delta g(\ss) - \delta \RR^{-1} \BB(\BB^\top \RR^{-1} \BB)^+ \BB^\top g(\ss)\,,
\label{def:ele}
\end{align}
where $\delta = \frac{1}{\sqrt{m}}$.
Returns an edge set $Z$ of size $\tO{\eps^{-2}}$ that with high probability contains
all $e$ such that $\sqrt{r_e} \left|\tf_e^*\right| \geq \eps$.
}
\end{itemize}
The data structure works as long as
the total number of calls to
$\textsc{Update}$, plus the sum of $|Z|$ for all
calls to $\textsc{BatchUpdate}$ is $O(\beta m)$.
\label{def:locator}
\end{definition}

In Section~\ref{sec:locator} we will prove the following lemma, which constructs an
$(\alpha,\beta,\eps)$-\textsc{Locator} and outlines its runtime guarantees:
\begin{lemma}[Efficient $(\alpha,\beta,\eps)$-\textsc{Locator}]
For any graph $G(V,E)$ and
parameters $\alpha \geq 1$, $\beta \in(0,1)$, $\eps \geq \tOm{\beta^{-2} m^{-1/2}}$,
and 
$\heps \in \left(\tOm{\beta^{-2} m^{-1/2}},\eps\right)$, 
there exists an $(\alpha,\beta,\eps)$-\textsc{Locator} 
for $G$ with the following runtimes per operation:
\begin{itemize}
\item {$\textsc{Initialize}(\ff)$:
$\tO{m \cdot \left(\heps^{-4} \beta^{-8}
+ \heps^{-2} \eps^{-2} \alpha^2 \beta^{-4}\right)}$.}
\item{$\textsc{Update}(e,\ff)$: 
$\tO{m \cdot \frac{\heps \alpha^{1/2}}{\eps^{3}} + 
\heps^{-4} \eps^{-2} \beta^{-8} + 
\heps^{-2} \eps^{-4} \alpha^2 \beta^{-6}}$ amortized.
}
\item{$\textsc{BatchUpdate}(Z,\ff)$: 
$\tO{m \cdot \frac{1}{\eps^2} + |Z|\cdot \frac{1}{\eps^2 \beta^2} }$.
}
\item{$\textsc{Solve}()$:
$\tO{\beta m \cdot \frac{1}{\eps^2}}$.
}
\end{itemize}
\label{lem:locator}
\end{lemma}

Note that even though a $\textsc{Locator}$ computes a set that contains all $\eps$-congested edges, it does not return the actual flow values.
The reason for that is that it only works against oblivious adversaries, and allowing (randomized) flow values to affect future updates
constitutes an adaptive adversary. As in~\cite{gao2021fully}, we resolve this by sanitizing the outputs through a different data structure
called $\textsc{Checker}$, which computes the flow values and works against semi-adaptive adversaries. As 
the definition and implementation of $\textsc{Checker}$ is
orthogonal to our contribution and also does not affect
the final runtime, we defer the discussion
to Appendix~\ref{sec:checker}. To simplify the presentation in this section,
we instead define the following
idealized version of it, called $\textsc{PerfectChecker}$.

\begin{definition}[$\eps$-\textsc{PerfectChecker}]
For any error $\eps > 0$, an $\eps$-$\textsc{PerfectChecker}$ is an oracle that 
given a graph $G(V,E)$, slacks $\ss$, resistances $\rr$, supports the following
operations:
\begin{itemize}
\item{\textsc{Update}$(e,\ff)$: Set $s_e^+ = u_e - f_e$, $s_e^- = f_e$, 
$r_e = \frac{1}{(s_e^+)^2}
+ \frac{1}{(s_e^-)^2}$.
}
\item{\textsc{Check}$(e)$:
Compute a flow value $\tf_e$ such that $\sqrt{r_e} \left|\tf_e - \tf_e^*\right| \leq \eps$, where
\[ \tff^* = \delta g(\ss) - \delta \RR^{-1} \BB (\BB^\top \RR^{-1} \BB)^+ \BB^\top g(\ss) \,, \]
with $\delta = 1/\sqrt{m}$.
If $\sqrt{r_e} \left|\tf_e\right| < \eps / 2$ return $0$, otherwise return $\tf_e$.

}
\end{itemize}
\label{def:perfect_checker}
\end{definition}

\subsection{The minimum cost flow algorithm}
Now, we will show how the data structure defined in
Section~\ref{sec:locator_def} can be used to make progress along the central path. 
The main lemma that analyzes the performance of the minimum cost flow
algorithm given access to an $(\alpha,\beta,\eps)$-\textsc{Locator} is Lemma~\ref{lem:mincostflow}.
Also, the skeleton of the algorithm is described in Algorithm~\ref{alg:main}.

\begin{algorithm} %
\begin{algorithmic}[1]
\caption{Minimum Cost Flow}
\Procedure{\textsc{MinCostFlow}}{$G,\cc,\dd,\uu$}
\State $\off, \mu = \textsc{Initialize}(G,\cc, \dd,\uu)$ \Comment{Lemma~\ref{lem:init}. $\off$ is $\mu$-central at all times.}
\State $i = 0$
\While {$\mu \geq m^{-10}$}
	\If {$i$ is a multiple of $\lfloor\es \sqrt{\beta m} / k\rfloor$} \Comment{Re-initialize when $|C|$ exceeds $O(\beta m)$.}
		\State $\cL = \textsc{Locator}.\textsc{Initialize}(\off)$ with error $\eps/2$
	\EndIf
	\If {$i$ is a multiple of $\lfloor\es \sqrt{\bc m} / k\rfloor$} 
		\State $\cC^i = \textsc{Checker}.\textsc{Initialize}(\off, \eps, \bc)$ for $i\in[k\e^{-1}]$
	\EndIf
	\If {$i$ is a multiple of $\lfloor 0.5\alpha^{1/4} / k - 1\rfloor$} \Comment{Update important edges when $\cL.\rr^0$ expires}
		\State $\cL.\textsc{BatchUpdate}(\emptyset)$ \label{line:batchupdate1}
	\EndIf
	\State $\off, \mu = \textsc{MultiStep}(\off, \mu)$
	\If {$i$ is a multiple of $\hT$}
		\State $Z = \emptyset$
		\For {$e \in E$}
			\State $\os_e^+ = u_e - \of_e$, $\os_e^- = \of_e$
			\If {$\os_e^+ \not\approx_{\es / 16} \cL.s_e^+$ or $\os_e^- \not\approx_{\es / 16} \cL.s_e^-$}
				\State $\cC^i.\textsc{Update}(e, \off)$ for $i\in[k\e^{-1}]$
				\State $Z = Z\cup\{e\}$
			\EndIf
		\EndFor
		\State $\cL.\textsc{BatchUpdate}(Z, \off)$ \label{line:batchupdate2}
	\Else
		\For {$e \in E$}
			\State $\os_e^+ = u_e - \of_e$, $\os_e^- = \of_e$
			\If {$\os_e^+ \not\approx_{\es / 8} \cL.s_e^+$ or $\os_e^- \not\approx_{\es / 8} \cL.s_e^-$}
				\State $\cC^i.\textsc{Update}(e, \off)$ for $i\in[k\e^{-1}]$
				\State $\cL.\textsc{Update}(e, \off)$
			\EndIf
		\EndFor
	\EndIf
	\State $i = i + 1$
\EndWhile
\State \Return $\textsc{Round}(G,\cc,\dd,\uu,\off)$ \Comment{Lemma~\ref{lem:rounding}}
\EndProcedure
\label{alg:main}
\end{algorithmic} 
\end{algorithm}

\begin{lemma}[\textsc{MinCostFlow}]
Let $\cL$ be an $(\alpha,\beta,\eps)-\textsc{Locator}$,
$\ff$ be a $\mu$-central flow
where $\mu = \mathrm{poly}(m)$, and $k\in\left[m^{1/316}\right]$,
$\beta \geq \tOm{k^3 / m^{1/4}}$,
$\hT \in\left[\tO{m^{1/2} / k}\right]$
be some parameters.
There is an algorithm that with high probability computes a $\mu'$-central flow $\ff'$, where
$\mu' \leq m^{-10}$. Additionally, the algorithm runs in time $\tO{m^{3/2} / k}$, plus
\begin{itemize}
\item $\tO{k^3 \beta^{-1/2}}$ calls to $\cL.\textsc{Initialize}$,
\item $\tO{m^{1/2} k^3}$ calls to $\cL.\textsc{Solve}$,
\item $\tO{m^{1/2} \left(k^6\hT + k^{15}\right)}$ calls to $\cL.\textsc{Update}$,
\item $\tO{m^{1/2} \alpha^{-1/4}}$ calls to $\cL.\textsc{BatchUpdate}(\emptyset)$, and
\item {
$\tO{m^{1/2} k^{-1} \hT^{-1}}$ calls to 
$\cL.\textsc{BatchUpdate}(Z,\off)$ for some $Z\neq \emptyset,\off$. Additionally,
the sum of $|Z|$ over all such calls is $\tO{mk^3\beta^{1/2}}$.}
\end{itemize}
\label{lem:mincostflow}
\end{lemma}
The proof appears in Appendix~\ref{proof_lem_mincostflow}.
Its main ingredient is the following lemma, which easily follows from Lemma~\ref{lem:approx_central}
and essentially shows how $k$ steps of the interior point method can be performed in 
$\tO{m}$ instead of $\tO{m k}$. Its proof appears in Appendix~\ref{proof_lem_multistep}.

\begin{lemma}[\textsc{MultiStep}]
Let $k\in\{1,\dots,\sqrt{m}/10\}$. 
We are given $\ff(\mu)$, an $(\alpha,\beta,\eps/2)$-\textsc{Locator} $\cL$, 
and an $\eps$-\textsc{PerfectChecker} $\cC$,
such that
\begin{itemize}
\item{
$\cL.\rr = \cC.\rr$ are $(\mu,1+\es/8)$-central resistances, and 
}
\item{
$\cL.\rr^0$ are
$(\mu^0,1+\es/8)$-central resistances, where
$\mu^0 \leq \mu \cdot (1 + \e/\sqrt{m})^{\hT}$ and $\hT = (0.5\alpha^{1/4} - k)\e^{-1}$.
Additionally, for any resistances $\hrr$ that $\cL$ had at any point 
since the last call to $\cL.\textsc{BatchUpdate}$, 
$\hrr$ are $(\hmu, 1.1)$-central for some $\hmu\in[\mu,\mu^0]$.
}
\end{itemize}

Then, %
there is an algorithm that with high probability computes $\ff(\mu')$, %
where $\mu' = \mu / (1+\e/\sqrt{m})^{k\e^{-1}}$. The algorithm runs in time $\tO{m}$,
plus $O(k^{16})$ calls to $\cL.\textsc{Update}$, $O(k^4)$ calls to $\cL.\textsc{Solve}$,
and $O(k^{16})$ calls to $\cC.\textsc{Update}$ and $\cC.\textsc{Check}$.
Additionally, $\cL.\rr$ and $\cC$ are unmodified.
\label{lem:multistep}
\end{lemma}

\begin{algorithm} %
\begin{algorithmic}[1]
\caption{MultiStep}
\Procedure{\textsc{MultiStep}}{$\ff, \mu$}
\Comment{Makes equivalent progress to $k$ interior point method steps}
\State $\hrr = \cL.\rr$ \Comment{Save resistances to restore later}
\For{$i=1,\dots, k\e^{-1}$}
	\State $Z = \cL.\textsc{Solve}()$
	\For {$e\in Z$} \Comment{$Z$: Set of edges with sufficiently changed flow}
		\State $\tf_e = \cC^i.\textsc{Check}(e)$
		\If {$\tf_e \neq 0$}
			\State $f_e = f_e + \e \tf_e$
			\State $\cL.\textsc{Update}(e, \ff)$
			\State $\cC^j.\textsc{TemporaryUpdate}(e, \ff)$ for $j \in \left[i+1,k\e^{-1}\right]$
		\EndIf
	\EndFor
\EndFor
\State $\mu = \mu / (1 + \e/\sqrt{m})^{k\e^{-1}}$
\State $\ff = \textsc{Recenter}(\ff, \mu)$ \Comment{Lemma~\ref{lem:recenter}}
\For {$e\in E$}
	\If {$\cL.r_e \neq \hr_e$}
		\State $\cL.\textsc{Update}(e,\hrr)$ \Comment{Return $\textsc{Locator}$ resistances to their original state}
	\EndIf
\EndFor
\State Call $\cC^i.\textsc{Rollback}()$ to undo all $\textsc{TemporaryUpdate}$s for all $\cC^i$
\State \Return $\ff, \mu$
\EndProcedure
\label{alg:multistep}
\end{algorithmic} 
\end{algorithm}

\subsection{Proof of Theorem~\ref{thm:main}}

\paragraph{Correctness.}
First of all, we apply capacity and cost scaling~\cite{gabow1983scaling} to make sure that 
$\left\|\cc\right\|_\infty, \left\|\uu\right\|_\infty = \mathrm{poly}(m)$. These incur
an extra factor of $\log(U+W)$ in the runtime.

We first get an initial solution to the interior point method by using the following lemma:
\begin{lemma}[Interior point method initialization, Appendix A in~\cite{axiotis2020circulation}]
Given a min cost flow instance $\mathcal{I} = \left(G(V,E),\cc,\dd,\uu\right)$,
there exists an algorithm that 
runs in time $O(m)$ and produces a new min cost flow instance
$\mathcal{I}'=\left(G'(V',E'),\cc',\dd',\uu'\right)$, 
where $|V'| = O(|V|)$ and $|E'|=O(|E|)$,
as well as a flow $\ff$ such that
\begin{itemize}
\item{$\ff$ is $\mu$-central for $\mathcal{I}'$ for some 
$\mu = \Theta\left(\left\|\cc\right\|_2\right)$}
\item{Given an optimal solution for $\mathcal{I}'$, 
an optimal minimum cost flow solution for $\mathcal{I}$
can be computed in $O(m)$}
\end{itemize}
\label{lem:init}
\end{lemma}

Therefore we now have a $\mathrm{poly}(m)$-central solution for an instance $\mathcal{I}$.
We can now apply Lemma~\ref{lem:mincostflow} to get a $\mu'$-central solution with $\mu'\leq m^{-10}$.
Then we can apply the following lemma to round the solution, which follows from Lemma 5.4 in~\cite{axiotis2020circulation}.
\begin{lemma}[Interior point method rounding]
Given a min cost flow instance $\mathcal{I}$ and a $\mu$-central flow $\ff$
for $\mu \leq m^{-10}$, there is an algorithm that runs in time $\widetilde{O}(m)$ and returns an optimal integral flow.
\label{lem:rounding}
\end{lemma}
By Lemma~\ref{lem:init}, this solution can be turned into an exact solution for the original
instance.
As Lemma~\ref{lem:mincostflow} succeeds with high probability, the whole algorithm does too.

\paragraph{Runtime.}

To determine the final runtime, we analyze each operation in Algorithm~\ref{alg:main} separately.

The \textsc{Initialize} (Lemma~\ref{lem:init}) and \textsc{Round} (Lemma~\ref{lem:rounding}) 
operations take time $\tO{m}$.
Now, the runtime of Lemma~\ref{lem:mincostflow} is $\tO{m^{3/2} / k}$ plus the runtime incurred because of
calls to the locator $\cL$. 
We will use the runtimes per operation from Lemma~\ref{lem:locator}.

{\bf $\cL$.\textsc{Solve}}: This operation is run $\tO{m^{1/2} k^3}$ times, and each of these costs
$\tO{\frac{\beta m}{\eps^2}}
= \tO{m k^{12} \beta}$. Therefore in total
$\tO{m^{3/2} k^{15} \beta}$.

We pick $\beta$ by 
$m^{3/2} k^{15} \beta \leq m^{3/2} /k$ as 
$\beta = k^{-16}$, so the runtime is
$\tO{m^{3/2} / k}$. Note that this satisfies the constraint
$\beta \geq \tOm{k^3 / m^{1/4}}$ as long as $k\leq \tOm{m^{1/76}}$.

{\bf $\cL$.\textsc{BatchUpdate}}: This is run
$\tO{m^{1/2} / \alpha^{1/4}}$ times with empty arguments,
each of which takes time 
$\tO{m /\eps^2} = \tO{m k^{12}}$.
The total runtime because of these is 
$\tO{m^{3/2} k^{12} \alpha^{-1/4}}$.
As we 
need this to be below $\tO{m^{3/2} / k}$, we set 
$\alpha = k^{52}$.

This operation is also run
$\tO{m^{1/2} k^{-1} \hT^{-1}}
$
times with 
some non-empty argument $Z$, each of which takes time 
$\tO{m/\eps^2 + |Z|/(\eps^2\beta^2)} = \tO{m k^{12} + k^{44}|Z|}$.
As by Lemma~\ref{lem:mincostflow} the total sum
of $|Z|$ %
over all calls is $\tO{mk^3 \beta^{1/2}} = \tO{mk^{-5}}$,
we get a runtime of
\[ \tO{m^{1/2} k^{-1} \hT^{-1} \cdot mk^{12} + k^{44} \cdot mk^{-5}}
= \tO{m^{3/2} k^{11} \hT^{-1} + mk^{39}}\,.\]
In order to set the first term to be at most $\tO{m^{3/2} / k}$, we set
$\hT = k^{12}$.

Therefore the total runtime of this operation is 
$\tO{m^{3/2} / k + m k^{39}}$.

{\bf $\cL$.\textsc{Update}}: 
This is run 
$\tO{m^{1/2} \left(k^6 \hT + k^{15}\right)}
=\tO{m^{1/2} k^{18}}$
times
and the amortized cost per operation is
\begin{align*}
& \tO{m \cdot \frac{\heps \alpha^{1/2}}{\eps^{3}} + 
\heps^{-4} \eps^{-2} \beta^{-8} + 
\heps^{-2} \eps^{-4} \alpha^2 \beta^{-6}}\\
& =\tO{m \cdot k^{44} \heps + k^{140} \heps^{-4} + 
k^{224} \heps^{-2}}\,,
\end{align*}
so in total
\[ m^{3/2} k^{62} \heps + m^{1/2} k^{158} \heps^{-4} + m^{1/2} k^{242} \heps^{-2}\,. \]
As we need the first term to be $\tO{m^{3/2} / k}$, we set
$\heps = k^{-63}$.
Therefore the total runtime is 
\[ \tO{m^{3/2} / k + m^{1/2} k^{410} + m^{1/2} k^{368}} = \tO{m^{3/2} / k + m^{1/2} k^{410}}\,. \]

{\bf $\cL$.\textsc{Initialize}}: This is run 
$\tO{k^{3} \beta^{-1/2}} = k^{11}$ times
in total, and the runtime for each run is 
\[ \tO{m \cdot \left(\heps^{-4} \beta^{-8} +
\heps^{-2} \eps^{-2} \alpha^2 \beta^{-4}}\right)
= \tO{m \cdot \left(k^{380} +
k^{306}\right)}
= \tO{m\cdot k^{380}} \,,\]
so in total $\tO{m k^{380}}$.

Therefore, for the whole algorithm, we get
$\tO{m^{3/2} / k + m^{1/2} k^{410} + mk^{380}}$
which after balancing gives $k=m^{1/762}$.

\section{An Efficient $(\alpha,\beta,\eps)$-\textsc{Locator}}
\label{sec:locator}

In this section we will show how to implement an 
$(\alpha,\beta,\eps)$-\textsc{Locator}, as defined in Definition~\ref{def:locator}.
In order to maintain the approximate electrical flow $\tff$
required by Lemma~\ref{lem:locator}
we will keep a vertex sparsifier in the form of a sparsified Schur complement onto
some vertex set $C$. As in~\cite{gao2021fully}, we choose $C$ to be a \emph{congestion reduction subset}.
\begin{definition}[Congestion reduction subset~\cite{gao2021fully}]\label{def:cong_red}
Given a graph $G(V,E)$ with resistances $\rr$ and any parameter $\beta\in(0,1)$,
a vertex subset $C\subseteq V$ is called a 
\emph{$\beta$-congestion reduction subset} (or just \emph{congestion reduction subset}) if:
\begin{itemize}
\item{$|C| \leq O(\beta m)$}
\item{For any $u\in V$, a random walk starting from $u$ that visits $\tOm{\beta^{-1} \log n}$ distinct vertices
hits $C$ with high probability}
\item{
If we generate $\mathrm{deg}(u)$ random walks from each $u\in V\backslash C$, the expected number of
these that hit some fixed $v\in V\backslash C$ before $C$ is $\tO{1/\beta^2}$. Concretely:
\begin{align}
\sum\limits_{u\in V} \mathrm{deg}(u) \cdot p_v^{C\cup\{v\}}(u) \leq \tO{1/\beta^2}\,.
\label{eq:cong_red}
\end{align}
}
\end{itemize}
\end{definition}
The following lemma shows that such a vertex subset can be constructed efficiently:
\begin{lemma}[Construction of congestion reduction subset~\cite{gao2021fully}]
Given a graph $G(V,E)$ with resistances $\rr$ and a parameter $\beta\in(0,1)$,
there is an algorithm that generates a $\beta$-congestion reduction subset
in time $\tO{m/\beta^2}$.
\label{lem:cong_red}
\end{lemma}

Intuitively, (\ref{eq:cong_red}) says that ``not too many'' random walks go through a given
vertex before reaching $C$.
This property is crucial for
ensuring that when inserting a new vertex into $C$, the data structure
will not have to change too much. 
As we will see in Section~\ref{sec:Fsystem}, this property plays an even more central role
when general demands are introduced, as 
it allows us to show that the demands outside $C$
can be pushed to $C$.
Additionally, in Section~\ref{sec:important} we will use it to show that 
edges that are too far from $C$ in effective resistance metric are not \emph{important},
in the sense that neither can they get congested, nor can their demand congest anything else.

\subsection{Moving demands to the sparsifier}
\label{sec:Fsystem}

The goal of this section is to show that
if $C$ is a congestion reduction subset, then
any demand of the form
$\dd = \BB^\top \frac{\qq}{\sqrt{\rr}}$
for some $\qq\in[-1,1]^m$
can be approximated by 
$\vpi^C(\dd)$, i.e.
its demand projection 
onto $C$ (Definition~\ref{def:demand_projection}).
This allows us to move all demands to the sublinear-sized $C$ and thus enables us to work
with the Schur complement of $G$ onto $C$.
\begin{lemma}\label{lem:non-projected-demand-contrib}
Consider a graph $G(V,E)$ with resistances $\rr$ and Laplacian $\LL$, a $\beta$-congestion reduction subset $C$,
and a demand $\dd = \delta \BB^\top \frac{\qq}{\sqrt{\rr}}$ for some $\delta > 0$ and
$\qq \in[-1,1]^m$.
Then, the potential embedding defined as
\begin{align*}
\vphi = \LL^+ \left(\dd - \vpi^C\left(\dd\right)\right)
\end{align*}
has congestion $\delta \cdot \tO{1/\beta^2}$, i.e. %
$\left\|\frac{\BB \vphi}{\sqrt{\rr}}\right\|_\infty \leq \delta \cdot \tO{1/\beta^2}$. %
\end{lemma}
We first prove a restricted version of the lemma where 
$\dd$ is an $s-t$ demand. Then, Lemma~\ref{lem:non-projected-demand-contrib} follows trivially
by applying (\ref{eq:cong_red}).
\begin{lemma}
Consider a graph $G(V,E)$ with resistances $\rr$ and Laplacian $\LL$, a $\beta$-congestion reduction subset $C$,
and a demand $\dd = \delta \BB^\top \frac{\onev_{st}}{\sqrt{\rr}}$ for some $\delta > 0$ and
$(s,t)\in E\backslash E(C)$.
Then, for the potential embedding defined as
\begin{align*}
\vphi = \LL^+ \left(\dd - \vpi^C\left(\dd\right)\right)
\end{align*}
it follows that for any $e=(u,v)\in E$ we have
\begin{align*}
\left|\frac{(\BB \vphi)_e}{\sqrt{r_e}}\right| \leq 
2 \delta \cdot 
\left(p_u^{C\cup\{u\}}(s) + p_v^{C\cup\{v\}}(s) + p_u^{C\cup\{u\}}(t) + p_v^{C\cup\{v\}}(t)
\right)\,.
\end{align*}
\label{lem:F_system1}
\end{lemma}
The proof of Lemma~\ref{lem:F_system1} appears in Appendix~\ref{proof_lem_F_system1}.

\subsection{$\eps$-Important edges}
\label{sec:important}

In this section we will show that the effect of edges that are ``far'' from the congestion reduction subset $C$
is negligible, as both their congestion and the congestion incurred because of their demands are small.
More specifically, given a demand 
$\dd$ supported on $C$ with energy $\leq 1$, i.e. $\mathcal{E}_{\rr}\left(\dd\right) \leq 1$,
the congestion $\vrho = \RR^{-1/2} \BB\LL^+ \dd_C$ that it induces satisfies:
\begin{align*}
\left|\rho_e\right| 
&= \left|\left\langle \onev_e, \RR^{-1/2}\BB \LL^+ \dd \right\rangle\right| \\
& = \left|\left\langle\BB^\top \frac{\onev_e}{\sqrt{\rr}}, \LL^+ \dd\right\rangle\right|\\
& = \left|\left\langle\vpi^C\left(\BB^\top \frac{\onev_e}{\sqrt{\rr}}\right), SC^+ \dd_C\right\rangle\right|\\
&\leq \sqrt{\mathcal{E}_{\rr}\left(\vpi^C\left(\BB^\top \frac{\onev_e}{\sqrt{\rr}}\right)\right) \mathcal{E}_{\rr}\left(\dd\right)} \\
& \leq \sqrt{\mathcal{E}_{\rr}\left(\vpi^C\left(\BB^\top \frac{\onev_e}{\sqrt{\rr}}\right)\right)}\,.
\end{align*}
For the last equality we used Fact~\ref{fact:demand_proj},
for the first inequality we applied Cauchy-Schwarz, 
and for the second one we used the upper bound on the energy required to route $\dd$.
Therefore, if we bound the energy of the projection of $\BB^\top\frac{\onev_e}{\sqrt{\rr}}$ onto $C$, we can
also bound the congestion of $e$. This is done in the following lemma, whose proof appears
in Appendix~\ref{proof_st_projection1_energy}.

\begin{lemma}
Consider a graph $G(V,E)$ with resistances $\rr$ and $C\subseteq V$. %
Then, for all $e\in E\backslash E(C)$ we have  
\[ \sqrt{\mathcal{E}_r\left(\vpi^C\left(\BB^\top \frac{\onev_e}{\sqrt{r_e}}\right)\right)} \leq 
6\cdot \sqrt{\frac{r_e}{R_{eff}(C,e)}}\,. \]
\label{st_projection1_energy}
\end{lemma}
This is the consequence of the following lemma, which bounds the magnitude of the projection
on a specific vertex, based on its effective resistance distance from $e$, as well as hitting
probabilities from $e$ to $C$.
The proof appears in Appendix~\ref{proof_st_projection1}.
\begin{lemma}
Consider a graph $G(V,E)$ with resistances $\rr$ and a subset of vertices $C\subseteq V$.
For any vertex $v\in V\backslash C$
we have that
\begin{align*}
\left|\pi_v^{C\cup\{v\}}\left(\BB^\top \frac{\onev_e}{\sqrt{\rr}}\right)\right| \leq (p_v^{C\cup\{v\}}(u) + p_v^{C\cup\{v\}}(w)) \cdot \frac{\sqrt{r_e}}{R_{eff}(v,e)} \,.
\end{align*}
\label{st_projection1}
\end{lemma}
This lemma is complementary to the more immediate property 
\[ \left|\pi_v^{C\cup\{v\}}\left(\BB^\top \frac{\onev_e}{\sqrt{\rr}}\right)\right| \leq (p_v^{C\cup\{v\}}(u) + p_v^{C\cup\{v\}}(w)) \cdot \frac{1}{\sqrt{r_e}}\,, \]
and they are both used in Section~\ref{sec:demand_projection}
in order to estimate demand projections.
In fact, the just by multiplying these two, we get the following
lemma, which is nice because it doesn't depend on $r_e$:
\begin{lemma}
Consider a graph $G(V,E)$ with resistances $\rr$ and a subset of vertices $C\subseteq V$.
For any vertex $v\in V\backslash C$
we have that
\begin{align*}
\left|\pi_v^{C\cup\{v\}}\left(\BB^\top \frac{\onev_e}{\sqrt{\rr}}\right)\right| \leq (p_v^{C\cup\{v\}}(u) + p_v^{C\cup\{v\}}(w)) \cdot \frac{1}{\sqrt{R_{eff}(v,e)}} \,.
\end{align*}
\label{st_projection_combined}
\end{lemma}

By the previous discussion, Lemma~\ref{st_projection1}
implies that if $\left|\rho_e\right| \geq \eps$, then
$R_{eff}(C,e) \leq r_e\cdot \frac{36}{\eps^2}$.
This motivates the following definition of $\eps$-\emph{important edges}.

\begin{definition}[Important edges]
An edge $e\in E$ is called $\eps$-\emph{important} (or just \emph{important}) if 
$R_{eff}(C,e) \leq r_e / \eps^2$.
\end{definition}

Now it is time for the main lemma of this section, which uses Lemma~\ref{st_projection1_energy}
to show that if our goal is to detect edges with congestion $\geq \eps$,
it is sufficient to restrict to computing demand projections of 
$\Omega(\eps)$-important edges.
Its proof appears in Appendix~\ref{proof_lem_important_edges}.
\begin{lemma}[Localization lemma]
Let $\vphi^*$ be any solution of
\begin{align*}
\LL \vphi^* = \delta \cdot \vpi^C\left(\BB^\top \frac{\pp}{\sqrt{\rr}} \right)\,,
\end{align*}
where $\rr$ are any resistances, $\pp\in[-1,1]^m$, and $C\subseteq V$.
Then,
for any $e\in E$ that is not $\eps$-important
we have
$\left|\frac{\BB\vphi^*}{\sqrt{\rr}}\right|_e \leq 6\eps$.
\label{lem:important_edges}
\end{lemma}

\subsection{Proving Lemma~\ref{lem:locator}}
\label{sec:proof_lem_locator}

Before moving to the description of how $\textsc{Locator}$
works and its proof, we will provide a lemma which 
bounds how fast a demand projection changes. 

We will use the following observation, which states that if our
congestion reduction subset $C$ contains an $\beta m$-sized
uniformly random edge subset, then with high probability, effective resistance
neighborhoods that are disjoint from $C$ only have $\tO{\beta^{-1}}$ edges.
Note that this is will be true throughout the algorithm as long as the resistances
do not depend on the randomness of $C$. This is true, as resistance updates
are only ever given as inputs to $\textsc{Locator}$.

\begin{lemma}[Few edges in a small neighborhood]
\label{lem:small-neighborhood2}
Let $\beta\in(0,1)$ be a parameter and $C$ be a vertex set which contains a subset of $\beta m$ edges sampled at random.
Then with high probability, for any $v \in V\backslash C$ we have that 
$\vert N_E(v, R_{eff}(C,v) / 2) \vert \leq 10\beta^{-1} \ln m$,
where 
\[ N_E(v,R) := \{e\in E\ |\ R_{eff}(e,v) \leq R\} \,. \]
\end{lemma}
\begin{proof}
Suppose that for some vertex $v\in V\backslash C$, $\vert N_E(u,R_{eff}(C,v) / 2) \vert \geq 10 \beta^{-1} \ln m$. 
Since by construction $C$ contains a random edge subset of size $\beta m$, with high probability 
$N_E(u,R_{eff}(C,u)/2)  \cap C \neq \emptyset$, so there exists $u\in C$ such that $R_{eff}(u,v) \leq R_{eff}(C,v) / 2$.
This is a contradiction since $u\in C$ implies $R_{eff}(C,v) \leq R_{eff}(u,v)$.
Union bounding over all $v$ yields the claim.
\end{proof}

Using this fact, we can now show that 
the change of the demand projection (measured in energy) is quite mild.
The proof of the following
lemma can be found in Appendix~\ref{proof_lem_projection_change_energy}.
\begin{lemma}[Projection change]
Consider a graph $G(V,E)$ with resistances $\rr$, $\qq\in[-1,1]^m$,
and a $\beta$-congestion reduction subset $C$.
Then, with high probability,
\[ \sqrt{\mathcal{E}_{\rr}\left(
\vpi^{C\cup\{v\}}\left(\BB^\top \frac{\qq}{\sqrt{\rr}}\right) 
-  \vpi^{C}\left(\BB^\top \frac{\qq}{\sqrt{\rr}}\right)\right)}
\leq \tO{\beta^{-2}} \,.
\]
\label{lem:projection_change_energy}
\end{lemma}

The above lemma can be applied over multiple vertex
insertions and resistance changes, to bound
the overall energy change. This is shown in the following lemma, which is proved in
Appendix~\ref{sec:proof_old_projection_approximate}:

\begin{lemma}
Consider a graph $G(V,E)$ with resistances $\rr^0$, $\qq^0\in[-1,1]^m$, a $\beta$-congestion
reduction subset $C^0$, %
and a fixed sequence of updates,
where the $i$-th update $i\in\{0,\dots,T-1\}$ is of the following form:
\begin{itemize}
    \item {\textsc{AddTerminal}($v^i$): Set $C^{i+1} = C^{i} \cup \{v^i\}$ for some $v^i\in V\backslash C^i$, $q_e^{i+1} = q_e^{i}, r_e^{i+1} = r_e^{i}$}
    \item {\textsc{Update}($e^i,\qq,\rr$) Set $C^{i+1} = C^{i}$, $q_e^{i+1} = q_e$ $r_e^{i+1} = r_e$, where $e^i\in E(C^{i})$}
\end{itemize}
Then, with high probability,
\[
\sqrt{\mathcal{E}_{\rr^T} \left(\vpi^{C^0,\rr^0}\left(\BB^\top \frac{\qq_S^0}{\sqrt{\rr^0}}\right) - \vpi^{C^T,\rr^T}\left(\BB^\top \frac{\qq_S^T}{\sqrt{\rr^T}}\right)\right)}
\leq \tO{\max_{i\in\{0,\dots,T-1\}} \left\|\frac{\rr^T}{\rr^i}\right\|_\infty^{1/2} \beta^{-2}} \cdot T\,.
\]

\label{lem:old_projection_approximate}
\end{lemma}

We are now ready to describe the $\textsc{Locator}$ data structure.
We will give an outline here, and defer the full proof
to Appendix~\ref{sec:full_proof_lem_locator}.
The goal of an $(\alpha,\beta,\eps)$-\textsc{Locator} is, given some flow $\ff$ with
slacks $\ss$ and resistances $\rr$, to compute all $e\in E$ such that
$\sqrt{r_e} \left|\tf_e^*\right| \geq \eps$, where
\[ \tff^* = \delta g(\ss) - \delta \RR^{-1} \BB \LL^+ \BB^\top g(\ss)\]
($\LL = \BB^\top \RR^{-1} \BB$),
where $\delta = 1/\sqrt{m}$.

If we set
$\rho_e^* = \sqrt{r_e} \tf_e^*$, we can equivalently write
\[ \vrho^* = \delta \sqrt{\rr} g(\ss) - \delta \RR^{-1/2} \BB \LL^+ \BB^\top g(\ss)\,,\]
and require to find all the entries of $\vrho^*$ with magnitude at least $\eps$.
As $\delta\left\|\sqrt{\rr} g(\ss)\right\|_\infty \leq \delta \leq \eps / 100$,
we can concentrate on the second
term, and denote
\[ \vrho'^* = \delta \RR^{-1/2} \BB \LL^+ \BB^\top g(\ss)\]
for convenience.

First, we use Lemma~\ref{lem:non-projected-demand-contrib}
to show that 
\[ \delta \left\|\RR^{-1/2} \BB \LL^+ 
\left(g(\ss) - \vpi^{C}(g(\ss))\right)\right\|_\infty 
\leq \delta \cdot \tO{\beta^{-2}} \leq \eps / 100 \,. \]
Now, let's set $\vpi_{old} = \vpi^{C^0}\left(g(\ss^0)\right)$, where
$C^0$ was the vertex set of the sparsifier and $\ss^0$ the slacks after the last 
call to $\textsc{BatchUpdate}$. As we will be calling $\textsc{BatchUpdate}$ at least
every $T$ calls to 
$\textsc{Update}$ for some $T\geq 1$, Lemma~\ref{lem:old_projection_approximate} implies that
\[ \delta \left\|\RR^{-1/2} \BB \LL^+ 
\left(\vpi^{C}(g(\ss)) - \vpi_{old}\right)\right\|_\infty 
\leq \delta \cdot \tO{\alpha \beta^{-2}} T \leq \eps / 100 \,, \]
as long as $T \leq \eps \sqrt{m} / \tO{\alpha \beta^{-2}}$.

Importantly, we will never be \emph{removing} vertices from $C$, so $C^0\subseteq C$.
This implies that it suffices to find the large entries of 
\[ \delta \RR^{-1/2} \BB \LL^+ \vpi_{old}\,. \]

Now, note that for any edge $e$ that was 
\emph{not} $\eps / (100\alpha)$-important for $C^0$ and corresponding resistances $\rr^0$,
we have
\begin{align*}
& \delta \left|\RR^{-1/2} \BB \LL^+ \vpi_{old}\right|_e\\
& \leq \delta \sqrt{\mathcal{E}_{\rr}(\vpi^{C^0}(\BB^\top \frac{\onev_e}{\sqrt{\rr}}))}\sqrt{\mathcal{E}_{\rr}(\vpi_{old})}\\
& \leq \delta \cdot \sqrt{\alpha} \frac{\eps}{100\alpha} \cdot \sqrt{2\alpha m}\\
& = \eps / 50\,,
\end{align*}
where we used Lemma~\ref{st_projection1_energy} and the fact that
$\mathcal{E}_{\rr}(\vpi^{C^0}(g(\ss^0))) \leq 2 \mathcal{E}_{\rr}(g(\ss^0)) \leq 2 \alpha m$.
Therefore it suffices to approximate
\[ \delta \II_S \RR^{-1/2} \BB \LL^+ \vpi_{old} \,,\]
where $S$ was the set of $\frac{\eps}{100\alpha}$-important edges last computed during the 
last call to $\textsc{BatchUpdate}$.

Now, we will use the sketching lemma from (Lemma 5.1, \cite{gao2021fully} v2),
which shows that in order to find all $\Omega(\eps)$ large
entries of this vector, it suffices to compute the inner
products
\begin{align*}
& \delta \left\langle \vpi^C\left(\BB^\top \frac{\qq_S^i}{\sqrt{\rr}}\right), 
SC^+ \vpi_{old} \right\rangle
\end{align*}
for $i\in[\tO{\eps^{-2}}]$ up to $O(\eps)$ accuracy.
Here $SC$ is the Schur complement onto $C$.

Based on this, there are two types of quantities that we will maintain:
\begin{itemize}
\item $\tO{1/\eps^2}$ approximate demand projections 
$\tvpi^C\left(\BB^\top \frac{\qq_S^i}{\sqrt{\rr}}\right)$, and
\item an approximate Schur complement $\tSC$ of $G$ onto $C$.
\end{itemize}
For the latter, we will directly use the dynamic Schur complement data structure 
$\textsc{DynamicSC}$ that was also used
by~\cite{gao2021fully} and is based on~\cite{durfee2019fully}. For completeness, we present this data structure in Appendix~\ref{sec:maintain_schur}.

For the former, we will need $\tO{1/\eps^2}$ data structures for maintaining demand projections onto $C$, under vertex insertions to $C$.
The guarantees of each such a data structure, 
that we call an $(\alpha,\beta,\eps)$-\textsc{DemandProjector}, are as follows.
\begin{definition}[$(\alpha,\beta,\eps)$-\textsc{DemandProjector}]
Let $\heps\in(0,\eps)$ be a tradeoff parameter.
Given a graph $G(V,E)$, resistances $\rr$, and a vector $\qq\in[-1,1]^m$,
an $(\alpha,\beta,\eps)$-\textsc{DemandProjector} is a data structure that maintains a vertex subset $C\subseteq V$ and an approximation
to the demand projection $\vpi^C\left(\BB^\top \frac{\qq}{\sqrt{\rr}}\right)$, with high probability under
oblivious adversaries. The following operations are supported:
\begin{itemize}
\item{$\textsc{Initialize}(C,\rr,\qq,S,\mathcal{P})$: Initialize 
the data structure in order to maintain 
an approximation of $\vpi^C\left(\BB^\top \frac{\qq_S}{\sqrt{\rr}}\right)$,
where $C\subseteq V$ is a $\beta$-congestion reduction subset,
$\rr$ are resistances,
$\qq\in[-1,1]^m$, and
$S\subseteq E$ is a subset of 
$\gamma$-important edges.
$\mathcal{P} = \{\mathcal{P}^{u,e,i}
\ |\ u\in V, e\in E, u\in e, i\in[h]\}$
for some $h\in\mathbb{Z}_{\geq 1}$,
is a set of independent random walks from $u$ to $C$ for any $u$.
}
\item{$\textsc{AddTerminal}(v,\tR_{eff}(C,v))$: Insert $v$ into $C$.
Also, $\tR_{eff}(C,v)$ is an estimate of $R_{eff}(C,v)$
such that $\tR_{eff}(C,v)\approx_2 R_{eff}(C,v)$.
Returns an estimate
\[
\tpi_v^{C\cup\{v\}}\left(\BB^\top \frac{\qq}{\sqrt{\rr}}\right)
\] for the demand projection of $\qq$ onto $C\cup \{v\}$ at coordinate $v$ 
such that
\[
\left| \tpi_v^{C\cup\{v\}}\left(\BB^\top \frac{\qq}{\sqrt{\rr}}\right) - \pi_v^{C\cup\{v\}}\left(\BB^\top \frac{\qq}{\sqrt{\rr}}\right) \right| \leq \frac{\heps}{\sqrt{R_{eff}(v,C)}}\,.
\]
}
\item{$\textsc{Update}(e,\rr',\qq')$: Set $r_e = r_e'$ and $q_e = q_e'$, 
where $e\in E(C)$,
and $q_e'\in[-1,1]$. Furthermore, $r_e'$ satisfies the inequality $r_e^{\max} / \alpha \leq r_e' \leq \alpha \cdot r_e^{\min}$, where $r_e^{\min}$ and $r_e^{\max}$ represent the minimum, respectively the maximum values that the resistance of $e$ has had since the last call to $\textsc{Initialize}$.
}
\item{$\textsc{Output}()$: Output 
$\tvpi^C\left(\BB^\top \frac{\qq_S}{\sqrt{\rr}}\right)$ such that
such that after $T \leq n^{O(1)}$ calls to  $\textsc{AddTerminal}$,
 for any fixed vector $\vphi$, $E_{\rr}(\vphi) \leq 1$, with high probability
\[
\left| \left\langle \tvpi^C\left(\BB^\top \frac{\qq_S}{\sqrt{\rr}}\right) - \vpi^C\left(\BB^\top \frac{\qq_S}{\sqrt{\rr}}\right), \vphi\right \rangle   \right| \leq \heps \cdot \sqrt{\alpha} \cdot T \,.
\]
}
\end{itemize}
\label{def:demand_projector}
\end{definition}
We will implement such a data structure in Section~\ref{sec:demand_projection}, where
we will prove the following lemma:
\begin{lemma}[Demand projection data structure]
For any graph $G(V,E)$ and parameters $\heps\in(0,\eps)$, $\beta\in(0,1)$,
there exists an $(\alpha,\beta,\eps)$-\textsc{DemandProjector} for $G$ 
which, given access to   $h=\widetilde{\Theta}(\heps^{-4}\beta^{-6}+\heps^{-2}\beta^{-2}\gamma^{-2})$ precomputed independent  random walks from $u$ to $C$ for each $e \in E$, $u \in e$,
has the following
runtimes per operation:
\begin{itemize}
\item{\textsc{Initialize}:
$\tO{m}$.
}
\item{\textsc{AddTerminal}:
$\tO{\heps^{-4}\beta^{-8}+\heps^{-2} \beta^{-6} \gamma^{-2}}$.
}
\item{\textsc{Update}:
$O(1)$.}
\item{\textsc{Output}:
$O(\beta m + T)$, where $T$ is the number of calls made to $\textsc{AddTerminal}$ after the last call to $\textsc{Initialize}$.}
\end{itemize}
\label{lem:ds}
\end{lemma}

Now we describe the way we will use the $\textsc{DemandProjector}$s 
and $\textsc{DynamicSC}$ to get an $(\alpha,\beta,\eps)$-\textsc{Locator} $\cL$.
\begin{algorithm} %
\begin{algorithmic}[1]
\caption{\textsc{Locator} $\cL$.\textsc{Initialize}}
\Procedure{$\cL$.\textsc{Initialize}}{$\ff$}
\State $\ss^+ = \uu - \ff$, $\ss^- = \ff$, 
$\rr = \frac{1}{(\ss^+)^2} + \frac{1}{(\ss^-)^2}$
\State $\QQ =$ Sketching matrix produced by (Lemma 5.1, \cite{gao2021fully} v2)
\State $\textsc{DynamicSC} = \textsc{DynamicSC}.\textsc{Initialize}(\GG,\emptyset,\rr,\eps,\beta)$
\State $C = \textsc{DynamicSC}.C$ \Comment{$\beta$-congestion reduction subset}
\State Estimate $\tR_{eff}(C,e) \approx_{4} R_{eff}(C,e)$ using Lemma~\ref{lem:approx_effective_res}
\State $S = \left\{e\in E\ |\ \tR_{eff}(C,e) \leq r_e \cdot \left(\frac{100\alpha}{\eps}\right)^2\right\}$
\State $h = \widetilde{\Theta}\left( \heps^{-4} \beta^{-6} +  \heps^{-2} \eps^{-2}\alpha^2 \beta^{-2}  \right)$
\State Sample walks $\cP^{u,e,i}$ from $u$ to $C$ for $e \in E\setminus E(C)$, $u \in e$, $i \in [h]$ (Lemma 5.15, \cite{gao2021fully} v2)
\State $\textsc{DP}^i = \textsc{DemandProjector}.\textsc{Initialize}(C,\rr,\qq^i,S,\cP)$ for all rows $\qq^i$ of $\QQ$
\State $\cL.\textsc{BatchUpdate}(\emptyset)$
\EndProcedure
\end{algorithmic} 
\end{algorithm} 

{\bf $\cL$.\textsc{Initialize}}: Every time 
$\cL.\textsc{Initialize}$ is called, we first generate a $\beta$-congestion reduction
subset $C$ based on Lemma~\ref{lem:cong_red} (takes time
$\tO{m/\beta^2}$), 
then a sketching matrix $\QQ$ and its rows $\qq^i$ for $i\in\left[\tO{1/\eps^2}\right]$
as in~(Lemma 5.1, \cite{gao2021fully} v2)
(takes time $\tO{m/\eps^2}$),
and finally random walks $\cP^{u,e,i}$ 
from $u$ to $C$ for each $u\in V$,
$e\in E\backslash E(C)$ with $u\in e$, and $i\in [h]$, where
$h = \tO{ \heps^{-4} \beta^{-6} +  \heps^{-2} \eps^{-2}\alpha^2 \beta^{-2}  }$
as in (Lemma 5.15, \cite{gao2021fully} v2)
(takes time $\tO{h/\beta^2}$ for each $(u,e)$).

We also compute $\tR_{eff}(C,u) \approx_{2} R_{eff}(C,u)$ for all $u\in V$ as described
in Lemma~\ref{lem:approx_effective_res}
so that we can let 
$S$ be a subset of $\eps/(100\alpha)$-important edges that
contains all $\eps/(200\alpha)$-important edges.
This takes time $\tO{m}$.
Then, we
call \textsc{DynamicSC}.\textsc{Initialize}$(G,C,\rr, O(\eps),\beta)$
(from Appendix~\ref{sec:maintain_schur})
to initialize the dynamic Schur complement onto $C$,
with error tolerance $O(\eps)$,
which takes time $\tO{m \cdot \frac{1}{\eps^4\beta^4}}$,
as well as $\textsc{DemandProjector}.\textsc{Initialize}(C,\rr,\qq,S,\mathcal{P})$
for the $\tO{1/\eps^2}$
$\textsc{DemandProjector}$s, i.e. one for each $\qq \in \{\qq^i\ |\ i\in[\tO{1/\eps^2}]\}$.
Also, we compute 
\[ \vpi^{old} = \vpi^C\left(\BB^\top g(\ss)\right)\,, \]
which takes $\tO{m}$ as in $\textsc{DemandProjector}.\textsc{Initialize}$.
All of this takes 
$\tO{m\cdot \left(\frac{1}{\heps^4\beta^8} + \frac{\alpha^2}{\heps^2 \eps^2 \beta^4}\right)}$.

\begin{algorithm} %
\begin{algorithmic}[1]
\caption{\textsc{Locator} $\cL$.\textsc{Update} and $\cL$.\textsc{BatchUpdate}}
\Procedure{\textsc{Update}}{$e = (u,w),\ff$}
\State $s_e^+ = u_e - f_e$, $s_e^- = f_e$, $r_e = \frac{1}{(s_e^+)^2} + \frac{1}{(s_e^-)^2}$
\State $\tR_{eff}(C,u) = \textsc{DynamicSC}.\textsc{AddTerminal}(u)$
\State $\tR_{eff}(C\cup\{u\},w) = \textsc{DynamicSC}.\textsc{AddTerminal}(w)$
\State $C = C\cup\{u,w\}$

\For {$i=1,\dots,\tO{1/\eps^2}$}
		\State 
		$\textsc{DP}^i.\textsc{AddTerminal}(u, \tR_{eff}(C,u))$
    	\State 
		$\textsc{DP}^i.\textsc{AddTerminal}(w, \tR_{eff}(C\cup\{u\},w))$
\EndFor
\State $\textsc{DynamicSC}.\textsc{Update}(e,r_e)$
\For {$i=1\dots \tO{1/\eps^2}$} 
	\State $\textsc{DP}^i.\textsc{Update}(e,\rr,\qq^i)$
\EndFor
\EndProcedure
\Procedure{\textsc{BatchUpdate}}{$Z,\ff$}
\State $\ss^+ = \uu - \ff$, $\ss^- = \ff$, $\rr = \frac{1}{(\ss^+)^2} + \frac{1}{(\ss^-)^2}$
\State Estimate $\tR_{eff}(C,e) \approx_{4} R_{eff}(C,e)$ using Lemma~\ref{lem:approx_effective_res}
\State $S = \left\{e\in E\ |\ \tR_{eff}(C,e) \leq r_e \cdot  \left(\frac{100\alpha}{\eps}\right)^2\right\}$
\Comment{$\frac{\eps}{100\alpha}$-important edges}
\For {$e=(u,w)\in Z$}
	\State $\textsc{DynamicSC}.\textsc{AddTerminal}(u)$
	\State $\textsc{DynamicSC}.\textsc{AddTerminal}(w)$
	\State $C = C\cup\{u,w\}$
	\State $\textsc{DynamicSC}.\textsc{Update}(e,r_e)$
\EndFor
\For {$i=\left[\tO{1/\epsilon^2}\right]$}
	\State $\textsc{DP}^i.\textsc{Initialize}(C,\rr, \qq^i, S, \mathcal{P})$
\EndFor
\State $\vpi_{old} = \frac{1}{\sqrt{m}} \cdot \vpi^C\left(\BB^\top \frac{\frac{1}{\ss^+} - \frac{1}{\ss^-}}{\rr}\right)$ \Comment{Compute exactly using Laplacian solve}
\EndProcedure
\end{algorithmic} 
\end{algorithm} 
{\bf $\cL$.\textsc{Update}}: Now, whenever $\cL.\textsc{Update}$ is called on an edge $e$, either $e\in E(C)$ or $e\notin E(C)$.
In the first case we simply call $\textsc{Update}$ on 
$\textsc{DynamicSC}$ and all \textsc{DemandProjector}s.

In the second case, we first call $\textsc{DynamicSC}.\textsc{AddTerminal}$ 
on one endpoint $v$ of $e$. After doing this we can also get an 
estimate $\tR_{eff}(C,v)\approx_{2} R_{eff}(C,v)$ 
by looking at the edges between $C$ and $v$ in the sparsified
Schur complement. By the guarantees of the
expander decomposition used inside $\textsc{DynamicSC}$~\cite{gao2021fully}, the
number of expanders containing $v$, 
amortized over all calls to $\textsc{DynamicSC}.\textsc{AddTerminal}$,
is $O(\mathrm{poly}\log(n))$. As
the sparsified Schur complement contains
$\tO{1/\eps^2}$ neighbors of $v$ from each expander, the amortized number of neighbors of $v$
in the sparsified Schur complement is $\tO{1/\eps^2}$, and the amortized
runtime to generate them (by random sampling) is $\tO{1/\eps^2}$.

Given the resistances $r_1,\dots,r_l$ of these edges,
setting $\tR_{eff}(C,v) = \left(\sum\limits_{i=1}^l r_i^{-1}\right)^{-1}$ we guarantee
that $\tR_{eff}(C,v) \approx_{1+O(\eps)} R_{eff}(C,v)$,
by the fact that $\textsc{DynamicSC}$ maintains an $(1+O(\eps))$-approximate sparsifier
of the Schur complement.
Then, we call $\textsc{AddTerminal}(v,\tR_{eff}(C,v))$ on all
\textsc{DemandProjector}s. 

After repeating the same process for the other endpoint of $e$,
we finally call $\textsc{Update}$ on 
$\textsc{DynamicSC}$ and all \textsc{DemandProjector}s.
This takes time $\tO{\frac{1}{\eps^2 \beta^2}}$ because of the Schur complement and amortized 
$\tO{m\cdot \frac{\heps \alpha^{1/2}}{\eps} + \frac{1}{\heps^4\beta^{8}} + \frac{\alpha^2}{\heps^2\eps^2\beta^{-6}}}$ for each of the
demand projectors, so  the total amortized runtime is
$\tO{m\cdot \frac{\heps \alpha^{1/2}}{\eps^3} + \frac{1}{\heps^4\eps^2\beta^{8}} + \frac{\alpha^2}{\heps^2\eps^4\beta^6}}$.

{\bf $\cL$.\textsc{BatchUpdate}}: When $\cL.\textsc{BatchUpdate}$ is called on a set of edges $Z$, 
we add them one by one in the $\textsc{DynamicSC}$ data structure following
the same process as in $\cL.\textsc{Update}$. For the demand projectors,
we first manually insert the endpoints of these edges
into $C$ and then re-initialize all \textsc{DemandProjector}s,
by calling $\textsc{Initialize}$ with a new subset
 $S$ 
of $\frac{\eps}{200\alpha}$-important edges that
contains all $\frac{\eps}{100\alpha}$-important edges.
Such a set can be computed by estimating $R_{eff}(C,u)$ for all $u\in V\backslash C$
up to a constant factor
and, by Lemma~\ref{lem:approx_effective_res}, takes time $\tO{m}$.
Also, we compute 
\[ \vpi^{old} = \vpi^C\left(\BB^\top g(\ss)\right)\,, \]
which takes $\tO{m}$ as in $\textsc{DemandProjector}.\textsc{Initialize}$.
The total runtime of this is $\tO{ m / \eps^2 + |Z| / (\beta^2\eps^2) }$.

\begin{algorithm} %
\begin{algorithmic}[1]
\caption{\textsc{Locator} $\cL$.\textsc{Solve}}
\Procedure{\textsc{Solve}}{$ $}
\State $\tSC = \textsc{DynamicSC}.\tSC()$
\State $\vphi_{old} = \tSC^+ \vpi_{old}$ %
\State $\vv = \zerov$
\For {$i=1,\dots,\tO{1/\eps^2}$}
    \State %
$\tvpi^i
= \DP^i.\textsc{Output}()$ 
\State $v_i = \langle \tvpi^i , \vphi_{old} \rangle$
\EndFor
\State $Z = \textsc{Recover}(\vv, \eps/100)$ \Comment{Recovers all $\eps/2$-congested edges
(Lemma 5.1, \cite{gao2021fully} v2)}
\State \Return $Z$

\EndProcedure
\end{algorithmic} 
\end{algorithm} 

{\bf $\cL$.\textsc{Solve}}: When $\cL.\textsc{Solve}$ is called, we set $\tSC = \textsc{DynamicSC}.SC()$,
call $\textsc{Output}$ on all \textsc{DemandProjector}s
to obtain vectors $\tvpi^i$ which are estimators for $\vpi^C(\BB^\top \frac{\qq^i_S}{\sqrt{\rr}})$ in the sense of Definition~\ref{def:demand_projector}.
Then we compute $v_i = \langle  \tvpi^i , \tSC^+ \vpi_{old}\rangle$ where $\vpi_{old}$ is the demand projection that was computed exactly the last time $\textsc{BatchUpdate}$ was called.
These computed terms represent an approximation to the update in $(\QQ\vrho)_i$ between two consecutive calls of $\cL.\textsc{Solve}$. 
As we will show in the appendix, $\langle  \tvpi^i , \tSC^+ \vpi_{old}\rangle$
 is an $\eps$-additive approximation of
$\langle \vpi^{C}(\BB^\top \frac{\qq^i}{\sqrt{\rr}}), \LL^+ \vpi^{C}(\BB^\top g(\ss)) \rangle$
for all $i\in\left[\tO{1/\eps^2}\right]$.
The key fact that makes this approximation feasible is that although updates to the demand projection are hard to approximate with few samples, when hitting them with the deterministic vector $\vpi_{old}$, the resulting inner products strongly concentrate.
The runtime of this is $\tO{\beta m / \eps^2}$.

Using these computed values with the $\ell_2$ heavy hitter data structure (Lemma 5.1, \cite{gao2021fully} v2)
we get all edges with congestion more than $\eps$.

The total runtime is $\tO{\beta m / \eps^2}$.

\section{The Demand Projection Data Structure}
\label{sec:demand_projection}

The main goal of this section is to construct an $(\alpha,\beta,\eps)$-\textsc{DemandProjector},
as defined in Definition~\ref{def:demand_projector}, and thus prove Lemma~\ref{lem:ds}.
The most important operation that needs to be implemented in order to prove Lemma~\ref{lem:ds} is 
to maintain the demand projection after inserting a vertex $v\in V\backslash C$ to $C$.
In order to do this, we use the following identity from~\cite{gao2021fully}:
\begin{align}
\vpi^{C\cup\{v\}}\left(\dd\right)
= \vpi^C\left(\dd\right)
+ \pi_v^{C\cup\{v\}}\left(\dd\right) \cdot \left(\onev_v - \vpi^C(\onev_v)\right)\,,
\label{eq:projection_update}
\end{align}
where $\dd$ is any demand (in our case, we have 
$\dd = \BB^\top \frac{\qq_S}{\sqrt{\rr}}$ for some $S\subseteq E$).
For this, we need to compute approximations to 
$\pi_v^{C\cup\{v\}}\left(\BB^\top \frac{\qq_S}{\sqrt{\rr}}\right)$
and $\vpi^C\left(\onev_v\right)$. 

In Section~\ref{sec:approx1}, we will show 
that if $S$ is a subset of $\gamma$-important edges, we can efficiently
estimate
$\pi_v^{C\cup\{v\}}\left(\BB^\top \frac{\qq_S}{\sqrt{\rr}}\right)$ up to 
additive accuracy $\frac{\heps}{R_{eff}(C,v)}$ by sampling random walks
to $C$ starting only from edges with relatively high resistance.
For the remaining edges, the $\gamma$-importance property will imply that
we are not losing much by ignoring them.

Then, in Section~\ref{sec:approx2} we will show how to approximate $\onev_v - \vpi^C(\onev_v)$.
This is equivalent to estimating the hitting probabilities from $v$ to $C$. 
The guarantee that we would ideally like to get is on the error to route 
\begin{align}
\mathcal{E}_{\rr}\left(\tvpi^C(\onev_v) - \vpi^C(\onev_v)\right) \leq \heps^2 R_{eff}(C,v)\,.
\label{eq:ideal_energy_bound}
\end{align}
Note that this is not possible to do efficiently for general $C$. For example, suppose that 
the hitting distribution is uniform. In this case, $\Omega(|C|)$ random walks are required
to get a bound similar to (\ref{eq:ideal_energy_bound}). However, it might still be possible
to guarantee it by using the structure of $C$, and this would simplify some parts of our analysis.
Instead, we are going to work with the following weaker approximation bound:
For any fixed potential vector $\vphi\in\mathbb{R}^n$ with $E_{\rr}(\vphi) \leq 1$, we have w.h.p.
\begin{align}
\left|\left\langle \tvpi^C(\onev_v) - \vpi^C(\onev_v), \vphi\right\rangle\right| \leq \heps \sqrt{R_{eff}(C,v)}\,.
\label{eq:less_than_ideal_energy_bound}
\end{align}

Now, using these estimation lemmas, we will bound how our demand projection 
degrades when inserting a new vertex into $C$.
This is stated in the following lemma and proved in Appendix~\ref{proof_lem_insert1}.
\begin{lemma}[Inserting a new vertex to $C$]
Consider a graph $G(V,E)$ with resistances $\rr$, $\qq\in[-1,1]^m$, 
a $\beta$-congestion reduction subset $C$,
and $v\in V\backslash C$.
We also suppose that we 
have an estimate of the $C-v$ effective resistance 
such that $\tR_{eff}(C,v) \approx_{2} R_{eff}(C,v)$, as well as to
independent random walks
$\mathcal{P}^{u,e,i}$ for each $u\in V\backslash C$, $e\in E\backslash E(C)$ with $u\in e$, $i\in[h]$,
where each random walk starts from $u$ and ends at $C$.

If we let $S$ be a subset of 
$\gamma$-important edges
for $\gamma > 0$,
then for any error parameter $\heps > 0$ we can compute 
$\tpi_v^{C\cup\{v\}}\left(\BB^\top \frac{\qq_S}{\sqrt{\rr}}\right)\in\mathbb{R}$ 
and 
$\tpi^{C\cup\{v\}}\left(\BB^\top \frac{\qq_S}{\sqrt{\rr}}\right)\in\mathbb{R} \in \mathbb{R}^n$
such that with high probability
\begin{align*}
\left|\tpi_{v}^{C\cup\{v\}}\left(\BB^\top \frac{\qq_S}{\sqrt{\rr}}\right)
- \pi_{v}^{C\cup\{v\}}\left(\BB^\top \frac{\qq_S}{\sqrt{\rr}}\right)\right|
\leq \frac{\heps}{ \sqrt{R_{eff}(C,v)}}\,,
\end{align*}
as long as $h=\tOm{\heps^{-4}\beta^{-6}+\heps^{-2}\beta^{-2}\gamma^{-2}}$.
Furthermore, for any fixed $\vphi$, $E_{\rr}(\vphi) \leq 1$, after $T$ insertions after the last call to $\textsc{Initialize}$, with high probability
\begin{align*}\label{eq:insert_error}
\left|\left\langle
 \tvpi^{C\cup\{v\}}\left(\BB^\top \frac{\qq_S}{\sqrt{\rr}}\right) 
 -
\vpi^{C\cup\{v\}}\left(\BB^\top \frac{\qq_S}{\sqrt{\rr}}\right)
, \vphi\right\rangle\right| \leq \heps T\,,
\end{align*}
as long as 
$h = \tOm{\heps^{-2} \beta^{-4} \gamma^{-2}}$.

\label{lem:insert1}
\end{lemma}
\begin{algorithm} %
\begin{algorithmic}[1]
\caption{\textsc{DemandProjector} \DP.\textsc{AddTerminal}}
\Procedure{\DP.\textsc{AddTerminal}}{$v,\tR_{eff}(C,v)$}
\If {$v\in C$}
\State \Return 
\EndIf
\State $t = t + 1$
\State $\tpi_v^{C\cup\{v\}}\left(\BB^\top \frac{\qq_S}{\sqrt{\rr}}\right) = 0$
\For {$u,e\in S,i$ such that $\mathcal{P}^{u,e,i} \ni v$ and $\tR_{eff}(C,v) \leq 
\frac{1}{\left( \min\{\heps / \tO{\beta^{-2}}, \gamma / 4\} \right)^2} r_e$} 
	\If {$e = (u,*)$}
		\State $\tpi_v^{C\cup\{v\}}\left(\BB^\top \frac{\qq_S}{\sqrt{\rr}}\right) 
		 = \tpi_v^{C\cup\{v\}}\left(\BB^\top \frac{\qq_S}{\sqrt{\rr}}\right) + \frac{1}{h} \frac{q_e}{\sqrt{r_e}}$
	\Else
		\State $\tpi_v^{C\cup\{v\}}\left(\BB^\top \frac{\qq_S}{\sqrt{\rr}}\right) 
		 = \tpi_v^{C\cup\{v\}}\left(\BB^\top \frac{\qq_S}{\sqrt{\rr}}\right) - \frac{1}{h} \frac{q_e}{\sqrt{r_e}}$
	\EndIf
	\State Shortcut $\mathcal{P}^{u,e,i}$ at $v$
\EndFor
\State $h' = \tO{\heps^{-2}\beta^{-4}\gamma^{-2}} $ %
\State $\tvpi^{C}\left(\onev_v\right) = \zerov$
\For {$i=1,\dots h'$}
	\State Run random walk from $v$ to $C$ with probabilities prop. to $\rr^{-1}$, let $u$ be the last vertex
	\State $\tpi_u^{C}\left(\onev_v\right) = \tpi_u^{C}\left(\onev_v\right) + \frac{1}{h'}$
\EndFor
\State $\tvpi^{C \cup \{v\}}(\BB^\top \frac{\qq_S}{\sqrt{\rr}}) =\tvpi^{C}(\BB^\top \frac{\qq_S}{\sqrt{\rr}}) +  \tpi_v^{C \cup \{v\}}(\BB^\top \frac{\qq_S}{\sqrt{\rr}}) \cdot (\onev_v - \tpi^C(\onev_v))$
\State $C = C\cup\{v\}$, $F = F\backslash \{v\}$
\EndProcedure
\end{algorithmic} 
\end{algorithm}

\subsection{Estimating $\pi_v^{C\cup\{v\}}\left(\BB^\top \frac{\qq_S}{\sqrt{\rr}}\right)$}
\label{sec:approx1}

There is a straightforward algorithm to estimate 
$\pi_v^{C\cup\{v\}}\left(\BB^\top \frac{\qq_S}{\sqrt{\rr}}\right)$.
For each edge $e=(u,w)\in E\backslash E(C)$, sample a number of random walks from $u$ and $w$
until they hit $C\cup\{v\}$. Then, add to the estimate $\frac{q_e}{\sqrt{r_e}}$ times the 
fraction of the random walks starting from $u$
that contain $v$,
minus $\frac{q_e}{\sqrt{r_e}}$ times the fraction of the
random walks starting from $w$ that contain $v$.
\cite{gao2021fully} uses this sampling method together with 
the following concentration bound, to get a good estimate if the resistances 
of all congested edges are sufficiently large.

\begin{lemma}[Concentration inequality 1~\cite{gao2021fully}]
Let $S = X_1 + \dots + X_n$ be the sum of $n$ independent random variables. The range of $X_i$ is $\{0,a_i\}$ for $a_i\in[-M,M]$. Let $t,E$ be positive numbers such
that $t \leq E$ and $\sum\limits_{i=1}^n \left|\mathbb{E}[X_i]\right| \leq E$. Then
\[ \Pr\left[|S - \mathbb{E}[S]| > t\right] \leq 2 \exp\left(-\frac{t^2}{6EM}\right) \,.\]
\label{conc1}
\end{lemma}
Unfortunately, in our setting there is no reason to expect 
these resistances to be large, so the variance of this estimate might be too high.
We have already introduced the concept of important edges in order to alleviate this problem,
and proved that we only need to look at important edges.
Even if all edges of which the demand projection is estimated
are important (i.e. close to $C$), however, %
$v$ can still be far from $C$. This is an issue, since we don't directly estimate projections
onto $C$, but instead estimate the projection
onto $C\cup\{v\}$ and then from $v$ onto $C$.

Intuitively, however, if $v$ is far from $C$, 
it should also be far from the set of important edges, so the insertion of $v$ should
not affect their demand projection too much.
As the distance upper bound between an important edge and $C$ is relative to the scale of
the resistance of that edge, this statement needs be more
fine-grained in order to take the resistances of important edges into account.

More concretely, in the following lemma,
which is proved in Appendix~\ref{proof_estimate1},
we show that if we only compute
demand projection estimates for edges $e$
such that
$r_e \geq c^2 R_{eff}(C,v)$
for some appropriately chosen $c > 0$, 
then we can guarantee a good bound on the number of random walks we need to sample.

For the remaining edges, we will show 
that the energy of their contributions to the projection is
negligible, so that we can reach to our desired statement in Lemma~\ref{estimate1_final}.

\begin{lemma} %
Consider a graph $G(V,E)$ with resistances $\rr$, $\qq\in[-1,1]^n$, 
a $\beta$-congestion reduction subset $C$,
as well as $v\in V\backslash C$.
If for some $c>0$ we are given a set of edges
\[ S' \subseteq 
\left\{e\in E\backslash E(C)\ |\  \text{$R_{eff}(C,v) \leq \frac{1}{c^2} r_e$}\right\} \,, \]
then for any $\delta_1' > 0$
we can compute 
$\tpi_v^{C\cup\{v\}}\left(\BB^\top \frac{\qq_{S'}}{\sqrt{\rr}}\right)\in\mathbb{R}$ such that 
with high probability
\[ \left|\tpi_v^{C\cup\{v\}}\left(\BB^\top \frac{\qq_{S'}}{\sqrt{\rr}}\right) - \pi_{v}^{C\cup\{v\}}\left(\BB^\top \frac{\qq_{S'}}{\sqrt{\rr}}\right)\right| \leq \frac{\delta_1'}{\beta c\sqrt{R_{eff}(C,v)}} \,.\]
The algorithm requires access to
$\tO{\delta_1'^{-2} \log n \log \frac{1}{\beta}}$ 
independent random walks from $u$ to $C$ for each $u\in V\backslash C$ and $e\in E\backslash E(C)$ with $u\in e$.
\label{estimate1}
\end{lemma}

This leads us to the desired statement for this section,
whose proof appears in Appendix~\ref{proof_estimate1_final}.
\begin{lemma}[Estimating $\pi_v^{C\cup\{v\}}\left(\BB^\top \frac{\qq}{\sqrt{\rr}}\right)$]
Consider a graph $G(V,E)$ with resistances $\rr$, $\qq\in[-1,1]^n$, 
a $\beta$-congestion reduction subset $C$,
as well as $v\in V\backslash C$.
If we are given a set
$S$ of $\gamma$-important edges for some $\gamma \in (0,1)$
and an estimate $\tR_{eff}(C,v)\approx_{2} R_{eff}(C,v)$,
then for any $\delta_1 \in(0,1)$
we can compute
$\tpi_v^{C\cup\{v\}}\left(\BB^\top \frac{\qq_{S}}{\sqrt{\rr}}\right)\in\mathbb{R}$
such that with high probability
\begin{align}
\left|\tpi_v^{C\cup\{v\}}\left(\BB^\top \frac{\qq_{S}}{\sqrt{\rr}}\right) - \pi_{v}^{C\cup\{v\}}\left(\BB^\top \frac{\qq_{S}}{\sqrt{\rr}}\right)\right| 
\leq \frac{\delta_1}{\sqrt{R_{eff}(C,v)}} \,.
\end{align}
The algorithm requires
$\tO{\delta_1^{-4} \beta^{-6} + \delta_1^{-2} \beta^{-2} \gamma^{-2}}$
independent random walks from $u$ to $C$ for each $u\in V\backslash C$ 
and $e\in E\backslash E(C)$ with $u\in e$.

Additionally, we have
\[ \left|\pi_{v}^{C\cup\{v\}}\left(\BB^\top \frac{\qq_{S}}{\sqrt{\rr}}\right)\right| 
\leq \frac{1}{\gamma \sqrt{R_{eff}(C,v)}} \cdot \tO{\frac{1}{\beta^2}} \,.\]
\label{estimate1_final}
\end{lemma}

\subsection{Estimating $\vpi^{C}(\onev_v)$}
\label{sec:approx2}

In contrast to the quantity 
$\pi_v^{C\cup\{v\}}\left(\BB^\top \frac{\qq}{\sqrt{\rr}}\right)$,
where there are cancellations between its two components
$\pi_v^{C\cup\{v\}}\left(\sum\limits_{e=(u,w)\in E}\frac{q_e}{\sqrt{r_e}} \onev_u\right)$ 
and
$\pi_v^{C\cup\{v\}}\left(\sum\limits_{e=(u,w)\in E}-\frac{q_e}{\sqrt{r_e}} \onev_w\right)$ 
(as $\BB^\top \frac{\qq}{\sqrt{\rr}}$ sums up to $\zerov$),
in $\vpi^C\left(\onev_v\right)$ there are no cancellations. 
The goal is to simply
estimate the hitting probabilities from $v$ to the vertices of $C$,
which can be done by sampling a number of random walks from $v$ to $C$.

As discussed before, even though ideally we would like to have an error bound of the form 
$\sqrt{\mathcal{E}_{\rr}(\tvpi^C(\onev_v) - \vpi^C(\onev_v))} \leq \delta_2 \sqrt{R_{eff}(C,v)}$, our analysis is only able to guarantee that for any fixed potential vector $\vphi$
with $E_{\rr}(\vphi) \leq 1$, with high probability
$\left|\langle \vphi, \tvpi^C(\onev_v) - \vpi^C(\onev_v)\rangle\right| \leq \delta_2 
\sqrt{R_{eff}(C,v)}$.
However, this is still sufficient for our purposes.

In Appendix~\ref{proof_conc2}
we prove the following general concentration inequality, which basically states that we can estimate the desired hitting
probabilities as long as we have a bound on the $\ell_2$ norm
of the potentials $\vphi$ weighted by the hitting probabilities.

\begin{lemma}[Concentration inequality 2]
Let $\vpi$ be a probability distribution over $[n]$
and $\tvpi$ an empirical distribution of $Z$ samples from $\vpi$.
For any $\ovphi\in\mathbb{R}^n$ with $\left\|\ovphi\right\|_{\vpi,2}^2 \leq \cV$, we have
\[\Pr\left[\left|\langle \tvpi - \vpi, \ovphi\rangle\right| > t\right] \leq \frac{1}{n^{100}} + \tO{\log \left(n \cdot \cV / t\right)}\exp\left(- \frac{Z t^2}{\tO{\cV \log^2 n}} \right) \,.\]
\label{conc2}
\end{lemma}

We will apply it for $\ovphi = \vphi - \phi_v\cdot \onev$, and it is important
to note that $\mathcal{E}_{\rr}(\ovphi) = \mathcal{E}_{\rr}(\vphi)$.
In order to get a bound on $\left\|\ovphi\right\|_{\vpi^C\left(\onev_v\right),2}^2$,
we use the following lemma,
which is proved in Appendix~\ref{proof_lem_variance}. 
\begin{lemma}[Bounding the second moment of potentials]
For any graph $G$, resistances $\rr$, potentials $\vphi$ with 
$E_{\rr}(\vphi) \leq 1$,
$C \subseteq V$ and $v\in V\backslash C$ we have
$
\left\|\vphi - \phi_v \onev\right\|_{\vpi^C(\onev_v),2}^2
\leq 8 \cdot R_{eff}(C,v)$.
\label{lem:variance}
\end{lemma}
To give some intuition on this, consider the case when $V = C\cup\{v\} = \{1,\dots,k\}\cup\{v\}$,
and there are edges $e_1,\dots,e_k$ between $C$ and $v$, one for each vertex of $C$.
Then, we have 
$\pi_i^C(\onev_v) = (r_{e_i})^{-1} / \sum\limits_{i=1}^k (r_{e_i})^{-1}$, and so
\[ 
\left\|\ovphi\right\|_{\vpi^C\left(\onev_v\right),2}^2 
= \sum\limits_{i=1}^k \frac{(\phi_i - \phi_v)^2}{r_{e_i}} \cdot \left(\sum\limits_{i=1}^k (r_{e_i})^{-1}\right)^{-1}
\leq \mathcal{E}_{\rr}(\ovphi) \cdot R_{eff}(C,v)
\leq R_{eff}(C,v)\,.
\]
We finally arrive at the desired statement about estimating $\vpi^C(\onev_v)$.
\begin{lemma}[Estimating $\vpi^C(\onev_v)$]
Consider a graph $G(V,E)$ with resistances $\rr$, 
a $\beta$-congestion reduction subset $C$, as well as $v\in V\backslash C$.
Then, for any $\delta_2 > 0$,
we can compute $\tvpi^{C}\left(\onev_v\right)\in\mathbb{R}^n$ such that with high probability
\begin{align}
    \left|\langle \vphi, \tvpi^{C}\left(\onev_{v}\right) - \vpi^{C}\left(\onev_{v}\right)\rangle\right| \leq \delta_2 \cdot \sqrt{R_{eff}(C,v)}\,,
    \label{eq:inner_product_bound}
\end{align} 
where $\vphi\in\mathbb{R}^n$ is a fixed vector with $E_{\rr}(\vphi) \leq 1$.
The algorithm computes $\tO{\frac{\log n}{\delta_2^2}}$ random walks from $v$ to $C$.
\label{estimate2}
\end{lemma}
\begin{proof}
Because both $\tvpi^C(\onev_v)$ and $\vpi^C(\onev_v)$
are probability distributions, the quantity (\ref{eq:inner_product_bound}) doesn't
change when a multiple of $\onev$ is added to $\vphi$,
and so we can replace it by $\ovphi = \vphi - \phi_v \onev$.

Now, $\tvpi^C\left(\onev_v\right)$ will be defined as the empirical hitting distribution
that results from sampling $Z$ random walks from $v$ to $C$. 
Directly applying the concentration bound in Lemma~\ref{conc2} and setting
$Z = \tO{\frac{\log n}{\delta_2^2}}$,
together with the fact that
$\left\|\ovphi\right\|_{\vpi^C(\onev_v),2}^2 
\leq 8\cdot R_{eff}(C,v)$ by Lemma~\ref{lem:variance} and
$\log \log R_{eff}(C,v) \leq O(\log\log n)$,
we get
\begin{align*}
& \Pr\left[\left|\langle \tvpi^C(\onev_v) - \vpi^{C}(\onev_v), \ovphi\rangle\right| > 
\delta_2 \cdot \sqrt{R_{eff}(C,v)}\right] < \frac{1}{n^{10}}\,.
\end{align*}
\end{proof}

\subsection{Proof of Lemma~\ref{lem:ds}}

We are now ready for the proof of Lemma~\ref{lem:ds}.
\begin{proof}[Proof of Lemma~\ref{lem:ds}]
Let $\DP$ be a demand projection data structure. We analyze its operations one by one.

\begin{algorithm} %
\begin{algorithmic}[1]
\caption{\textsc{DemandProjector} \DP.\textsc{Initialize} }
\Procedure{\DP.\textsc{Initialize}}{$C,\rr,\qq,S,\cP$}
\State Initialize $C,\rr,\qq,S,\cP$
\State $F = V\backslash C$
\State %
$h = \tO{\heps^{-4}\beta^{-6}+\heps^{-2}  \beta^{-4} \gamma^{-2}}$
\Comment{\#random walks for each pair $u\in V$, $e\in E$ with $u\in e$}
\State $t = 0$ \Comment{\#calls to \textsc{AddTerminal} since last call to $\textsc{UpdateFull}$}

\State $\vphi = \LL_{FF}^+ \left[\BB^\top \frac{\qq_S}{\sqrt{\rr}}\right]_F$
 \State $\tvpi^C\left(\BB^\top \frac{\qq_S}{\sqrt{\rr}}\right) = \left[\BB^\top \frac{\qq_S}{\sqrt{\rr}}\right]_C - \LL_{CF}\vphi$
\EndProcedure
\end{algorithmic} 
\end{algorithm} 
\noindent {\bf \DP.$\textsc{Initialize}(C,\rr,\qq,S,\mathcal{P})$:}
We initialize the values of $C,\rr,\qq,S,\cP$.
 
Then we exactly compute
the demand projection, i.e. $\tvpi^C\left(\BB^\top \frac{\qq_S}{\sqrt{\rr}}\right) = 
\vpi^C\left(\BB^\top \frac{\qq_S}{\sqrt{\rr}}\right)$,
which takes time $\tO{m}$ as shown in~\cite{gao2021fully}.
More specifically, we have
$
\vpi^C\left(\BB^\top \frac{\qq_S}{\sqrt{\rr}}\right)
=
\begin{pmatrix}\II & \LL_{CF} \LL_{FF}^{-1}\end{pmatrix} \BB^\top \frac{\qq_S}{\sqrt{\rr}}
$
which only requires applying the operators $\LL_{FF}^{-1}$ and $\LL_{CF}$.

\noindent{\bf \DP.\textsc{AddTerminal}($v$, $\tR_{eff}(C,v)$):}
We will serve this operation by applying Lemma~\ref{lem:insert1}. 
It is important to note that the error guarantee for the $\textsc{Output}$ procedure
increases with every call to $\textsc{AddTerminal}$, so in general we have a bounded budget
for the number of calls to thus procedure before having to call again $\textsc{Initialize}$.

We apply Lemma~\ref{lem:insert1} to obtain
$\tpi_v^{C\cup\{v\}}(\BB^\top \frac{\qq_S}{\sqrt{\rr}})$, 
and update the estimate  $\tvpi_v^{C\cup\{v\}}\left(\BB^\top \frac{\qq_S}{\sqrt{\rr}}\right)$.
The former can be achieved with $h = \widetilde{O}\left( \heps^{-4} \beta^{-6} + \heps^{-2} \beta^{-2} \gamma^{-2} \right)$ 
random walks. 
Note that these random walks are already stored in $\cP$, so accessing each of them takes time $\tO{1}$. Using the congestion reduction property of $C$, we see that the running time of the procedure, which is dominated by shortcutting the random  walks is $\tO{h\beta^{-2}}$, which gives the claimed bound.
The latter can be achieved with $h' =  \widetilde{O}\left( \heps^{-2} \beta^{-4} \gamma^{-2} \right)$ fresh random walks. Due to the congestion reduction
property, simulating each of these requires $\tO{\beta^{-2}}$ time.

\begin{algorithm} %
\begin{algorithmic}[1]
\caption{\textsc{DemandProjector} \DP.\textsc{Update} and \DP.\textsc{Output}}
\Procedure{\DP.\textsc{Update}}{$e,\rr',\qq'$}
\If {$e\in S$}
	\State $\tvpi^C\left(\BB^\top \frac{\qq_S}{\sqrt{\rr}}\right)
	= \tvpi^C\left(\BB^\top \frac{\qq_S}{\sqrt{\rr}}\right) 
	+ \left(\frac{q_e'}{\sqrt{r_e'}} - \frac{q_e}{\sqrt{r_e}}\right)\cdot\BB^\top \onev_e$
\EndIf
\State $q_e = q_e'$, $r_e = r_e'$
\EndProcedure
\Procedure{\DP.\textsc{Output}}{$ $}
\State \Return $\tvpi^C\left(\BB^\top \frac{\qq_S}{\sqrt{\rr}}\right)$ %
\EndProcedure
\end{algorithmic} 
\end{algorithm} 
\noindent{\bf \DP.\textsc{Update}($e,\rr',\qq'$):}
We update the values of $r_e,q_e$. We also update the projection, by noting that since $e\in E(C)$,
\begin{align*}
\vpi^C\left(\BB^\top \frac{\qq'}{\sqrt{\rr'}}\right)
=\vpi^C\left(\BB^\top \frac{\qq}{\sqrt{\rr}}\right)
+ \left(\frac{q_e'}{\sqrt{r_e'}} 
- \frac{q_e}{\sqrt{r_e}}\right) \BB^\top \onev_e\,,
\end{align*}
so we change
$\tvpi^C\left(\BB^\top \frac{\qq}{\sqrt{\rr}}\right)$ by the same amount,
which takes time $O(1)$ and does not introduce any additional error in our estimate.

\noindent{\bf \DP.\textsc{Output}():}
We output our estimate
$\tvpi^C\left(\BB^\top \frac{\qq_S}{\sqrt{\rr}}\right)$.
Per Lemma~\ref{lem:insert1} we see that each of the previous $T$ calls to $\textsc{AddTerminal}$
add an error to our estimate of at most $\heps$ in the sense that if $\vDelta^t$ were the true change in the demand projection at the $t^{th}$ insertion, and $\tvDelta^t$ were the update made to our estimate, then
\[
\left| \left\langle
\tvDelta^t - \vDelta^t, \vphi
\right\rangle \right| \leq\heps\,,
\]
w.h.p. for any fixed $\vphi$ such that $E_{\rr^t}(\vphi) \leq 1$, where $\rr^t$ represents the resistances when $t^{th}$ call to $\textsc{AddTerminal}$ is made. Equivalently, for any nonzero $\vphi$,
\[
\frac{1}{\sqrt{E_{\rr^t}(\vphi)}} \left| \left\langle
\tvDelta^t - \vDelta^t, \vphi
\right\rangle \right| \leq\heps\,,
\]

By the invariant satisfied by the resistances passed as parameters to the $\textsc{AddTerminal}$ routine, we have that 
$\rr^t \leq \alpha \cdot \rr^T$ for all $t$. Therefore $1/E_{\rr^T}(\vphi) \leq \alpha/E_{\rr^t}(\vphi)$.
So we have that 
\[
\frac{1}{\sqrt{E_{\rr^T}(\vphi)}}\left| \left\langle
\tvDelta^t - \vDelta^t, \vphi
\right\rangle \right| \leq\heps \cdot \sqrt{\alpha} \,.
\]
Summing up over $T$ insertions, we obtain the desired error bound.
Furthermore, note that returning the estimate takes time proportional to $\vert C \vert$, which is $\tO{\beta m + T}$.

\end{proof}

\newpage
\appendix

\section{Maintaining the Schur Complement}
\label{sec:maintain_schur}

Following the scheme from~\cite{gao2021fully} we maintain a dynamic Schur complement of the graph onto a subset of terminals $C$. The approach follows rather directly from~\cite{gao2021fully} and leverages the recent work of~\cite{bernstein2020fully} to dynamically maintain an edge sparsifier of the Schur complement of the graph onto $C$. Compared to~\cite{gao2021fully} we do not require a parameter that depends on the adaptivity of the adversary.
In addition, when adding a vertex to $C$ we also return a $(1+\eps)$-approximation of the effective resistance $R_{eff}(v,C)$, which gets returned by the function call.

\begin{lemma}[\textsc{DynamicSC} (Theorem 4, \cite{gao2021fully})]
There is a \textsc{DynamicSC} data structure supporting the following operations with the given runtimes
against oblivious adversaries, for constants $0 < \beta,\eps < 1$:
\begin{itemize}
\item{\textsc{Initialize}$(G, C^{\text{(init)}}, \rr, \eps, \beta)$:
Initializes a graph $G$ with resistances $\rr$ and a set of safe terminals $C^{\text{(safe)}}$.
Sets the terminal set $C = C^{\text{(safe)}} \cup C^{\text{(init)}}$.
Runtime: $\tO{m\beta^{-4} \eps^{-4}}$.
}
\item{\textsc{AddTerminal}$(v\in V(G))$:
Returns $\tR_{eff}(C,v) \approx_{2} R_{eff}(C,v)$ and adds $v$ as a terminal. Runtime: Amortized 
$\tO{\beta^{-2} \eps^{-2}}$. 
}
\item{\textsc{TemporaryAddTerminals}$(\Delta C\subseteq V(G))$:
Adds all vertices in the set $\Delta C$ as (temporary) terminals. 
Runtime: Worst case
$\tO{K^2 \beta^{-4} \eps^{-4}}$, where $K$
is the total number of terminals
added by all of the $\textsc{TemporaryAddTerminals}$ operations that have not been rolled back
using $\textsc{Rollback}$. All $\textsc{TemporaryAddTerminals}$ operations should be rolled back
before the next call to $\textsc{AddTerminals}$.
}
\item{\textsc{Update}$(e, r)$: Under the guarantee that 
both
endpoints of $e$ are terminals, updates $r_e = r$. Runtime: Worst case $\tO{1}$.
}
\item{$\tSC()$: Returns a spectral 
sparsifier $\tSC \approx_{1+\eps} SC(G,C)$ (with respect to resistances $\rr$) with
$\tO{|C|\eps^{-2}}$ edges.
Runtime: Worst case $\tO{\left(\beta m +
(K\beta^{-2}\eps^{-2})^2\right)\eps^{-2}}$ where $K$
is the total number of terminals
added by all of the $\textsc{TemporaryAddTerminals}$ operations that have not been rolled back.
}
\item{\textsc{Rollback}$()$: Rolls back the last $\textsc{Update}$,
$\textsc{AddTerminals}$, or \textsc{TemporaryAddTerminals} if it exists.
The runtime is the same as the original operation.
}
\end{itemize}
Finally, all calls return valid outputs with high probability. The size of $C$ should always be $O(\beta m)$.
\end{lemma}

This data structure is analyzed in detail in~\cite{gao2021fully}. Additionally, let us show that an approximation to $R_{eff}(v,C)$ can be efficiently computed along with the \textsc{AddTerminal} operation.
To get an estimate we simply inspect the neighbors of $v$ in the sparsified Schur complement of $C \cup \{v\}$ and compute the inverse of the sum of their inverses. This is indeed a $1+O(\epsilon)$-approximation, as effective resistances are preserved within a $1+O(\epsilon)$ factor in the sparsifier. 

To show that this operation takes little amortized time, we note that by the proof appearing in~\cite[Lemma 6.2]{gao2021fully}, vertex $v$ appears in amortized $\tO{1}$ expanders maintained dynamically. As the dynamic sparsifier keeps $\tO{\epsilon^{-2}}$ neighbors of $v$ from each expander, the number of neighbors to inspect with each call is $\tO{\epsilon^{-2}}$, which also bounds the time necessary to approximate the resistance.

\section{Auxiliary Lemmas}
\label{sec:aux}

\begin{replemma}{lem:sc-energy-bd}
Let $\dd$ be a demand vector,  let $\rr$ be resistances, and let $C \subseteq V$ be a subset of vertices.
Then 
\[
\mathcal{E}_{\rr}\left( \vpi^C(\dd) \right) \leq \mathcal{E}_{\rr}\left( \dd \right) \,.
\]
\end{replemma}
\begin{proof}
Letting $F = V\setminus C$, and $\LL$ be the Laplacian of the underlying graph, we can write
\[
\vpi^C(\dd) = \dd_C - \LL_{CF} \LL_{FF}^{-1} \dd_F\,.
\]
By factoring  $\LL^+$ as
\begin{align*}
\LL^+ =\left[\begin{array}{cc}
I & 0\\
-\LL_{FF}^{-1}\LL_{FC} & I
\end{array}\right]\left[\begin{array}{cc}
SC(\LL,C)^{+} & 0\\
0 & \LL_{FF}^{-1}
\end{array}\right]\left[\begin{array}{cc}
I & -\LL_{CF}\LL_{FF}^{-1}\\
0 & I
\end{array}\right]
\end{align*}
we can write
\begin{align*}
\mathcal{E}_{\rr}(\dd) = \dd^\top \LL^+ \dd  =
\left[\begin{array}{cc}
\vpi^C(\dd) \\
\dd_F
\end{array}\right]^\top \left[\begin{array}{cc}
SC(\LL,C)^{+} & 0\\
0 & \LL_{FF}^{-1}
\end{array}\right]\left[\begin{array}{cc}
\vpi^C(\dd) \\
\dd_F
\end{array}\right]
= \|\vpi^C(\dd)\|_{SC(\LL,C)^+}^2 + \|\dd_F\|_{\LL_{FF}^{-1}}^2\,.
\end{align*}
Furthermore, we can use the same factorization to write 
\begin{align*}
\mathcal{E}_{\rr}(\vpi^C(\dd)) =
\left[\begin{array}{cc}
\vpi^C(\dd) \\
0
\end{array}\right]^\top \left[\begin{array}{cc}
SC(\LL,C)^{+} & 0\\
0 & \LL_{FF}^{-1}
\end{array}\right]\left[\begin{array}{cc}
\vpi^C(\dd) \\
0
\end{array}\right]
= \|\vpi^C(\dd)\|_{SC(\LL,C)^+}^2\,,
\end{align*}
which proves the claim.
\end{proof}

\begin{lemma}
For any $\mu\in(1/\mathrm{poly}(m),\mathrm{poly}(m))$, we have 
$\left\|\rr(\mu)\right\|_\infty \leq m^{\tO{\log m}}$.
\end{lemma}
\begin{proof}

By Appendix A in~\cite{axiotis2020circulation}, for some
$\mu_0 = \Theta(\left\|\cc\right\|_2)$,
the solution $\ff = \uu / 2$ 
has 
\[
\left\|\CC^\top \left(\frac{\cc}{\mu_0} + 
\frac{\onev}{\ss^+} - \frac{\onev}{\ss^-}\right)\right\|_{(\CC^\top \RR \CC)^+} \leq 1/10\,.\] This implies that 
$\min_e \left\{s_e(\mu_0)^+,s_e(\mu_0)^-\right\}
\geq \min_e u_e / 4 \geq 1/4$, and so 
$\left\|\rr(\mu_0)\right\|_\infty \leq O(1)$.
Additinally,
$\left\|\cc\right\|_\infty \in\left[1,\mathrm{poly}(m)\right]$,
so
$\mu_0 = \Theta\left(\mathrm{poly}(m)\right)$.

Now, for any integer $i \geq 0$ we let 
$\mu_{i+1} = \mu_{i}\cdot (1 - 1/\sqrt{m})^{\sqrt{m}/10}$.
By Lemma~\ref{lem:resistance_stability2} we have that
$\rr\left(\mu_{i+1}\right) \approx_{m^2}
\rr\left(\mu_{i}\right)$, and so 
\[ \rr\left(\frac{1}{\mathrm{poly}(m)}\right)
= \rr(\mu_{\tO{\log m}}) 
\leq \left(\frac{9}{100} m^2\right)^{\tO{\log m}} \rr(\mu_0)
\leq m^{\tO{\log m}} \rr(\mu_0)
\leq m^{\tO{\log m}} \,. \]
\end{proof}

\begin{lemma}
Given a graph $G(V,E)$ with resistances $\rr$ and any parameter $\eps > 0$,
there exists an algorithm that runs in time $\tO{m/\eps^2}$ and produces
a matrix $\QQ\in\mathbb{R}^{\tO{1/\eps^2}\times n}$ such that 
with high probability
for any $u,v\in V$, 
\[ R_{eff}(u,v) \approx_{1+\eps} \left\|\QQ\onev_u - \QQ\onev_v\right\|_2^2\]
\label{lem:approx_effective_res}
\end{lemma}

\section{Deferred Proofs from Section~\ref{sec:ipm}}
\subsection{Central path stability bounds}
\label{proof_stability_bounds}
\begin{lemma}[Central path energy stability]
Consider a minimum cost flow instance on a graph $G(V,E)$.
For any $\mu > 0$ and 
$\mu' = \mu / (1+1/\sqrt{m})^k$ for some $k \in (0,\sqrt{m}/10)$, we have
\begin{align*}
\sum\limits_{e\in E} \left(\frac{1}{s_e(\mu)^+\cdot s_e(\mu')^+} + \frac{1}{s_e(\mu)^-\cdot s_e(\mu')^-}\right) \left(f_e(\mu') - f_e(\mu)\right)^2
\leq 2k^2
\end{align*}
\label{lem:energy_stability}
\end{lemma}

\begin{proof}[Proof of Lemma~\ref{lem:energy_stability}]
We let $\delta = 1/\sqrt{m}$, 
$\ff = \ff(\mu)$, $\ss = \ss(\mu)$,
$\rr = \rr(\mu)$, $\ff' = \ff(\mu')$, $\ss'= \ss(\mu')$,
and
$\rr' = \rr(\mu')$.
We also set $\tff = \ff' - \ff$.
By definition of centrality we have
\begin{align*}
& \CC^\top \left(\frac{1}{\ss^-} - \frac{1}{\ss^+}\right) = \CC^\top \frac{\cc}{\mu}\\
& \CC^\top \left(\frac{1}{\ss^- + \tff} - \frac{1}{\ss^+ - \tff}\right) = \CC^\top \frac{\cc}{\mu'}\,,
\end{align*}
which, after subtracting, give
\begin{align*}
& \CC^\top \left(\frac{1}{\ss^- + \tff} - \frac{1}{\ss^-} - \frac{1}{\ss^+ - \tff} + \frac{1}{\ss^+}\right) = \CC^\top \left(\frac{\cc}{\mu'} - \frac{\cc}{\mu}\right)\\
& \Leftrightarrow \CC^\top \left(\left(\frac{1}{\ss^-(\ss^- + \tff)}  +\frac{1}{\ss^+(\ss^+ - \tff)}\right)\tff\right) 
= -\left((1 + \delta)^k - 1\right)\CC^\top \frac{\cc}{\mu}\,.
\end{align*}
As $\tff = \CC\xx$ for some $\xx$, after taking the inner product of both sides with $\xx$ we get
\begin{align}
& \left\langle \tff, \left(\frac{1}{\ss^-(\ss^- + \tff)}  +\frac{1}{\ss^+(\ss^+ - \tff)}\right)\tff\right\rangle = -\left((1+\delta)^k - 1\right) \left\langle\frac{\cc}{\mu},\tff\right\rangle\,.
\label{eq:energy_stable_prelim}
\end{align}

We will now prove that $-\left\langle\frac{\cc}{\mu},\tff\right\rangle \leq k \sqrt{m}$.
First of all, by differentiating the centrality condition 
\[ \CC^\top \left(\frac{\cc}{\nu} + \frac{\onev}{\ss(\nu)^+} - \frac{\onev}{\ss(\nu)^-}\right) = \zerov\]
with respect to $\nu$ we get
\[ \CC^\top \left(-\frac{\cc}{\nu^2} + 
\left(\frac{\onev}{(\ss(\nu)^+)^2} + \frac{\onev}{(\ss(\nu)^-)^2}\right) \frac{d\ff(\nu)}{d\nu}
\right) = \zerov \,,\]
or equivalently
\[ \CC^\top \left(\rr(\nu) \frac{d\ff(\nu)}{d\nu}\right) = 
-\frac{1}{\nu} \CC^\top \left(\frac{\onev}{\ss(\nu)^+} - \frac{\onev}{\ss(\nu)^-}\right) \,.\]
If we set $g(\ss) = \frac{\frac{\onev}{\ss^+}-\frac{\onev}{\ss^-}}{\rr}$, this can also be equivalently
written as
\[ \frac{d\ff(\nu)}{d\nu} = 
-\frac{1}{\nu} \left(g(\ss(\nu)) - (\RR(\nu))^{-1}\BB (\BB^\top (\RR(\nu)^{-1})\BB)^+ \BB^\top g(\ss(\nu))\right)\,.\]
We have 
\begin{align*}
-\left\langle\frac{\cc}{\mu},\tff\right\rangle
& = -\int_{\nu=\mu}^{\mu'}\left\langle\frac{\cc}{\mu}, d\ff(\nu) \right\rangle\\
& = \frac{1}{\mu}\int_{\nu=\mu}^{\mu'}\left\langle\frac{\nu}{\ss(\nu)^-} - \frac{\nu}{\ss(\nu)^+}, \frac{1}{\nu}
\left(g(\ss(\nu)) - (\RR(\nu))^{-1}\BB(\BB^\top (\RR(\nu))^{-1} \BB)^+ \BB^\top g(\ss(\nu))\right) \right\rangle d\nu\\
& = -\frac{1}{\mu}\int_{\nu=\mu}^{\mu'}\left\langle \sqrt{\rr(\nu)} g(\ss(\nu)),
\vPi_{\mathrm{ker}(\BB^\top (\RR(\nu))^{-1/2})} \sqrt{\rr(\nu)}g(\ss(\nu)) \right\rangle d\nu\\
& = \frac{1}{\mu}\int_{\nu=\mu'}^{\mu}\left\|\vPi_{\mathrm{ker}(\BB^\top (\RR(\nu))^{-1/2})} \sqrt{\rr(\nu)}g(\ss(\nu)) \right\|_2^2 d\nu\\
& \leq \frac{1}{\mu}\int_{\nu=\mu'}^{\mu}\left\|\sqrt{\rr(\nu)}g(\ss(\nu)) \right\|_2^2 d\nu\\
& \leq \frac{1}{\mu}\int_{\nu=\mu'}^{\mu} m d\nu\\
& = m \frac{\mu - \mu'}{\mu}\\
& =  m(1 - (1+\delta)^{-k})\\
& \leq \delta k m\\
&  =  k \sqrt{m}\,,
\end{align*}
where $\vPi_{\mathrm{ker}(\BB^\top(\RR(\nu))^{-1/2})}
 =\II - 
(\RR(\nu))^{-1/2}\BB
 (\BB^\top (\RR(\nu))^{-1}\BB)^+\BB^\top (\RR(\nu))^{-1/2}$ is
the orthogonal projection onto the kernel of $\BB^\top (\RR(\nu))^{-1/2}$.

Plugging this into (\ref{eq:energy_stable_prelim}) and using the fact
that $(1+\delta)^k \leq 1+1.1\delta k = 1 + 1.1 k / \sqrt{m}$, we get
\begin{align*}
\sum\limits_{e\in E} \left(\frac{1}{s_e(\mu)^+\cdot s_e(\mu')^+} + \frac{1}{s_e(\mu)^-\cdot s_e(\mu')^-}\right) \left(f_e(\mu') - f_e(\mu)\right)^2
\leq 2k^2\,.
\end{align*}
\end{proof}

We give an auxiliary lemma which converts between different kinds of slack approximations.
\begin{lemma}
We consider flows $\ff,\ff'$ with slacks $\ss,\ss'$ and resistances $\rr,\rr'$.
Then,
\[ \max\left\{\left|\frac{s_e'^+-s_e^+}{s_e^+}\right|,\left|\frac{s_e'^--s_e^-}{s_e^-}\right|\right\}\leq \sqrt{r_e} \left|f_e' - f_e\right| \leq \sqrt{2}\max\left\{\left|\frac{s_e'^+-s_e^+}{s_e^+}\right|,\left|\frac{s_e'^--s_e^-}{s_e^-}\right|\right\}\]
and if $r_e\not\approx_{1+\gamma} r_e'$ for some $\gamma\in(0,1)$,
then $\sqrt{r_e}\left|f_e' - f_e\right| \geq \gamma/6$.
\label{lem:approx_conversion}
\end{lemma}
\begin{proof}
For the first one, note that
\begin{align*}
r_e = \frac{1}{(s_e^+)^2} + \frac{1}{(s_e^-)^2} 
\in\left[\max\left\{\frac{1}{(s_e^+)^2} ,\frac{1}{(s_e^-)^2}\right\}, 2\max\left\{\frac{1}{(s_e^+)^2}, \frac{1}{(s_e^-)^2}\right\}\right]\,.
\end{align*}
Together with the fact that 
$\left|f_e'-f_e\right| = \left|s_e'^+-s_e^+\right| = \left|s_e'^--s_e^-\right|$, it implies
the first statement.

For the second one, 
without loss of generality let 
$s_e^+ \leq s_e^-$, so by the previous statement
we have $\sqrt{r_e}\left|f_e' - f_e\right|
\geq \frac{\left|s_e'^+ - s_e^+\right|}{s_e^+}$.
If this is $<\gamma/6$ then
$(1-\gamma/6) s_e^+ \leq s_e'^+ \leq (1 + \gamma/6) s_e^+$,
so $s_e'^+ \approx_{1+\gamma/3} s_e^+$.
However, we also have that
$\frac{\left|s_e'^- - s_e^-\right|}{s_e^-} \leq
\frac{\left|s_e'^+ - s_e^+\right|}{s_e^+}
\leq \gamma /6$,
so $s_e'^- \approx_{1+\gamma/3} s_e^-$.
Therefore, 
$r_e' = \frac{1}{(s_e'^+)^2} + \frac{1}{(s_e'^-)^2} 
\approx_{1+\gamma} 
\frac{1}{(s_e^+)^2} + \frac{1}{(s_e^-)^2}  = 
r_e$, a contradiction.
\end{proof}
The following lemma is a fine-grained explanation of how resistances can change.
\begin{lemma}
Consider a minimum cost flow instance on a graph $G(V,E)$ and
parameters $\mu > 0$ and $\mu' \geq \mu / (1+1/\sqrt{m})^{k}$, where $k \in(0,\sqrt{m} / 10)$.
For any $e\in E$ and $\gamma\in(0,1)$
we let $\mathrm{change}(e,\gamma)$ be the largest integer $t(e)\geq 0$
such that there are real numbers
$\mu = \mu_{1}(e) > 
\mu_{2}(e) > 
\dots > 
\mu_{t(e)+1}(e) = \mu'$
with 
$\sqrt{r_e(\mu_i)} \left|f_e(\mu_{i+1}) - f_e(\mu_i)\right| \geq \gamma$
for all $i\in[t(e)]$.

Then,
$\sum\limits_{e\in E} \left(\mathrm{change}(e,\gamma)\right)^2 \leq O(k^2 / \gamma^2)$.

\label{lem:resistance_stability1}
\end{lemma}
\begin{proof}[Proof of Lemma~\ref{lem:resistance_stability1}]
First, we assume that without loss of generality,
$r_e(\mu_{i+1}(e)) \approx_{(1+6\gamma)^2} r_e(\mu_i(e))$
for all $e\in E$ and $i\in[t(e)]$.
If this is not true, then by continuity
there exists a $\nu\in(\mu_{i+1},\mu_i)$ such that
$r_e(\mu_{i+1}(e)) \not\approx_{1+6\gamma} r_e(\nu)$
and $r_e(\nu) \not\approx_{1+6\gamma} r_e(\mu_i(e))$.
By Lemma~\ref{lem:approx_conversion}, this implies that
$\sqrt{r_e(\mu_{i+1}(e))} \left|f_e(\nu) - f_e(\mu_{i+1}(e))\right| \geq \gamma$
and
$\sqrt{r_e(\nu)} \left|f_e(\mu_{i}(e)) - f_e(\nu)\right| \geq \gamma$.
Therefore we can break the interval $(\mu_{i+1},\mu_i)$ into
$(\mu_{i+1},\nu)$
and $(\nu,\mu_{i})$ and make the statement stronger.

Similarly, we also assume that 
$r_e(\nu) \approx_{(1+6\gamma)^3} r_e(\mu_i(e))$
for all $e\in E$, $i\in[t(e)]$, and $\nu\in(\mu_{i+1},\mu_i)$. If this is
not the case, then by using the fact that
$r_e(\mu_{i+1}(e)) \approx_{(1+6\gamma)^2} r_e(\mu_i(e))$, we also get that
$r_e(\mu_{i+1}) \not \approx_{1+6\gamma} r_e(\nu)$,
and so we can again break the interval as before and obtain a stronger statement.

Now, we look at the following integral:
\[ \mathcal{E} := \int_{\nu=\mu}^{\mu'} \sum\limits_{e\in E} r_e(\nu)
\left(\frac{df_e(\nu)}{d\nu}\right)^2 \left|d\nu\right|\,, \]
where $df_e(\nu)$ is the differential of the flow $f_e(\nu)$ with respect to
the centrality parameter.
Similarly to Lemma~\ref{lem:energy_stability}, we use the following equation
that describes how the flow changes:
\[ \frac{d\ff(\nu)}{d\nu} = 
-\frac{1}{\nu} \left(g(\ss(\nu)) - (\RR(\nu))^{-1}\BB (\BB^\top (\RR(\nu)^{-1})\BB)^+ \BB^\top g(\ss(\nu))\right)\,.\]
This implies that
\begin{align*}
\left\|\sqrt{\rr(\nu)} \frac{d\ff(\nu)}{d\nu}\right\|_2^2 
& =
\frac{1}{\nu^2} \left\|\sqrt{\rr(\nu)}g(\ss(\nu)) - (\RR(\nu))^{-1/2}\BB (\BB^\top (\RR(\nu)^{-1})\BB)^+ \BB^\top g(\ss(\nu))\right\|_2^2\\
& \leq 
\frac{1}{\nu^2} \left\|\left(I - (\RR(\nu))^{-1/2}\BB (\BB^\top (\RR(\nu)^{-1})\BB)^+ \BB^\top(\RR(\nu))^{-1/2}\right)
\sqrt{\rr}g(\ss(\nu))\right\|_2^2\\
& \leq 
\frac{1}{\nu^2} \left\|
\sqrt{\rr}g(\ss(\nu))\right\|_2^2\\
&\leq \frac{m}{\nu^2}\,,
\end{align*}
and so
\begin{align}
 \mathcal{E} \leq 
\int_{\nu=\mu}^{\mu'} \frac{m}{\nu^2} \left|d\nu\right|
= m \left(\frac{1}{\mu'} - \frac{1}{\mu}\right)
\leq m \frac{1.1\delta k}{\mu}
= 1.1 k \sqrt{m} / \mu\,.
\label{eq:integral_upper_bound}
\end{align}
On the other hand, for any $e\in E$ and $i\in[t(e)]$ we have
\begin{align*}
\int_{\nu=\mu_i(e)}^{\mu_{i+1}(e)} r_e(\nu) \left(\frac{df_e(\nu)}{d\nu}\right)^2 \left|d\nu\right|
& \geq 
\frac{r_e(\mu_i(e))}{(1+6\gamma)^3} 
\int_{\nu=\mu_i(e)}^{\mu_{i+1}(e)} \left(\frac{df_e(\nu)}{d\nu}\right)^2 \left|d\nu\right| \\
& \geq 
\frac{r_e(\mu_i(e))}{(1+6\gamma)^3} 
\frac{\left(\int_{\nu=\mu_i(e)}^{\mu_{i+1}(e)}
\left|\frac{df_e(\nu)}{d\nu}\right| \left|d\nu\right|\right)^2}
{
\int_{\nu=\mu_i(e)}^{\mu_{i+1}(e)} \left|d\nu\right|
}\\
& = 
\frac{r_e(\mu_i(e))}{(1+6\gamma)^3(\mu_i(e) - \mu_{i+1}(e))} 
\left(f(\mu_i(e)) - f(\mu_{i+1}(e))\right)^2\\
& \geq 
\frac{\gamma^2}{36(1+6\gamma)^3(\mu_i(e) - \mu_{i+1}(e))}\,,
\end{align*}
where we used the Cauchy-Schwarz inequality.
Now, note that 
\begin{align*}
\int_{\nu=\mu_1(e)}^{\mu_{t(e)+1}(e)} r_e(\nu) \left(\frac{df_e(\nu)}{d\nu}\right)^2 \left|d\nu\right|
& \geq 
\sum\limits_{i=1}^{t(e)} \frac{\gamma^2}{(1+6\gamma)^3(\mu_i(e) - \mu_{i+1}(e))}\\
& \geq 
\frac{\gamma^2 (t(e))^2}{(1+6\gamma)^3(\mu - \mu')}\\
& \geq 
\frac{\gamma^2 (t(e))^2\sqrt{m}}{(1+6\gamma)^3 k\mu}\,,
\end{align*}
where remember that $t(e) = \mathrm{change}(e,\gamma)$
and we again used Cauchy-Schwarz.
Summing this up for all $e\in E$ and combining with (\ref{eq:integral_upper_bound}),
we get that
$\sum\limits_{e\in E} (\mathrm{change}(e,\gamma))^2 \leq O(k^2 / \gamma^2)$.
\end{proof}

\begin{lemma}[Central path $\ell_\infty$ slack stability]
Consider a minimum cost flow instance on a graph $G(V,E)$.
For any $\mu > 0$ and 
$\mu' = \mu / (1+1/\sqrt{m})^k$ for some $k \in (0,\sqrt{m}/10)$, we have
\begin{align*}
\ss(\mu') \approx_{3k^2} \ss(\mu)\,.
\end{align*}
\label{lem:resistance_stability2}
\end{lemma}
\begin{proof}[Proof of Lemma~\ref{lem:resistance_stability2}]
By Lemma~\ref{lem:energy_stability}, for any $e\in E$ we have that
\begin{align}
\left(\frac{1}{s_e(\mu)^+\cdot s_e(\mu')^+} + \frac{1}{s_e(\mu)^-\cdot s_e(\mu')^-}\right) \left(f_e(\mu') - f_e(\mu)\right)^2
\leq 2k^2\,.
\label{eq:edge_energy_stability}
\end{align}

If $s_e(\mu')^+ = (1+c)\cdot s_e(\mu)^+$ for some $c\geq 0$, then 
\[ (f_e(\mu') - f_e(\mu))^2 = c^2 (s_e(\mu)^+)^2\] 
and
\[s_e(\mu)^+\cdot s_e(\mu')^+ = (1+c) (s_e(\mu)^+)^2\,,\]
so by (\ref{eq:edge_energy_stability}) we have that $c \leq 3k^2$.

Similarly, if $s_e(\mu')^+ = (1+c)^{-1} \cdot s_e(\mu)^+$ for some $c \geq 0$, then
\[ (f_e(\mu') - f_e(\mu))^2 = c^2 (s_e(\mu')^+)^2\]
and 
\[ s_e(\mu)^+\cdot s_e(\mu')^+ = (1+c) (s_e(\mu')^+)^2\,,\]
so by (\ref{eq:edge_energy_stability}) we have that $c \leq 3k^2$.

We have proved that $s_e(\mu')^+\approx_{3k^2} s_e(\mu)^+$ and by symmetry we also have
$s_e(\mu')^- \approx_{3k^2} s_e(\mu)^-$.
\end{proof}

\subsection{Proof of Lemma~\ref{lem:approx_central}}
\label{sec:proof_lem_approx_central}

Our goal is to keep track of
how close $\ff^*$ remains to centrality (in $\ell_2$ norm) and 
how close $\ff$ remains to $\ff^*$ in $\ell_\infty$ norm.
From these two we can conclude that at all times $\ff$ is close in $\ell_\infty$
to the central flow.
We first prove the following lemma, which bounds how the distance
of $\ff^*$ to centrality (measured in energy of the residual)
degrades when taking a progress step.
\begin{lemma}
Let $\ff^*$ be a flow with slacks $\ss^*$ and resistances $\rr^*$, %
and $\ff$ be a flow with slacks $\ss$ and resistances $\rr$,
where $\ss\approx_{1+\es} \ss^*$ for some $\es \in(0,0.1)$.
We define 
$\ff'^* = \ff^* + \e\tff^*$ 
for some $\e \in(0,0.1)$ (and the new slacks $\ss'^*$), where
\begin{align}
& \tff^* = \delta g(\ss) - \delta \RR^{-1}\BB(\BB^\top \RR^{-1} \BB)^+ \BB^\top g(\ss) \,, 
\label{eq:ipm_step}
\end{align}
$\delta = \frac{1}{\sqrt{m}}$,
and $g(\ss) := \frac{\frac{1}{\ss^+} - \frac{1}{\ss^-}}{\rr}$.
If we let $\hh = \frac{\cc}{\mu} + \frac{1}{\ss^{*+}} - \frac{1}{\ss^{*-}}$
and
$\hh' = \frac{\cc(1+\e\delta)}{\mu} + \frac{1}{\ss'^{*+}} - \frac{1}{\ss'^{*-}}$
be the residuals of $\ff^*$ and $\ff'^*$ for some $\mu > 0$,
then
\[ 
\left\|\CC^\top \hh'\right\|_{\oHH^+} \leq 
(1+\e\delta) 
\left\|\CC^\top \hh\right\|_{\oHH^+}
+ 
5 \left\|\frac{\rr^*}{\orr}\right\|_\infty^{1/2} \es \cdot \e +
2\left\|\frac{\rr'^*}{\orr}\right\|_\infty^{1/2} \e^2
\,,
\] where $\orr$ are some arbitrary resistances and 
$\oHH = \CC^\top \overline{\RR} \CC$.
\label{lem:ipm_step}
\end{lemma}
\begin{proof}
Let $\vrho^+ = \e \tff^* / \ss^{*+}$ and $\vrho^- = -\e\tff^* /\ss^{*-}$.
First of all, it is easy to see that
\begin{align*}
\left\|\vrho\right\|_2
&\leq \left\|\frac{\ss}{\ss^*}\right\|_\infty \left\|\frac{\ss^*}{\ss}\vrho\right\|_2\\
&\leq \e(1+\es)\left\|\tff^*\right\|_{\rr,2}\\
& = \e\delta(1+\es) \left\|\sqrt{\rr}g(\ss) - \RR^{-1/2} \BB (\BB^\top \RR^{-1}\BB)^+ \BB^\top g(\ss)\right\|_{2}\\
& = \e\delta (1+\es)\left\|\left(\II - \RR^{-1/2} \BB (\BB^\top \RR^{-1}\BB)^+ \BB^\top\RR^{-1/2}\right) \sqrt{\rr}g(\ss)\right\|_{2}\\
& \leq \e\delta (1+\es)\left\|\sqrt{\rr}g(\ss)\right\|_{2}\\
& = \e\delta (1+\es)\left\|\frac{\frac{1}{\ss^+} - \frac{1}{\ss^-}}{\sqrt{\frac{1}{(\ss^+)^2} + \frac{1}{(\ss^-)^2}}}\right\|_{2}\\
& \leq \e \delta (1+\es)\sqrt{m}\\
& = \e(1+\es)\,.
\end{align*}
We bound the energy to route the residual of $\ff'^*$ as
\begin{align*}
& \left\|\CC^\top \hh'\right\|_{\oHH^+} \\
& =  \left\|\CC^\top \hh + \CC^\top \left(\frac{\e\delta \cc}{\mu} + \frac{1}{\ss'^{*+}} - \frac{1}{\ss^{*+}} - \frac{1}{\ss'^{*-}} + \frac{1}{\ss^{*-}}\right) \right\|_{\oHH^{+}}\\
& =  \left\|\CC^\top \hh + 
\CC^\top \left(\frac{\e\delta \cc}{\mu} + 
\frac{\vrho^+}{\ss'^{*+}} - \frac{\vrho^-}{\ss'^{*-}} 
\right) \right\|_{\oHH^{+}}\\
& =  \left\|\CC^\top \hh + \CC^\top \left(\frac{\e\delta \cc}{\mu} + 
\frac{\vrho^+}{\ss^{*+}} - \frac{\vrho^-}{\ss^{*-}}\right) + \CC^\top\left( \frac{(\vrho^+)^2}{\ss'^{*+}}  - \frac{(\vrho^-)^2}{\ss'^{*-}}\right)\right\|_{\oHH^{+}}\,.
\end{align*}
Now, using (\ref{eq:ipm_step}) we get that
$\rr \tff^* = \delta \rr g(\ss) - \delta \BB (\BB^\top \RR^{-1} \BB)^+ \BB^\top g(\ss)$
and so
$\CC^\top \left(\rr \tff^*\right) 
= \delta \CC^\top \left(\rr g(\ss)\right)$, which follows by the fact that for any $i$,
$\onev_i^\top \CC^\top \BB = \left(\BB^\top \CC \onev_i\right)^\top = \zerov$, since $\CC \onev_i$ 
is a circulation by definition of $\CC$.
As $\rr \tff^* = 
\left(\frac{1}{(\ss^+)^2} + \frac{1}{(\ss^-)^2}\right) \tff^* 
= \e^{-1} \frac{\ss^{*+}}{(\ss^+)^2}\vrho^+ - \e^{-1}\frac{\ss^{*-}}{(\ss^-)^2} \vrho^-
$, we have
$\e \delta \CC^\top \left(\rr g(\ss)\right) = \CC^\top \left(\frac{\ss^{*+}}{(\ss^+)^2}\vrho^+ - 
\frac{\ss^{*-}}{(\ss^-)^2}\vrho^-\right)$ and so
\begin{align*}
& \left\|\CC^\top \hh + \CC^\top \left(\frac{\e\delta \cc}{\mu} + 
\frac{\vrho^+}{\ss^{*+}} - \frac{\vrho^-}{\ss^{*-}}\right) + \CC^\top\left( \frac{(\vrho^+)^2}{\ss'^{*+}}  - \frac{(\vrho^-)^2}{\ss'^{*-}}\right)\right\|_{\oHH^{+}}\\
& =\left\|\CC^\top \hh + \CC^\top \left(
\frac{\e \delta \cc}{\mu} + \e\delta \rr g(\ss) -\frac{\ss^{*+}}{(\ss^{+})^2}\vrho^+ + \frac{\ss^{*-}}{(\ss^{-})^2} \vrho^-
+ \frac{\vrho^+}{\ss^{*+}} - \frac{\vrho^-}{\ss^{*-}}\right) + \CC^\top\left( \frac{(\vrho^+)^2}{\ss'^{*+}}  -\frac{(\vrho^-)^2}{\ss'^{*-}}\right)\right\|_{\oHH^{+}}\\
& =  \Big\|\CC^\top \hh 
+ \e \delta \CC^\top \left(\frac{\cc}{\mu} + \rr g(\ss)\right)
+ \CC^\top \left(\left(
\onev - \left(\frac{\ss^{*+}}{\ss^+}\right)^2\right) \frac{\vrho^+}{\ss^{*+}}
- \left(\onev - \left(\frac{\ss^{*-}}{\ss^-}\right)^2\right) \frac{\vrho^-}{\ss^{*-}}\right)\\
& + \CC^\top\left( \frac{(\vrho^+)^2}{\ss'^{*+}} - \frac{(\vrho^-)^2}{\ss'^{*-}}\right)\Big\|_{\oHH^{+}}\\
& \leq (1+\e\delta)\left\|\CC^\top \hh\right\|_{\oHH^{+}}
+ 5 \left\|\frac{\rr^*}{\orr}\right\|_\infty^{1/2} \es \cdot \e +
2 \left\|\frac{\rr'^{*}}{\orr}\right\|_\infty^{1/2} \e^2
\end{align*}
where we have used the triangle inequality, the fact that 
\begin{align*}
& \e\delta \left\|\CC^\top \left(\frac{\cc}{\mu} + \rr g(\ss)\right)\right\|_{\oHH^{+}}\\
& = \e\delta \left\|\CC^\top \left(\frac{\cc}{\mu} + \frac{1}{\ss^+} - \frac{1}{\ss^-}\right)\right\|_{\oHH^{+}}\\
& \leq \e\delta \left\|\CC^\top \left(\frac{\cc}{\mu} + \frac{1}{\ss^{*+}} - \frac{1}{\ss^{*-}}\right)\right\|_{\oHH^{+}}
+\e\delta \left\|\CC^\top \left(\frac{1}{\ss^+} - \frac{1}{\ss^{*+}} - \frac{1}{\ss^-} + \frac{1}{\ss^{*-}}\right)\right\|_{\oHH^{+}}\\
& \leq \e\delta \left\|\CC^\top \left(\frac{\cc}{\mu} + \frac{1}{\ss^{*+}} - \frac{1}{\ss^{*-}}\right)\right\|_{\oHH^{+}}
+\e\delta 
\left\|\frac{\rr^*}{\orr}\right\|_\infty^{1/2} \left\|\CC^\top \left(\frac{1}{\ss^+} - \frac{1}{\ss^{*+}} - \frac{1}{\ss^-} + \frac{1}{\ss^{*-}}\right)\right\|_{(\CC^\top \RR^* \CC)^{+}}\\
& \leq \e\delta \left\|\CC^\top \hh\right\|_{\oHH^{+}}
+ \e\delta 
\left\|\frac{\rr^*}{\orr}\right\|_\infty^{1/2}
\left\|
\frac{\frac{1}{\ss^+} - \frac{1}{\ss^{*+}} - \frac{1}{\ss^-} + \frac{1}{\ss^{*-}}}{
	\left(\frac{1}{(\ss^{*+})^2} + \frac{1}{(\ss^{*-})^2}\right)^{1/2}}\right\|_2\\
& \leq \e\delta \left\|\CC^\top \hh\right\|_{\oHH^{+}}
+ \e\delta 
\left\|\frac{\rr^*}{\orr}\right\|_\infty^{1/2}\left(
\left\|
\frac{\frac{1}{\ss^+} - \frac{1}{\ss^{*+}}}{
	\left(\frac{1}{(\ss^{*+})^2}\right)^{1/2}}\right\|_2 + 
\left\|
\frac{\frac{1}{\ss^-} - \frac{1}{\ss^{*-}}}{
	\left(\frac{1}{(\ss^{*-})^2}\right)^{1/2}}\right\|_2\right)\\
& =\e\delta \left\|\CC^\top \hh\right\|_{\oHH^{+}}
+ \e\delta 
\left\|\frac{\rr^*}{\orr}\right\|_\infty^{1/2}
\left(
\left\|
\frac{\ss^{*+}}{\ss^{+}} - \onev\right\|_2 +
\left\|
	\frac{\ss^{*-}}{\ss^{-}} - \onev
	\right\|_2\right)\\
& \leq\e\delta \left\|\CC^\top \hh\right\|_{\oHH^{+}}
+ \e\delta 
\left\|\frac{\rr^*}{\orr}\right\|_\infty^{1/2} 2\es\sqrt{m}
\text{ (as $\ss \approx_{1+\es} \ss^*$) }
\\
& =\e\delta \left\|\CC^\top \hh\right\|_{\oHH^{+}}
+ 2
\left\|\frac{\rr^*}{\orr}\right\|_\infty^{1/2} \e\es
\end{align*}
and similarly, using 
$\left\|1 - \left(\frac{\ss^{*}}{\ss}\right)^2\right\|_\infty
\leq \es(2+\es)$, the fact that
\begin{align*}
& \Big\|\CC^\top \left(\left(
\onev - \left(\frac{\ss^{*+}}{\ss^+}\right)^2\right) \frac{\vrho^+}{\ss^{*+}}
- \left(\onev - \left(\frac{\ss^{*-}}{\ss^-}\right)^2\right) \frac{\vrho^-}{\ss^{*-}}\right)
+ \CC^\top\left( \frac{(\vrho^+)^2}{\ss'^{*+}}  +\frac{(\vrho^-)^2}{\ss'^{*-}}\right)\Big\|_{\oHH^{+}}\\
& \leq 
\left\|\frac{\rr^*}{\orr}\right\|_\infty^{1/2} \es(2+\es)\left\|\vrho\right\|_2 + \left\|\frac{\rr'^*}{\orr}\right\|_\infty^{1/2} \left\|\vrho\right\|_4^2\\
& \leq 
\left\|\frac{\rr^*}{\orr}\right\|_\infty^{1/2} \e\es(2+\es)(1+\es) + 
\left\|\frac{\rr'^*}{\orr}\right\|_\infty^{1/2} \e^2 (1+\es)^2\\
& \leq 
3 \left\|\frac{\rr^*}{\orr}\right\|_\infty^{1/2} \e \es +
2 \left\|\frac{\rr'^*}{\orr}\right\|_\infty^{1/2} \e^2\,.
\end{align*}
\end{proof}

We will also use the following lemma, which is standard~\cite{axiotis2020circulation}.

\begin{lemma}[Small residual implies $\ell_\infty$ closeness]
Given a flow $\ff = \ff^0 + \CC \xx$ with slacks $\ss$ and resistances $\rr$,
if $\left\|\CC^\top\left(\frac{\cc}{\mu} + \frac{1}{\ss^+} - \frac{1}{\ss^-}\right)\right\|_{(\CC^\top \RR \CC)^+} \leq 1/1000$ 
then $\ff$ is $(\mu,1.01)$-central.
\label{lem:small_residual_linfty}
\end{lemma}

Applying Lemma~\ref{lem:ipm_step} for $T=\frac{k}{\e}$ iterations, we get the lemma below, 
which measures the closeness of $\ff^*$ to the central path in $\ell_2$ after $T$ iterations.
\begin{lemma}[Centrality of $\ff^*$]
Let $\ff^{*1},\dots,\ff^{*T+1}$
be flows with slacks $\ss^{*1},\dots,\ss^{*T+1}$
and resistances $\rr^{*1},\dots,\rr^{*T+1}$,
and $\ff^{1},\dots,\ff^{T+1}$
be flows with slacks $\ss^{1},\dots,\ss^{T+1}$ and resistances $\rr^{1},\dots,\rr^{T+1}$,
such that $\ss^t \approx_{1+\es}\ss^{*t}$
for all $t\in [T]$,
where $T = \frac{k}{\e}$ for some $k\leq \sqrt{m} / 10$,
$\e\in(0,0.1)$ and $\es\in(0,0.1)$.
Additionally, we have that
\begin{itemize}
\item $\ff^{*1}$ is $\mu$-central
\item {For all $t\in[T]$, $\ff^{*t+1} = \ff^{*t} + \e \cdot \tff^t$, where
\[ \left\|\sqrt{\rr^t}\left(\tff^{*t} - \tff^{t}\right) \right\|_\infty \leq \eps\,, \] 
\[ 
\tff^{*t} = \delta g(\ss^t) - \delta (\RR^t)^{-1} \BB\left(\BB^\top (\RR^t)^{-1} \BB\right)^+ 
\BB^\top g(\ss^t) \] 
and
$\delta = \frac{1}{\sqrt{m}}$.
}
\end{itemize}
Then, 
$\ff^{*T+1}$ is $(\mu/(1+\e\delta)^{T},1.01)$-central,
as long as we set $\e \leq 10^{-5}k^{-3}$
and $\es \leq 10^{-5} k^{-3}$.
\label{lem:small_residual}
\end{lemma}
\begin{proof}
For all $t\in[T+1]$, we denote the residual
of $\ff^{*t}$ as
$\hh^t = 
\frac{\cc(1+\e\delta)^{t-1}}{\mu} 
+ \frac{1}{\ss^{+,*t}} - \frac{1}{\ss^{-,*t}}$.
Note that $\CC^\top \hh^1 = \zerov$ as $\ff^{*1}$ is $\mu$-central.

We assume that the statement of the lemma is not true, and
let $\hT$ be the smallest $t\in [T+1]$ such that
$\ff^{*t}$ is not $(\mu / (1 + \e\delta)^{t-1},1.01)$-central. 
Obviously $\hT > 1$.
This means that $\ff^{*t}$ is 
$(\mu / (1 + \e\delta)^{t-1},1.01)$-central for all $t\in[\hT-1]$, i.e.
$\ss^{*t}\approx_{1.01} \ss\left(\mu / (1+\e\delta)^{t-1}\right)$.

Also, note that by Lemma~\ref{lem:resistance_stability2} about slack stability,
and since
$(1+\e\delta)^{\left|\hT-t\right|} \leq (1+\delta)^{1.1k}$,
we have
$\ss\left(\mu / (1+\e\delta)^{t-1}\right) 
\approx_{3.7k^2} \ss\left(\mu / (1+\e\delta)^{\hT-1}\right)$ for all $t\in[T+1]$.
Additionally, note that, as shown in proof of Lemma~\ref{lem:ipm_step}, we have
\[ \left\|\tff^{*\hT-1}\right\|_{\rr^{\hT-1},\infty} \leq
\left\|\tff^{*\hT-1}\right\|_{\rr^{\hT-1},2} \leq 1 \,,\]
so 
\begin{align*}
\left\|\frac{\ss^{*\hT}}{\ss^{*\hT-1}} - \onev\right\|_\infty
& =\e\left\|\frac{\tff^{\hT-1}}{\ss^{*\hT-1}}\right\|_\infty\\
& \leq \e(1+\es)\left\|\sqrt{\rr^{\hT-1}} \tff^{\hT-1}\right\|_\infty\\
& \leq \e(1+\es)\left(\left\|\sqrt{\rr^{\hT-1}} \tff^{*\hT-1}\right\|_\infty + \eps\right)\\
& \leq \e(1+\es)\left(1 + \eps\right)\\
& \leq 1.3\e\,.
\end{align*}
From this we conclude that $\ss^{*\hT} \approx_{1+2.6\e} \ss^{*\hT-1}$, and from the previous
discussion we get that
\[\ss^{*\hT} 
\approx_{1+2.6\e} 
\ss^{*\hT-1}
\approx_{1.01} 
\ss(\mu/(1+\e\delta)^{\hT-1})
\approx_{3.7k^2}
\ss(\mu/(1+\e\delta)^{t-1})
\approx_{1.01}
\ss^{*t}\,,\]
so
$\ss^{*\hT} \approx_{4k^2} \ss^{*t}$
for all $t\in[\hT-1]$.

On the other hand,
if we apply Lemma~\ref{lem:ipm_step} $\hT-1$ times
with $\orr = \rr^{*\hT}$, we get
\begin{align*}
& \left\|\CC^\top \hh^{\hT}\right\|_{\left(\CC^\top \RR^{\hT}\CC\right)^+} \\
& =\left\|\CC^\top \hh^{\hT}\right\|_{\oHH^+} \\
& \leq 
(1+\e\delta) 
\left\|\CC^\top \hh^{{\hT}-1}\right\|_{\oHH^+}
+ 5 \left\|\frac{\rr^{*{\hT}-1}}{\orr}\right\|_\infty^{1/2} \e \cdot \es +
2\left\|\frac{\rr^{*\hT}}{\orr}\right\|_\infty^{1/2} \e^2\\
& \dots\\
& \leq 
5 \sum\limits_{t=1}^{\hT-1} (1+\e\delta)^{\hT-t-1}
\left\|\frac{\rr^{*t}}{\rr^{*\hT}}\right\|_\infty^{1/2} \e \cdot \es
+
2\sum\limits_{t=1}^{\hT-1} (1+\e\delta)^{\hT-t-1}
\left\|\frac{\rr^{*t+1}}{\rr^{*\hT}}\right\|_\infty^{1/2} \e^2\\
& \leq 
6 T
\max_{t\in [\hT-1]} \left\|\frac{\rr^{*t}}{\rr^{*\hT}}\right\|_\infty^{1/2} \e \cdot \es
+
2.4 T \max_{t\in [\hT-1]} 
\left\|\frac{\rr^{*t+1}}{\rr^{*\hT}}\right\|_\infty^{1/2} \e^2\\
& \leq 
24 T k^2 \e \cdot \es
+
10 T k^2
\e^2 \\
& =
24
k^3 \es
+
10
k^3
\e\\
& \leq 1/1000\,,
\end{align*}
where we used the fact that 
$(1+\e\delta)^{\hT} \leq e^{\e\delta T} = e^{\delta k} \leq 1.2$
and our setting of $\es \leq 10^{-5}k^{-3}$ and 
$\e \leq 10^{-5}k^{-3}$.
By Lemma~\ref{lem:small_residual_linfty}
this implies that 
$\ff^{*\hT}$
is $(\mu / (1+\e\delta)^{\hT-1},1.01)$-central, a contradiction.

\end{proof}

We are now ready to prove the following lemma, which is the goal of this section:
\begin{replemma}{lem:approx_central}
Let $\ff^1,\dots,\ff^{T+1}$ be flows with
slacks $\ss^t$ and resistances $\rr^t$ for $t\in[T+1]$, 
where $T = \frac{k}{\e}$ for some $k\leq \sqrt{m}/10$
and $\e = 10^{-5}k^{-3}$, %
such that
\begin{itemize}
\item $\ff^1$ is $(\mu,1+\es/8)$-central for $\es = 10^{-5} k^{-3}$ %
\item {For all $t\in[T]$,
$\ff^{t+1} = \begin{cases}
\ff(\mu) + \e \sum\limits_{i=1}^{t} \tff^i & \text{if $\exists i\in[t]:\tff^{i} \neq \zerov$}\\
\ff^1 & \text{otherwise}
\end{cases}$, where
\[ 
\tff^{*t} = \delta g(\ss^t) - \delta (\RR^t)^{-1} \BB\left(\BB^\top (\RR^t)^{-1} \BB\right)^+ \BB^\top g(\ss^t) \] 
for $\delta = \frac{1}{\sqrt{m}}$ and
\[ \left\|\sqrt{\rr^t}\left(\tff^{*t} - \tff^{t}\right) \right\|_\infty \leq \eps \] 
for $\eps = 10^{-6} k^{-6}$.
}
\end{itemize}
Then, setting $\e = \es = 10^{-5} k^{-3}$ and
$\eps = 10^{-6} k^{-6}$
we get that
$\ss^{T+1} \approx_{1.1} \ss\left(\mu/(1+\e\delta)^{k\e^{-1}}\right)$.
\end{replemma}
\begin{proof}
We set $\ff^{*1} = \ff(\mu)$ and
for each $t\in[T]$, %
\[ \ff^{*t+1} = \ff^{*t}  + \e \tff^{*t} \,, \]
and the corresponding slacks $\ss^{*t}$ and resistances $\rr^{*t}$.
Let $\hT$ be the first $t\in[T+1]$ such that
$\ss^{\hT} \not\approx_{1+\es} \ss^{*\hT}$.
Obviously $t > 1$ as $\ss^1 \approx_{1+\es/8} \ss(\mu) = \ss^{*1}$.

Now, for all $t\in[T]$ we have
\[ \left\|\sqrt{\rr^t}\left(\tff^{*t} - \tff^t\right)\right\|_\infty \leq \eps \,. \]
Fix some $e\in E$. If $\tf_e^t = \zerov$ for all $t\in [\hT-1]$, then we have
$\sqrt{r_e^{\hT}}\left|\tf_e^{*t}\right| = \sqrt{r_e^t}\left|\tf_e^{*t}\right| \leq \eps$ for all such $t$. 
This means that
\begin{align*}
 \sqrt{r_e^{\hT}} \left|f_e^{*\hT} - f_e^{\hT}\right|
& \leq \sqrt{r_e^{\hT}} \left|f_e^{*\hT} - f_e^{*1}\right|
+ \sqrt{r_e^{\hT}} \left|f_e^{*1} - f_e^{\hT}\right|\\
& = \sqrt{r_e^{\hT}} \left|f_e^{*\hT} - f_e^{*1}\right|
+ \sqrt{r_e^{\hT}} \left|f_e^{*1} - f_e^{1}\right|\\
& \leq \e\sqrt{r_e^{\hT}}\sum\limits_{t=1}^{\hT-1} \left|\tf_e^{*t}\right|
+ \sqrt{r_e^{\hT}} \left|f_e^{*1} - f_e^{1}\right|\\
& \leq \hT\e\eps
+ \sqrt{r_e^{\hT}} \left|f_e^{*1} - f_e^{1}\right|\\
& \leq k \eps + \sqrt{2} \es/8\\
& \leq \es / 2\,,
\end{align*}
as long as $\eps \leq \es / (2k) = O(1/k^4)$.
In the second to last inequality
we used Lemma~\ref{lem:approx_conversion}.

Otherwise, there exists $t\in[\hT-1]$ such that $\tf_e^t \neq \zerov$,
and by definition
$f_e^{\hT} = f_e(\mu) + \e\sum\limits_{t=1}^{\hT-1} \tf_e^t$,
so
\begin{align*}
\sqrt{r_e^{*\hT}}\left|f_e^{*\hT} - f_e^{\hT}\right|
& \leq 
\sqrt{r_e^{*\hT}}\left|f_e^{*1} - f_e(\mu)\right| + 
\e\sum\limits_{t=1}^{\hT-1}\sqrt{r_e^{*\hT}}\left|\tf_e^{*t} - \tf_e^t\right|\\
& \leq 3k^2
\e\sum\limits_{t=1}^{\hT-1}\sqrt{r_e^{*t}}
\left|\tf_e^{*t} - \tf_e^t\right|\\
& \leq 
 3k^2
 \e(1+\es)\sum\limits_{t=1}^{\hT-1}\sqrt{r_e^{t}}\left|\tf_e^{*t} - \tf_e^t\right|\\
& \leq 
3k^2\e(1+\es) T \eps\\
& \leq4k^3\eps\,,
\end{align*}
where we have used 
Lemma~\ref{lem:resistance_stability2}
and 
the fact that $\ss^t \approx_{1+\es} \ss^{*t}$ for all $t\in [\hT-1]$
which also implies that $\sqrt{\rr^t}\approx_{1+\es} \sqrt{\rr^{*t}}$.
Setting $\eps = \frac{\es}{8k^3} = O\left(\frac{1}{k^6}\right)$, this becomes $\leq \es/2$.

Therefore we have proved that
$\left\|\sqrt{\rr^{*\hT}}\left(\ff^{*\hT} - \ff^{\hT}\right)\right\|_\infty \leq \es/2$, and so
$\ss^{*\hT} \approx_{1+\es} \ss^{\hT}$, a contradiction.
Therefore we 
conclude that $\ss^t \approx_{1+\es} \ss^{*t}$ for all $t\in [T+1]$.

Now,
as long as $\e, \es \leq 10^{-5} k^{-3}$, 
we can apply Lemma~\ref{lem:small_residual}, which guarantees
that $\ss^{*T+1} \approx_{1.01} \ss\left(\mu / (1+\e\delta)^{T}\right)$,
and so $\ss^{T+1} \approx_{1.1} \ss\left(\mu / (1+\e\delta)^{T}\right)$.
Therefore we set $\e = \es = 10^{-5} k^{-3}$ and
$\eps = 10^{-6} k^{-6} \leq \frac{\es}{8k^3}$.
\end{proof}

\subsection{Proof of Lemma~\ref{lem:multistep}}
\label{proof_lem_multistep}
\begin{proof}

We will apply Lemma~\ref{lem:approx_central} 
with $\ff^1$ being the flow corresponding to the resistances
$\cL.\rr = \cC.\rr$,
and $T = k\e^{-1}$.
Note that it is important to maintain the invariant $\cL.\rr = \cC.\rr$ throughout
the algorithm so that both data structures
correspond to the same electrical flow problem.
For each $t\in [T]$, for the $t$-th iteration,
Lemma~\ref{lem:approx_central}
requires an estimate $\tff^t$
such that
$\left\|\sqrt{\rr^t}\left(\tff^{*t} - \tff^t\right)\right\|_\infty \leq \eps$,
where
\[ \tff^{*t} = \delta g(\ss^t) - \delta 
(\RR^t)^{-1}\BB \left(\BB^\top (\RR^t)^{-1}\BB\right)^+\BB^\top g(\ss^t) \]
and $\delta = 1/\sqrt{m}$.

We claim that such an estimate can be computed for all $t$ by using $\cL$ and $\cC$.
We apply the following process for each $t\in[T]$:
\begin{itemize}
\item{Let $Z$ be the edge set returned by $\cL.\textsc{Solve}()$.}
\item{Call $\cC.\textsc{Check}(e)$ for each $e\in Z$ to obtain flow values
$\tf_e^t$.}
\item{Compute $\ff^t$ and its slacks $\ss^{t+1}$ and resistances $\rr^{t+1}$
as in Lemma~\ref{lem:approx_central}, i.e.
\[ f_e^{t+1} = \begin{cases}
f_e(\mu) + \e \sum\limits_{i=1}^{t} \tf_e^i & \text{if $\exists i\in[t]:\tf_e^{i} \neq 0$}\\
f_e^1 & \text{otherwise}
\end{cases}\,. \]
This can be computed in $O\left(\left|Z\right|\right)$ by 
adding either $\e\tf_e^i$ or $f_e(\mu) - f_e^1 + \e\tf_e^i$ to $f_e^t$ for each $e\in Z$.
}
\item{Call $\cL.\textsc{Update}(e,\ff^{t+1})$ 
and 
$\cC.\textsc{Update}(e,\ff^{t+1})$ 
for all $e$ in the support of $\tff^t$.
Note that $\cL.\textsc{Update}$ works as long as
\[ r_e^{\max} / \alpha \leq r_e^{t+1} \leq \alpha \cdot r_e^{\min} \,,\]
where $r_e^{\max}$, $r_e^{\min}$ are the maximum and minimum values of
$\cL.r_e$ since the last call to $\cL.\textsc{BatchUpdate}$.
After this, we have $\cL.\rr = \cC.\rr = \rr^{t+1}$.}
\end{itemize}
In the above process, when $\cL.\textsc{Solve}()$ is called
we have $\cL.\rr = \rr^t$ (for $t=1$ this is true because
$\cL.\rr$ are the resistances corresponding to $\ff^1$).
By the $(\alpha,\beta,\eps/2)$-\textsc{Locator} guarantees
in Definition~\ref{def:locator}, with high probability
$Z$ contains all the edges $e$ such that 
$\sqrt{r_e^t} \left|\tf_e^{*t} \right| \geq \eps / 2$.
Now, for each $e\in Z$,
$\cC.\textsc{Check}(e)$ 
returns a flow value $\tf_e^t$ such that:
\begin{itemize}
\item $\sqrt{r_e^t} \left|\tf_e^t - \tf_e^{*t}\right| \leq \eps$
\item if $\sqrt{r_e^t} \left|\tf_e^{*t}\right| < \eps / 2$, then $\tf_e^t = 0$.
\end{itemize}
Therefore, the condition that
\[ \sqrt{\rr^t}\left\|\tff^t - \tff^{*t}\right\|_\infty \leq \eps \] 
is satisfied.
Additionally $\tff^t$ is independent of the randomness of $\cL$, because
(the distribution of) $\tff^t$ would be the same if $\cC.\textsc{Check}$
was run for {\bf all} edges $e$.

It remains to show that the $\textsc{Locator}$ requirement
\[ \rr^{\max} / \alpha \leq \rr^{t+1} \leq \alpha \cdot \rr^{\min} \]
is satisfied.
Consider the minimum value of $t$ for which this is not satisfied.
By Lemma~\ref{lem:approx_central}, we have that
\begin{align}
\ss^{\tau+1} \approx_{1.1} \ss\left(\mu/(1+\e\delta)^{\tau}\right)
\label{eq:slack_1.1_approx}
\end{align}
for any $\tau\in [t]$.

Now let $\hrr$ be the resistances of $\cL$ at any point since the last call to 
$\cL.\textsc{BatchUpdate}$. By the lemma statement and (\ref{eq:slack_1.1_approx}),
we know that $\hrr \approx_{1.1^2} \rr(\hmu)$
for some $\hmu\in[\mu / (1+\e\delta)^t, \mu^0]$.
However, we also know that $\mu^0 \leq \mu \cdot (1+\e\delta)^{(0.5\alpha^{1/4} - k)\e^{-1}}$
and so 
\[ \frac{\hmu}{\mu/(1+\e\delta)^{t}} 
\leq (1 + \e\delta)^{0.5\alpha^{1/4}\e^{-1}}
\leq (1 + \delta)^{0.5\alpha^{1/4}}
\,, \]
so by Lemma~\ref{lem:resistance_stability2} we have
\[ 
\ss\left(\mu/(1+\e\delta)^t\right)
\approx_{0.75\alpha^{1/2}}
\ss\left(\hmu\right)\,.
\]
As 
$\ss^{t+1} \approx_{1.1} \ss\left(\mu/(1+\e\delta)^{t}\right)$,
we have that
$\ss^{t+1} \approx_{0.825\alpha^{1/2}} \ss(\hmu)$,
and so 
\[ \rr^{t+1} \approx_{0.825\alpha} \rr(\hmu) \approx_{1.1^2} \hrr \,.\]
This means that $\rr^{t+1} \approx_{\alpha} \hrr$ and is a contradiction.

We conclude that the requirements of $\cL$ are met for all $t$,
and as a result Lemma~\ref{lem:approx_central}
shows that 
$\ss^{T+1} \approx_{1.1} \ss\left(\mu / (1+\e\delta)^{k\e^{-1}}\right)$.
By Lemma~\ref{lem:recenter}, we can
now obtain $\ff\left(\mu / (1+\e\delta)^{k\e^{-1}}\right)$.
Finally, we return $\cL.\rr$ and $\cC$ to their original states.

\paragraph{Success probability.}
We note that all the outputs of $\cC$ are independent of the randomness
of $\cL$, and $\cL$ is only updated based on these outputs. 
As each operation of $\cL$ succeeds with high probability,
the whole process succeeds with high probability as well.

\paragraph{Runtime.}
The recentering operation in Lemma~\ref{lem:recenter} takes $\tO{m}$.
Additionally, we call $\cL.\textsc{Solve}$ 
$k\e^{-1} = O(k^4)$ times
and, as $|Z| = O(1/\eps^2)$, the total
number of times $\cL.\textsc{Update}$,
$\cC.\textsc{Update}$, and
$\cC.\textsc{Check}$
are called is 
$O(k\e^{-1} \eps^{-2}) = O(k^{16})$.

\end{proof}

\subsection{Proof of Lemma~\ref{lem:mincostflow}}
\label{proof_lem_mincostflow}
\begin{proof}

Let $\delta = 1/\sqrt{m}$.
Over a number of $T=\tO{m^{1/2}/k}$ iterations, we will repeatedly apply $\textsc{MultiStep}$ 
(Lemma~\ref{lem:multistep}). We will also replace the oracle from 
Definition~\ref{def:perfect_checker} by the $\textsc{Checker}$ data structure
in Section~\ref{sec:checker}.

\paragraph{Initialization.}
We first initialize the $\textsc{Locator}$ with error $\eps / 2$, by calling
$\cL.\textsc{Initialize}(\ff)$.
Let $\ss^t$ be the slacks $\cL.\ss$ 
before the $t$-th iteration and $\rr^t$ the corresponding resistances,
and $\ss^{0t}$ be the slacks $\cL.\rr^0$ before the $t$-th iteration
and $\rr^{0t}$ the corresponding resistances,
for $t\in[T]$.
Also, we let $\mu_t = \mu / \left(1 + \e\delta\right)^{(t-1)k\e^{-1}}$.
We will maintain the invariant that
$\ss^t \approx_{1+\es/8} \ss\left(\mu_t\right)$, which is a requirement
in order to apply Lemma~\ref{lem:multistep}.

As in~\cite{gao2021fully}, we will also need to maintain $O(k^4)$
$\textsc{Checker}$s $\cC^i$ for $i\in[O(k^4)]$,
so we call $\cC^i.\textsc{Initialize}(\ff, \eps, \bc)$ for each one of these.
Note that in general $\bc \neq \beta$, as 
the vertex sparsifiers $\cL$ and $\cC^i$ will not be on the same vertex set.
As in Lemma~\ref{lem:multistep}, we will maintain the invariant that
$\cL.\rr = \cC^i.\rr$ for all $i$.

\paragraph{Resistance updates.}
Assuming that all the requirements of Lemma~\ref{lem:multistep} (\textsc{MultiStep})
are satisfied at the $t$-th iteration, that lemma computes a flow 
$\off = \ff\left(\mu_{t+1}\right)$ with slacks $\oss$.
In order to guarantee that 
$\ss^{t+1} \approx_{1+\es/8} \ss\left(\mu_{t+1}\right)$,
we let $Z$ be the set 
of edges such that either
\[ s_e^{+,t} \not\approx_{1+\es/8} \os_e^+ \text{ or }
 s_e^{-,t} \not\approx_{1+\es/8} \os_e^- \]
and then call
$\cL.\textsc{Update}(e,\off)$ for all $e\in Z$.
This guarantees that 
$s_e^{+,t+1} = \os_e^+$
and
$s_e^{-,t+1} = \os_e^-$
for all $e\in Z$
and so 
$\ss^{t+1}\approx_{1+\es/8} \oss = \ss\left(\mu_{t+1}\right)$.
We also apply the same updates to the $\cC^i$'s 
using $\cC^i.\textsc{Update}$, in order to ensure that they have
the same resistances with $\cL$.

\paragraph{Batched resistance updates.}
The number of times $\cL.\textsc{Update}$ is called can be quite large
because of multiple edges on which error slowly accumulates. This is because in general
$\Omega(m)$ resistances will be updated
throughout the algorithm. As \textsc{Locator}.\textsc{Update} is only slightly sublinear,
this would lead to an $\Omega(m^{3/2})$-time algorithm.
For this reason, as in~\cite{gao2021fully}, we occasionally
(every $\hT$ iterations for some $\hT\geq 1$ to be defined later)
perform batched updates by calling $\cL.\textsc{BatchUpdate}(Z,\off)$,
where $Z$
is the set
of edges such that either
\[ s_e^{+,t} \not\approx_{1+\es/16} \os_e^+
\text{ or } s_e^{-,t} \not\approx_{1+\es/16} \os_e^- \,. \]
This again guarantees that 
$s_e^{+,t+1} = \os_e^+$
and $s_e^{-,t+1} = \os_e^-$
for all $e\in Z$
and so $\ss^{t+1}\approx_{1+\es/16} \oss = \ss\left(\mu_{t+1}\right)$.
Note that after updating 
$\cL.\ss$ and $\cL.\rr$, this operation also sets
$\cL.\rr^0 = \cL.\rr$.
We perform the same resistance updates to the $\cC^i$'s in the regular
(i.e. not batched) way, using $\cC^i.\textsc{Update}$.

\paragraph{\textsc{Locator} requirements.}
What is left is to ensure that the requirements of Lemma~\ref{lem:multistep}
are satisfied at the $t$-th iteration, as well as that 
the requirements of $\cL$.\textsc{Update}
and $\cL$.\textsc{BatchUpdate} 
from Definition~\ref{def:locator}
are satisfied. 

The requirements are as follows:
\begin{enumerate}
\item{Lemma~\ref{lem:multistep}: $\ss^{0t} \approx_{1+\es/8} \ss\left(\mu^0\right)$
for some $\mu^0 \leq 
\mu_t\cdot(1+\e\delta)^{(0.5\alpha^{1/4} - k)\e^{-1}}$.

Note that
$\cL.\ss^0$ is updated every time $\cL.\textsc{BatchUpdate}$ is called,
and after the call we have
$\cL.\ss^0 = \cL.\ss \approx_{1+\es/16} \ss(\mu^0)$ for some $\mu^0 > 0$.
To ensure that it is called often enough,
we call
$\cL.\textsc{BatchUpdate}(\emptyset)$ every $(0.5\alpha^{1/4} / k - 1)\e^{-1}$ iterations.
Because of this, we have
$\mu^0
\leq \mu_t \cdot \left(1+\e\delta\right)^{(0.5\alpha^{1/4}/k - 1)\e^{-1} \cdot k}
= \mu_t \cdot \left(1+\e\delta\right)^{(0.5\alpha^{1/4} - k)\e^{-1}}
$. Additionally, for any resistances $\hrr$ that $\cL$ had at any point since the last
call to $\cL.\textsc{BatchUpdate}$, it is immediate that 
\[ \hrr \approx_{(1+\es/8)^2} \rr(\hmu)\]
for some $\hmu\in[\mu_t, \mu^0]$, as this is exactly the invariant that our calls to
$\cL.\textsc{Update}$ maintain. Therefore, $\hrr \approx_{1.1^2} \rr(\hmu)$.
}
\item{$\cL$.\textsc{Update}: $r_e^{\max} / \alpha \leq r_e^{t+1} \leq \alpha \cdot r_e^{\min}$,
where $r_e^{\min}, r_e^{\max}$ are the minimum and maximum values that
$\cL.r_e$ has had since the
last call to $\cL.\textsc{BatchUpdate}$.

Let $\hrr$ be any value of $\cL.\rr$ since the last call to $\cL.\textsc{BatchUpdate}$.
Because of the invariant maintained by resistance updates (including inside 
$\textsc{MultiStep}$), we have that
$\hrr$ are $(\hmu,1.1)$-central resistances for some $\hmu$ such that
\[ \mu_{t+1} \leq \hmu
\leq \mu_{t+1} \cdot \left(1+\e\delta\right)^{(0.5\alpha^{1/4} - k)\e^{-1}}
\leq \mu_{t+1} \cdot (1+\delta)^{0.5\alpha^{1/4}} \,.\]
As in the previous item, we have that 
and $\ss^{0t}$ are $(\mu^0, 1+\es/8)$-central.
By Lemma~\ref{lem:resistance_stability2} this implies
\[ \ss(\mu_{t+1}) 
\approx_{0.75\alpha^{1/2}}
\ss(\hmu)\,,\]
and since $\hrr \approx_{1.1^2} \rr(\hmu)$, we conclude that 
$\rr(\mu_{t+1}) \approx_{\alpha} \hrr$.
}
\item{$\cL$.\textsc{BatchUpdate}: Between any two successive calls to $\cL.\textsc{Initialize}$,
the number of edges updated
(number of calls to $\cL.\textsc{Update}$ plus
the sum of $|Z|$ for all calls to $\cL.\textsc{BatchUpdate}$)
is $O(\beta m)$.

We make sure that this is satisfied by calling $\cL.\textsc{Initialize}(\off)$
every $\es \sqrt{\beta m} / k$ iterations, where
$\off = \ff(\mu_t)$, at the beginning of the $t$-th iteration.

Consider any two successive initializations at iterations
$t^{init}$ and $t^{end}$ respectively.
Let $\ell$ be the number of edges $e$ that have potentially
been updated, i.e. such that either
\[ s_e(\mu_{t^{init}})^{+} \not\approx_{1+\es/16}s_e(\mu_{i})^{+} \text{ or }
s_e(\mu_{t^{init}})^- \not\approx_{1+\es/16}s_e(\mu_{i})^- \]
for some $i\in\left[t^{init},t^{end}\right]$.
First, note that this implies that
\[ \sqrt{r_e(\mu_{t^{init}})}\left|f_e(\mu_{t^{init}}) - f_e(\mu_i)\right| > 
\frac{\es / 16}{1+\es/16} > \es / 17 \,. \]
Now, by the fact that $t^{end} - t^{init} \leq \es \sqrt{\beta m} / k$, we have that
\[ \mu_{t^{init}} \leq 
\mu_i \cdot (1+\e\delta)^{k\e^{-1}\cdot \es\sqrt{\beta m} / k}
\leq \mu_i \cdot (1+\delta)^{\es\sqrt{\beta m}} \,.\]
By applying Lemma~\ref{lem:resistance_stability1}
with $k=\es\sqrt{\beta m}$
and $\gamma = \es/17$, we get that 
$\ell \leq O(\beta m)$.
Therefore the statement follows.
}
\end{enumerate}

We will also need to show how to implement the $\textsc{PerfectChecker}$
used in $\textsc{MultiStep}$ using the $\cC^i$'s, as well as how to satisfy
all $\textsc{Checker}$ requirements.

\paragraph{\textsc{Checker} requirements.}

\begin{enumerate}
\item{Implementing $\eps$-$\textsc{PerfectChecker}$ inside $\textsc{MultiStep}$. 

We follow almost the same procedure as in~\cite{gao2021fully}, other than the fact
that we also need to provide some additional information to $\cC^i.\textsc{Solve}$.
Each call to $\textsc{PerfectChecker}.\textsc{Update}$ translates to calls to
$\cC^i.\textsc{TemporaryUpdate}$ for all $i$. In addition, the $i$-th 
batch of calls to $\textsc{PerfectChecker}.\textsc{Check}$ inside $\textsc{MultiStep}$
(i.e. that corresponding to a single set of edges returned by $\cL.\textsc{Solve}$)
is only run on $\cC^i$ using $\cC^i.\textsc{Check}$. As each call to
$\cC^i.\textsc{Check}$ is independent of previous calls to it, we can get correct
outputs with high probability even when we run it multiple times (one for each edge
returned by $\cL.\textsc{Solve}$).

In order to guarantee that we have a vector $\vpi_{old}^i$ as required by $\cC^i.\textsc{Check}$,
once every $k^4$ calls to $\textsc{MultiStep}$ (i.e. if $t$ is a multiple of $k^4$) we compute
\[ \vpi_{old}^{i} = \vpi^{C^{i,t}}\left(\BB^\top g(\ss(\mu_t))\right) \]
for all $i\in[O(k^4)]$, where
$C^{i,t}$ is the vertex set of the Schur complement data 
structure stored internally by the $\cC^i$ right before the $t$-th call to $\textsc{MultiStep}$.
This can be computed in $\tO{m}$ for each $i$ as in $\textsc{DemandProjector}.\textsc{Initialize}$
in Lemma~\ref{lem:ds}.
Now, the total number of $\cC^i.\textsc{TemporaryUpdate}$s that have not been rolled back
is $O(k^{16})$, and the total number of $\cC^i.\textsc{Update}$s over $k^4$ calls to 
$\textsc{MultiStep}$ by Lemma~\ref{lem:resistance_stability1}
is $O(k^{10} / \es^2) = O(k^{16})$. This means that the total number
of terminal insertions to $C^{i,t}$ as well as resistance changes is $O(k^{16})$.
By Lemma~\ref{lem:old_projection_approximate}, if $C^i$ is the current state of the vertex set
of the Schur complement of $\cC^i$ and $\ss$ are the current slacks,
\[ \mathcal{E}_{\rr} \left(\vpi_{old}^i - 
\vpi^{C^i}(g(\ss))\right) \leq \tO{\alpha' \bc^{-4}} \cdot k^{32}\,,
\]
where $\alpha'$ is the largest possible multiplicative change of some $\cC^i.r_e$ since
the computation of $\vpi_{old}^i$.
Furthermore, note that $\vpi_{old}$ is supported on $C^i$. This is because
$C^{i,t} \subseteq C^i$ and $C^{i,t}$ does not contain any temporary 
terminals.

Now, as we have already proved in Lemma~\ref{lem:multistep}, at any
point inside the $t$-th call to $\textsc{MultiStep}$, $\cC^i.\rr$ are 
$(\hmu,1.1)$-central resistances for some 
$\hmu\in[\mu_{t+1}, \mu_t]$. 

Fix $\hmu\in[\mu_{t+1},\mu_t], \hmu'\in[\mu_{t'+1},\mu_{t'}]$, as well as the corresponding resistances
of $\cC^i$, $\hrr$, $\hrr'$, where $t' \geq t$.
Now, note that since we are computing $\vpi_{old}^i$ every $k^4$ calls to $\textsc{MultiStep}$,
we have that 
\[ \frac{\hmu}{\hmu'} \leq \frac{\mu_t}{\mu_{t'+1}} 
\leq (1+\e\delta)^{k\e^{-1}\cdot (t'-t + 1)} 
\leq (1+\delta)^{O(k^5)}\,,\]
so Lemma~\ref{lem:resistance_stability2} implies that
\[ \ss(\hmu) \approx_{O(k^{10})} \ss(\hmu')\,. \]
As $\hrr \approx_{1.1^2} \rr(\hmu)$
and $\hrr' \approx_{1.1^2} \rr(\hmu')$,
we get that $\hrr \approx_{O(k^{20})} \hrr'$. Therefore,
$\alpha' \leq O(k^{20})$.
Setting $\bc \geq \tOm{\alpha'^{1/4} k^{8} \eps^{-1/2} m^{-1/4}} = \tOm{k^{16} / m^{1/4}}$, we get that
\[ \tO{\alpha'\bc^{-4}}\cdot k^{32} \leq \eps^2 m / 4\,, \]
as required
by $\cC^i.\textsc{Check}$.

Finally, at the end of $\textsc{MultiStep}$ we bring all
$\cC^i$ to their original state before calling $\textsc{MultiStep}$,
by calling $\cC^i.\textsc{Rollback}$. We also update
all the resistances of $\cL$ to their original state by calling $\cL.\textsc{Update}$.
}
\item{Between any two successive calls to 
$\cC^i$.$\textsc{Initialize}$, the total number of edges updated
at any point (via $\cC^i.\textsc{Update}$ or
$\cC^i.\textsc{TemporaryUpdate}$ that have not been rolled back)
is $O(\bc m)$.

For $\textsc{Update}$, we can apply a similar analysis as in the
$\textsc{Locator}$ case to show that if we call 
$\cC^i$.$\textsc{Initialize}$ every $\es\sqrt{\bc m} /k$ iterations,
the total number of updates never exceeds $O(\bc m)$.
For $\textsc{TemporaryUpdate}$, note that at any time there 
are at most $O(k^{16})$ of these that have not been rolled back
(this is inside \textsc{MultiStep}). Therefore, as long as 
$k^{16} \leq O(\bc m) \Leftrightarrow \bc \geq \Omega(k^{16} / m)$,
the requirement is met.
}
\end{enumerate}

\paragraph{Output Guarantee.}

After the application of Lemma~\ref{lem:multistep} at the last
iteration,
we will have $\off = \ff(\mu_{T+1})$, where
$\mu_{T+1} = 
\mu / (1+\e\delta)^{\tO{m^{1/2} \e^{-1}}}
= \mu / \mathrm{poly}(m)
\leq m^{-10}$.

\paragraph{Success probability.}
Note that all operations of $\textsc{Locator}$ and
$\textsc{Checker}$ work with high probability. Regarding
the interaction of the randomness of 
these data structures and the fact that they work against
oblivious adversaries, we defer to~\cite{gao2021fully}, where
there is a detailed discussion of why this works. 

In short,
note that outside of $\textsc{MultiStep}$,
all updates are deterministic (as they only depend on the central path),
and in $\textsc{MultiStep}$
the updates to $\textsc{Locator}$ and $\textsc{Checker}$
only depend on outputs of a $\textsc{Checker}$. As each time
we are getting the output from a different $\textsc{Checker}$,
the inputs to $\cC^i.\textsc{Check}$ are independent
of the randomness of $\cC^i$, and thus succeed with high probability. Finally,
note that the output of $\textsc{Locator}$ is only passed onto 
$\cC^i.\textsc{Check}$, whose output is then independent of the inputs
received by $\textsc{Locator}$. Therefore, $\textsc{Locator}$ does not
``leak'' any randomness.

Our only deviation from~\cite{gao2021fully} in $\textsc{Checker}$ has to
do with the extra input of $\textsc{Checker}.\textsc{Check}$ ($\vpi_{old}^i$).
However, note that this is computed outside of $\textsc{MultiStep}$,
and as such the only randomness it depends on is the
$\bc$-congestion reduction subset $C^i$ generated when calling 
$\cC^i.\textsc{Initialize}$. 
As such, it only depends on the internal randomness of $\cC^i$. As we 
mentioned, the output of $\cC^i$.$\textsc{Check}$ is never fed back to
$\cC^i$, and thus the operation works with high probability.

\paragraph{Runtime (except $\textsc{Checker}$).}
Each call to $\textsc{MultiStep}$ (Lemma~\ref{lem:multistep})
takes time $\tO{m}$ plus $O(k^{16})$ calls to 
$\cL.\textsc{Update}$ and $O(k^4)$ calls to $\cL.\textsc{Solve}$.
As the total number of iterations is $\tO{m^{1/2} /k}$,
the total time because of calls to $\textsc{MultiStep}$
is $\tO{m^{3/2} / k}$, plus $\tO{m^{1/2} k^{15}}$ 
calls to 
$\cL.\textsc{Update}$ and $\tO{m^{1/2} k^3}$ calls to $\cL.\textsc{Solve}$.

Now, the total number of calls to $\cL.\textsc{Initialize}$ is
$\tO{\frac{m^{1/2} / k}{\es\sqrt{\beta m}/k}} = \tO{k^3 \beta^{-1/2}}$.

The total number of calls to $\cL.\textsc{BatchUpdate}(\emptyset)$
is $\tO{\frac{m^{1/2} / k}{0.5\alpha^{1/4} / k}} = \tO{m^{1/2} \alpha^{-1/4}}$
and the total number of calls to $\cL.\textsc{BatchUpdate}(Z,\ff)$
is $\tO{\frac{m^{1/2}}{k \hT}}$.
Regarding the size of $Z$, let us focus on the calls
to $\cL.\textsc{BatchUpdate}(Z,\ff)$
between two successive calls
to $\cL.\textsc{Initialize}$.
We already showed that the sum of $|Z|$ over all
calls during this interval is $O(\beta m)$. Therefore
the total sum of $|Z|$ over all iterations of the algorithm is
$\tO{m k^3 \beta^{1/2}}$.

In order to bound the number of calls to $\cL.\textsc{Update}$,
we concentrate on those between two successive calls to 
$\cL.\textsc{BatchUpdate}(Z,\ff)$ in iterations $t^{old}$ and $t^{new} > t^{old}$.
After the call to $\cL.\textsc{BatchUpdate}(Z,\ff)$ 
in iteration $t^{old}$ we have
$\cL.\ss \approx_{1+\es/16} \ss(\mu_{t^{old}})$.
Fix $\mu\in\left[\mu(t^{new}),\mu(t^{old})\right]$ and let $\ell$ be the 
number of $e\in E$ such that
$s_e(\mu)\not\approx_{1+\es/8} s_e^{t^{old}}$.
As $s_e^{t^{old}} \approx_{1+\es/16} s_e(\mu_{t^{old}})$
by the guarantees of $\cL.\textsc{BatchUpdate}$,
this implies that
$s_e(\mu)\not\approx_{1+\es/16} s_e(\mu_{t^{old}})$,
and so 
\[ \sqrt{r_e(\mu^{t^{old}})}\left|
	f_e(\mu_{t^{old}}) -
	f_e(\mu) \right| \geq \frac{\es / 16}{1+\es/16} > \es/17 \,.\]
As 
\[ \mu_{t^{old}} \leq \mu\cdot 
(1+\e\delta)^{k\e^{-1}\cdot \hT}
\leq \mu\cdot 
(1+\delta)^{k\cdot\hT}
\,,\]
by applying Lemma~\ref{lem:resistance_stability1}
with $k = (1+\delta)^{k\cdot\hT}$
and $\gamma = \es / 17$, we get that 
$\ell \leq O(k^2 \hT^2 \es^{-2})$.
As there are 
$\tO{\frac{m^{1/2}}{k\hT}}$ calls to
$\cL.\textsc{BatchUpdate}(Z,\ff)$,
the total number of calls to $\cL.\textsc{Update}$
is 
\[  \tO{ m^{1/2} \hT \es^{-2}}
  \tO{ m^{1/2} k^6 \hT}
\,.\]

We conclude that we have runtime $\tO{m^{3/2}/k}$, plus
\begin{itemize}
\item $\tO{k^3 \beta^{-1/2}}$ calls to $\cL.\textsc{Initialize}$,
\item $\tO{m^{1/2} k^3}$ calls to $\cL.\textsc{Solve}$,
\item $\tO{m^{1/2} \left(k^6 \hT + k^{15}\right)}$ calls to $\cL.\textsc{Update}$,
\item $\tO{m^{1/2} \alpha^{-1/4}}$ calls to $\cL.\textsc{BatchUpdate}(\emptyset)$, and
\item $\tO{m^{1/2} k^{-1} \hT^{-1}}$ calls to $\cL.\textsc{BatchUpdate}(Z,\ff)$.
\end{itemize}

\paragraph{Runtime of $\textsc{Checker}$.} We look at each operation separately.
We begin with the runtime of $\textsc{Checker}.\textsc{Check}$. We have 
\[ \tO{\underbrace{m^{1/2} / k}_{\text{\# calls to \textsc{MultiStep}}}
\cdot
\underbrace{k^{16}}_{\text{\# calls in each $\textsc{MultiStep}$}}
\cdot
\underbrace{(\bc m + (k^{16} \bc^{-2} \eps^{-2})^2))\eps^{-2}}_{\text{runtime per call}}
} \,.\] 
To make the first term $\tO{m^{3/2} / k}$, we set $\bc = k^{-28}$. Note that this satisfies
our previous requirements that $\bc \geq \tOm{k^{16} / m}$ and $\bc \geq \tOm{k^{16}/m^{1/4}}$
as long as $k \leq m^{1/176}$.
Therefore the total runtime because of this operation is $\tO{m^{3/2} / k + m^{1/2} k^{195}}$.

For $\textsc{Checker}.\textsc{Initialize}$, we have
\[ \tO{\underbrace{k^4}_{\# \textsc{Checker}s} \cdot \underbrace{k^3 \bc^{-1/2}}_{\# \text{times initialized}} 
\cdot \underbrace{m \bc^{-4} \eps^{-4}}_{\text{runtime per init}}} =
\tO{m k^{157}}\,. \]

For $\textsc{Checker}.\textsc{Update}$, similarly with the analysis of $\textsc{Locator}$ but
noting that there are no batched updates, we have
\[ \tO{
\underbrace{k^4}_{\text{\#\textsc{Checker}s}}
\cdot \underbrace{m \es^{-2}}_{\text{\#calls per \textsc{Checker}}}
\cdot \underbrace{\bc^{-2} \eps^{-2}}_{\text{runtime per call}}
} = \tO{mk^{78}}\,. \]

For $\textsc{Checker}.\textsc{TemporaryUpdate}$, we have
\[ 
\tO{
\underbrace{k^4}_{\text{\#\textsc{Checker}s}}
\cdot
\underbrace{m^{1/2} / k}_{\text{\# calls to \textsc{MultiStep}}}
\cdot
\underbrace{k^{16}}_{\text{\# calls per \textsc{MultiStep}}}
\cdot
\underbrace{(k^{16} \bc^{-2} \eps^{-2})^2}_{\text{runtime per call}}
}
= \tO{m^{1/2} k^{187}}\,.
\]
Finally, note that, by definition, computing the vectors $\vpi_{old}^i$ takes $\tO{m^{3/2}/k}$,
as we do it once per $k^{4}$ calls to $\textsc{MultiStep}$ and it takes $\tO{mk^4}$.

As long as $k \leq m^{1/316}$, the total runtime because of $\textsc{Checker}$ 
is $\tO{m^{3/2}/k}$.
\end{proof}

\section{Deferred Proofs from Section~\ref{sec:locator}}

\subsection{Proof of Lemma~\ref{lem:F_system1}}
We first provide a helper lemma for upper bounding escape probabilities in terms of the underlying graph's resistances.

\begin{lemma}[Bounding escape probabilities]
\label{lem:escape_prob}
Let a graph with resistances $\rr$, and consider a random walk which at each step moves from the current vertex $u$ to an adjacent vertex $v$ sampled with probability proportional to $1/{r_{uv}}$.
Let $p_u^{\{u,t\}}(s)$ represent the probability that a walk starting at $s$ hits $u$ before $t$. Then
\[
p_u^{\{u,t\}}(s) = \frac{R_{eff}(s,t)}{R_{eff}(u,t)} \cdot p_s^{\{s,t\}}(u) \leq \frac{R_{eff}(s,t)}{R_{eff}(u,t)} \leq \frac{r_{st}}{R_{eff}(u,t)}\,.
\]
\end{lemma}
\begin{proof}
Using standard arguments we can prove that if $\LL$ is the Laplacian associated with the underlying graph, then
\[
p_u^{\{u,t\}}(s) = \frac{(\onev_s - \onev_t)^\top \LL^+ (\onev_u - \onev_t)}{R_{eff}(u,t)}\,.
\]
This immediately yields the claim as we can further write it as
\[
p_u^{\{u,t\}}(s) = \frac{R_{eff}(s,t)}{R_{eff}(u,t)} \cdot \frac{(\onev_u - \onev_t)^\top \LL^+ (\onev_s - \onev_t)}{R_{eff}(s,t)} 
= \frac{R_{eff}(s,t)}{R_{eff}(u,t)} \cdot p_s^{\{s,t\}}(u) \leq \frac{R_{eff}(s,t)}{R_{eff}(u,t)}\,,
\]
where we crucially used the symmetry of $\LL$. The final inequality is due to the fact that $R_{eff}(s,t) \leq r_{st}$.

Now let us prove the claimed identity for escape probabilities. Let $\vpsi$ be the vector defined by $\psi_i = p_u^{\{u,t\}}(i)$ for all $i \in V$, which clearly satisfies $\psi_u = 1$ and $\psi_t = 0$. Furthermore, for all $i \notin\{u,t\}$ we have 
\[
\psi_i = \sum_{j \sim i} \frac{  r_{ij}^{-1}  }{ \sum_{k\sim i}r_{ik}^{-1}  } \psi_j\,,
\]
which can be written in short as
\[
(\LL \vpsi)_i = 0\,\quad \textnormal{for all } i\notin\{s,t\}\,.
\]
Now we solve the corresponding linear system. We interpret $\vpsi$ as electrical potentials corresponding to routing $1/R_{eff}(u,t)$ units of electrical flow from $u$ to $t$. Indeed, by Ohm's law, this corresponds to a potential difference $\psi_u  - \psi_t = 1$. Furthermore, this shows that 
\[
\psi_s - \psi_t = (\onev_s - \onev_t)^\top  \LL^+ (\onev_u - \onev_t) \cdot \frac{1}{R_{eff}(u,t)}\,,
\]
which concludes the proof.

\end{proof}

Now we are ready to prove the main statement.
\begin{proof}[Proof of Lemma~\ref{lem:F_system1}]\label{proof_lem_F_system1}
Note that the demand can be decomposed as
$\dd - \vpi^C(\dd) = \dd^1 - \dd^2$, 
where $\dd^1 = \frac{\onev_s}{\sqrt{r_{st}}} - \vpi^C(\frac{\onev_s}{\sqrt{r_{st}}})$
and
$\dd^2 = \frac{\onev_t}{\sqrt{r_{st}}} - \vpi^C(\frac{\onev_t}{\sqrt{r_{st}}})$.
Now let $p^1$ be the probability distribution of $s-C$ random walks
obtained via a random walk from $s$ with transition probabilities proportional to inverse resistances. 
Similarly, let $p^2$ be the probability distribution of $t-C$ random walks obtained by running the same process starting from $t$.

Now, it is well known that an electrical flow is the sum of these random walks, i.e.
\[ \RR^{-1} \BB \LL^+ \dd^1 = \frac{1}{\sqrt{r_{st}}} \cdot \mathbb{E}_{P\sim p^1} \left[net(P)\right] \]
and similarly for $\dd^2$
\[ \RR^{-1} \BB \LL^+ \dd^2 = \frac{1}{\sqrt{r_{st}}} \cdot \mathbb{E}_{P\sim p^2} \left[net(P)\right] \,, \]
where $net(P)\in\mathbb{R}^m$ is a flow vector whose $e$-th entry is the net number of times the edge $e=(u,v)$ is used by $P$.
Therefore we can write:
\begin{align*}
 \left|\frac{\phi_u - \phi_v}{\sqrt{r_{uv}}}\right| 
 & = \sqrt{r_{uv}} \left|\RR^{-1} \BB \LL^+ \dd\right|_{uv}\\
& = \sqrt{\frac{r_{uv}}{r_{st}}} \left|
 \mathbb{E}_{P^1\sim p^1}\left[net_e(P^1)\right]  - \mathbb{E}_{P^2\sim p^2}\left[net_e(P^2)\right]
\right|\,.
\end{align*}
Let us also subdivide $e$ by inserting an additional vertex $w$ in the middle (i.e. $r_{uw} = r_{wv} = r_{uv} / 2$). This has no effect in the random walks,
but will be slightly more convenient in terms of notation.
The first expectation term can be expressed as
\begin{align*}
\mathbb{E}_{P^1\sim p^1}\left[net_e(P^1)\right]
& = \Pr_{P^1\sim p^1}\left[P^1 \text{ visits $t$ before $C\cup\{w\}$}\right] \cdot \mathbb{E}_{P^1\sim p^1}\left[net_e(P^1)\ |\ P^1 \text{ visits $t$ before $C\cup\{w\}$}\right]\\
& + \Pr_{P^1\sim p^1}\left[P^1 \text{ visits $w$ before $C\cup\{t\}$}\right] \cdot \mathbb{E}_{P^1\sim p^1}\left[net_e(P^1)\ |\ P^1 \text{ visits $t$ before $C\cup\{t\}$}\right]\,.
\end{align*}
Now, note that 
\[ \mathbb{E}_{P^1\sim p^1}\left[net_e(P^1)\ |\ P^1 \text{ visits $t$ before $C\cup\{w\}$}\right] = \mathbb{E}_{P^2\sim p^2}\left[net_e(P^2)\right] \,.\]
Additionally, 
\begin{align*}
& \Pr_{P^1\sim p^1}\left[P^1 \text{ visits $w$ before $C\cup\{t\}$}\right] \\
& = \Pr_{P^1\sim p^1}\left[P^1 \text{ visits $w$ before $C$}\right] \cdot
\Pr_{P^1\sim p^1}\left[P^1 \text{ visits $w$ before $t$}\ |\ P^1 \text{ visits $w$ before $C$} \right] 
\end{align*}
The first term of the product is $p_w^{C\cup\{w\}}(s)$. For the second term, we define a new graph $\hG$
by deleting $C$, and denote the hitting probabilities in $\hG$ by $\hp$. Then, the second term is equal to
$\hp_w^{\{t,w\}}(s)$.

We have concluded that 
\begin{align*}
\mathbb{E}_{P^1\sim p^1}\left[net_e(P^1)\right]
& \leq 
\mathbb{E}_{P^2\sim p^2}\left[net_e(P^2)\right] + p_w^{C\cup\{w\}}(s) \cdot \hp_w^{\{t,w\}}(s)\,.
\end{align*}
Combining this with the symmetric argument for $p^2$ shows that 
\begin{align*}
\left|\mathbb{E}_{P^1\sim p^1}\left[net_e(P^1)\right]
- \mathbb{E}_{P^2\sim p^2}\left[net_e(P^2)\right]\right|
\leq p_w^{C\cup\{w\}}(s) \cdot \hp_w^{\{t,w\}}(s) + p_w^{C\cup\{w\}}(t) \cdot \hp_w^{\{s,w\}}(t)\,.
\end{align*}

Using Lemma~\ref{lem:escape_prob} and the fact that $\hR_{eff}(w, t) \geq r_{uv} / 4$ ($\hR_{eff}$ are the effective resistances in $\hG$), we can bound
 \[
 \widehat{p}_{w}^{\{t,w\}}(s) \leq \min\{1,\frac{r_{st}}{\hR_{eff}(w, t)}\} \leq \min\{1,4 \frac{r_{st}}{r_{uv}}\}\,,
\]
and the same upper bound holds for $\hp_w^{\{s,w\}}(t)$. Therefore,
\begin{align*}
& \left|\mathbb{E}_{P^1\sim p^1}\left[net_e(P^1)\right]
- \mathbb{E}_{P^2\sim p^2}\left[net_e(P^2)\right]\right|\\
& \leq \min\{1,4 \frac{r_{st}}{r_{uv}}\} \left(p_w^{C\cup\{w\}}(s) + p_w^{C\cup\{w\}}(t)\right)\\
& \leq 2 \sqrt{\frac{r_{st}}{r_{uv}}} \left(p_w^{C\cup\{w\}}(s) + p_w^{C\cup\{w\}}(t)\right)\\
& \leq 2 \sqrt{\frac{r_{st}}{r_{uv}}} 
\left(
p_u^{C\cup\{u\}}(s) 
+ p_v^{C\cup\{v\}}(s) 
+ p_u^{C\cup\{u\}}(t)
+ p_v^{C\cup\{v\}}(t)
\right)\,.
\end{align*}

Putting everything together, we have that
\begin{align*}
& \left|\frac{\phi_u - \phi_v}{\sqrt{r_{uv}}}\right| 
\leq 
2 \left(
p_u^{C\cup\{u\}}(s) 
+ p_v^{C\cup\{v\}}(s) 
+ p_u^{C\cup\{u\}}(t)
+ p_v^{C\cup\{v\}}(t)
\right)\,.
\end{align*}
\end{proof}

\subsection{Proof of Lemma~\ref{st_projection1_energy}}
\label{proof_st_projection1_energy}

\begin{proof}
Let $\dd = \BB^\top \frac{\onev_e}{\sqrt{\rr}}$ and $e = (u,w)$.
Note that $\mathcal{E}_r\left(\dd\right) \leq r_e \cdot (\frac{1}{\sqrt{r_e}})^2 = 1$, therefore
the case that remains is 
\begin{align}
R_{eff}(C,e) \geq 36 \cdot r_e\,.
\label{eq:st_projection1_condition}
\end{align}
For each $v\in C$
by Lemma~\ref{st_projection1} we have that
\begin{align*}
\left|\pi_v^{C}(\dd)\right| \leq (p_v^{C}(u) + p_v^{C}(w)) \cdot \frac{\sqrt{r_e}}{R_{eff}(v,e)} \,.
\end{align*}
Now, we would like to bound the energy of routing $\vpi^C(\dd)$ by the 
energy to route it via $w$.
For each $v\in C$ we let $\dd^v$ be the following demand:
\[\dd^v = %
\pi_v^C(\dd) \cdot (\onev_v - \onev_w )\,. \]%
Note that $\vpi^C(\dd) = \sum\limits_{v\in C} \dd^v$.
We have,
\begin{align*}
\sqrt{\mathcal{E}_{\rr}(\pi^C(\dd))}
& = \sqrt{\mathcal{E}_{\rr}\left(\pi^C\left(\sum\limits_{v\in C} \dd^v\right)\right)}\\
& \leq \sum\limits_{v\in C} \sqrt{\mathcal{E}_{\rr}\left(\pi^C\left(\dd^v\right)\right)}\\
& \leq \sum\limits_{v\in C} (R_{eff}(v,w))^{1/2} \left|\pi_v^C(\dd)\right|\\
& \leq \sum\limits_{v\in C} (R_{eff}(v,w))^{1/2} \cdot (p_v^{C}(u) + p_v^{C}(w)) \cdot \frac{\sqrt{r_e}}{R_{eff}(v,e)}\,.
\end{align*}
Now, note that, because $R_{eff}$ is a metric,
\begin{align*}
R_{eff}(v,u)
& \geq 
R_{eff}(v,w) - |R_{eff}(v,w) - R_{eff}(v,u)|\\
& \geq
R_{eff}(v,w) - r_e\\
& \geq
R_{eff}(v,w) - \frac{1}{36} R_{eff}(C,e)\\
& \geq 
\frac{35}{36} R_{eff}(v,w)\,,
\end{align*}
where we used 
the triangle inequality twice,
(\ref{eq:st_projection1_condition}),
and also the fact that $R_{eff}(v,w) \geq R_{eff}(C,e)$ because $v\in C$ and $w\in e$.
Now, note also that
\[ R_{eff}(v,e) \geq \frac{1}{2} \min\{R_{eff}(v,w),R_{eff}(v,u) \geq \frac{1}{2}\cdot \frac{35}{36} R_{eff}(v,w)\}\,,\]
so
\begin{align*}
& 2 \sum\limits_v (R_{eff}(v,w))^{1/2} \cdot (p_v^{C}(u) + p_v^{C}(w)) \cdot \frac{\sqrt{r_e}}{R_{eff}(v,e)}\\
& \leq 2\sqrt{2\cdot\frac{36}{35}} \sum\limits_v (p_v^{C}(u) + p_v^{C}(w)) \cdot \sqrt{\frac{r_e}{R_{eff}(v,e)}}\\
& \leq 6\cdot \sqrt{\frac{r_e}{R_{eff}(C,e)}}\,,
\end{align*}
where we used the fact that 
$\sum\limits_{v\in C} (p_v^{C}(u) + p_v^{C}(w)) = 2$ and $R_{eff}(v,e) \geq R_{eff}(C,e)$ because $v\in C$.
\end{proof}

\subsection{Proof of Lemma~\ref{lem:important_edges}}
\label{proof_lem_important_edges}
\begin{proof}
By definition of the fact that $e$ is not $\eps$-important, we have
\[ R_{eff}(C,e) > r_e / \eps^2 \,.\]
Using the fact that the demand $\vpi^C\left(\BB^\top \frac{\pp}{\sqrt{\rr}}\right)$
is supported on $C$ and Lemma~\ref{st_projection1_energy}, we get
\begin{align*}
 \left|\left\langle \onev_e, \RR^{-1/2} \BB \vphi^*\right\rangle\right| 
 &= \left|\left\langle \vpi^C\left(\BB^\top \frac{\onev_e}{\sqrt{\rr}}\right),
 \vphi_C^*\right\rangle\right| \\
 & \leq \sqrt{\mathcal{E}_{\rr}\left(\vpi^C\left(\BB^\top \frac{\onev_e}{\sqrt{\rr}}\right)\right)} 
 \cdot \sqrt{E_{\rr}(\vphi^*)}\\
 & = \sqrt{\mathcal{E}_{\rr}\left(\vpi^C\left(\BB^\top \frac{\onev_e}{\sqrt{\rr}}\right)\right)} 
 \cdot \delta \sqrt{\mathcal{E}_{\rr}\left(\vpi^C\left(\BB^\top \frac{\pp}{\sqrt{\rr}}\right)\right)}\\
 & \leq 6 \sqrt{\frac{r_e}{R_{eff}(C,e)}} \cdot \delta \sqrt{m}\\
 & \leq 6\eps\,.
\end{align*}
\end{proof}

\subsection{Proof of Lemma~\ref{lem:projection_change_energy}}
\label{proof_lem_projection_change_energy}
\begin{proof}
We note that 
\begin{align*}
& \sqrt{\mathcal{E}_{\rr}\left(
\vpi^{C\cup\{v\}}(\BB^\top \frac{\qq}{\sqrt{\rr}}) 
-  \vpi^{C}(\BB^\top \frac{\qq}{\sqrt{\rr}})\right)}\\
& = \left|\pi_v^{C\cup\{v\}}(\BB^\top \frac{\qq}{\sqrt{\rr}})\right| \sqrt{R_{eff}(C,v)} \\
& \leq \sum\limits_{e=(u,w)\in E} 
\left|\pi_v^{C\cup\{v\}}(\BB^\top \frac{q_e}{\sqrt{r_e}} \onev_e)\right| \sqrt{R_{eff}(C,v)} \\
& \leq \sum\limits_{e=(u,w)\in E} 
(p_v^{C\cup\{v\}}(u) + p_v^{C\cup\{v\}}(w)) \cdot
\min\left\{\sqrt{\frac{R_{eff}(C,v)}{r_e}}, 
\frac{\sqrt{r_e R_{eff}(C,v)}}{R_{eff}(e,v)}\right\}
\,,
\end{align*}
where we used Lemma~\ref{st_projection1}.
For some sufficiently large $c$ to be defined later, we 
partition $E$ into $X$ and $Y$, where
$X = \{e\in E\ |\ R_{eff}(C,v) \leq c^2 \cdot r_e \text{ or } r_e R_{eff}(C,v) \leq c^2 \cdot (R_{eff}(e,v))^2 \}$
and $Y = E\backslash X$. We first note that 
\begin{align*}
&  \sum\limits_{e=(u,w)\in X} 
(p_v^{C\cup\{v\}}(u) + p_v^{C\cup\{v\}}(w)) \cdot
\min\left\{\sqrt{\frac{R_{eff}(C,v)}{r_e}}, 
\frac{\sqrt{r_e R_{eff}(C,v)}}{R_{eff}(e,v)}\right\}\\
& \leq c \cdot  \sum\limits_{e=(u,w)\in X} 
(p_v^{C\cup\{v\}}(u) + p_v^{C\cup\{v\}}(w))\\
& \leq c\cdot \tO{\beta^{-2}}\,,
\end{align*}
where the last inequality follows by the congestion reduction property.

Now, let $e=(u,w)\in Y$. We will prove that both $u$ and $w$ are much closer to $v$ than $C$.
This, in turn, will imply that their hitting probabilities on $v$ are roughly the same,
and so they mostly cancel out in the projection.

First of all, we let 
$R_{eff}(C,v) = c_1^2 \cdot r_e$
and 
$r_e R_{eff}(C,v) = c_2^2 \cdot (R_{eff}(e,v))^2$,
for some $c_1,c_2 > 0$, where by definition $c_1,c_2\geq c$.
Now, we assume without loss of generality that $R_{eff}(u,v) \leq R_{eff}(w,v)$, and so
\begin{align*}
R_{eff}(u,v) \leq 2 R_{eff}(e,v) = \frac{2}{c_2} \sqrt{r_e R_{eff}(C,v)} 
= \frac{2}{c_1c_2} R_{eff}(C,v)
\leq 
 \frac{2}{c_1c_2} \left(R_{eff}(C,u) + R_{eff}(u,v)\right)
\,,
\end{align*}
so 
\begin{equation}
\begin{aligned}
R_{eff}(u,v) 
\leq 
 \frac{\frac{2}{c_1c_2}}{1 - \frac{2}{c_1c_2}} R_{eff}(C,u)
 = \frac{2}{c_1c_2 - 2} R_{eff}(C,u)
 \leq \frac{3}{c_1c_2} R_{eff}(C,u)
\,.
\end{aligned}
\label{eq:energy_Reff_1}
\end{equation}
Futhermore, note that
\[ R_{eff}(w,v) \leq R_{eff}(u,v) + r_e \leq 
\left(\frac{2}{c_1c_2} + \frac{1}{c_1^2}\right) R_{eff}(C,v) 
\leq \left(\frac{2}{c_1c_2} + \frac{1}{c_1^2}\right) (R_{eff}(C,w) + R_{eff}(w,v))
\,, \]
and so we have 
\begin{equation}
\begin{aligned}
 R_{eff}(w,v) 
\leq \frac{\frac{2}{c_1c_2} + \frac{1}{c_1^2}}{
	1 - \frac{2}{c_1c_2} - \frac{1}{c_1^2}
	} R_{eff}(C,w)
\leq 3\left(\frac{1}{c_1c_2} + \frac{1}{c_1^2}\right)
 R_{eff}(C,w)
\,. 
\end{aligned}
\label{eq:energy_Reff_2}
\end{equation}
Now, by Lemma~\ref{lem:escape_prob} together with
(\ref{eq:energy_Reff_1}) we have
\[ p_v^{C\cup\{v\}}(u) \geq 1 - \frac{3}{c_1c_2} \]
and 
with (\ref{eq:energy_Reff_2}) we have
\[ p_v^{C\cup\{v\}}(w) \geq 1 - 3 \left(\frac{1}{c_1c_2} + \frac{1}{c_1^2}\right)\,, \]
therefore
\[ \left|p_v^{C\cup\{v\}}(u) - p_v^{C\cup\{v\}}(w)\right| \leq 6\left(\frac{1}{c_1c_2} + \frac{1}{c_1^2}\right)\,. \]
So,
\begin{align*}
 \left|\pi_v^{C\cup\{v\}}(\BB^\top \frac{q_e}{\sqrt{r_e}} \onev_e)\right| \sqrt{R_{eff}(C,v)}
& =
\left|
p_v^{C\cup\{v\}}(u)
- p_v^{C\cup\{v\}}(w)
\right| \sqrt{\frac{R_{eff}(C,v)}{r_e}}\\
& \leq 6 \left(\frac{1}{c_1c_2} + \frac{1}{c_1^2}\right) \cdot c_1\\
& = 6 \left(\frac{1}{c_2} + \frac{1}{c_1}\right)\\
& = O\left(\frac{1}{c}\right)\,.
\end{align*}

Now, we will apply Lemma~\ref{lem:small-neighborhood2} 
to prove that with high probability 
$|Y| \leq \tO{\beta^{-1}}$. The reason we can apply the lemma
is that for any $e=(u,w)\in Y$,
we have 
\[ R_{eff}(e,v) = \frac{1}{c_1c_2}R_{eff}(C,v) \leq \frac{1}{2} R_{eff}(C,v)\,,\]
and so $Y\subseteq N_E(v,R_{eff}(C,v) / 2)$.
Therefore, we get that
\begin{align*}
\sum\limits_{e=(u,w)\in Y} 
\left|\pi_v^{C\cup\{v\}}(\BB^\top \frac{q_e}{\sqrt{r_e}} \onev_e)\right|
\sqrt{R_{eff}(C,v)}
\leq c^{-1} \cdot \tO{\beta^{-1}}\,.
\end{align*}
Overall, we conclude that 
\[ \sqrt{\mathcal{E}_{\rr}\left(
\vpi^{C\cup\{v\}}(\BB^\top \frac{\qq}{\sqrt{\rr}}) 
-  \vpi^{C}(\BB^\top \frac{\qq}{\sqrt{\rr}})\right)}
\leq (c + c^{-1}) \tO{\beta^{-2}}
= \tO{\beta^{-2}}\,,
\]
by setting $c$ to be a large enough constant.
\end{proof}

\subsection{Proof of Lemma~\ref{lem:old_projection_approximate}}
\label{sec:proof_old_projection_approximate}
\begin{proof}

We write
\begin{align*}
&  \vpi^{C^T,\rr^T}\left(\BB^\top\frac{\qq^T}{\sqrt{\rr^T}}\right) - \vpi^{C^0,\rr^0}\left(\BB^\top\frac{\qq^0}{\sqrt{\rr^0}}\right) \\
& =
\underbrace{\sum\limits_{\text{$i$ is an \textsc{AddTerminal}}}
\pi_{v^i}^{C^{i+1},\rr^i}\left(\BB^\top \frac{\qq^{i}}{\sqrt{\rr^{i}}}\right) \cdot \left(\onev_{v^i} - \vpi^{C^i,\rr^i}\left(\onev_{v^i}\right)\right)}_{\dd_{Add}}\\
& + \underbrace{\sum\limits_{i \text{ is an \textsc{Update}}} \vpi^{C^i,\rr^i}\left(\BB^\top \left(\frac{\qq^{i+1}}{\sqrt{\rr^{i+1}}} - \frac{\qq^i}{\sqrt{\rr^i}} \right) \onev_{e^i}\right)}_{\dd_{Upd}}\,,
\end{align*}
which implies that 
\begin{align*}
\sqrt{\mathcal{E}_{\rr^T} \left(
 \vpi^{C^T,\rr^T}\left(\BB^\top\frac{\qq^T}{\sqrt{\rr^T}}\right) - \vpi^{C^0,\rr^0}\left(\BB^\top\frac{\qq^0}{\sqrt{\rr^0}}\right) 
\right)} \leq \sqrt{\mathcal{E}_{\rr^T}(\dd_{Add})}
 + \sqrt{\mathcal{E}_{\rr^T}(\dd_{Upd})}\,.
\end{align*}
We bound each of these terms separately.
For the second one, we have that
\begin{align*}
& \sqrt{\mathcal{E}_{\rr^T}\left(\dd_{Upd}\right)} \\
& \leq 
	\sum\limits_{i \text{ is an \textsc{Update}}} 
	\sqrt{\mathcal{E}_{\rr^T}\left(\vpi^{C^i,\rr^i}\left(\BB^\top \left(\frac{\qq^{i+1}}{\sqrt{\rr^{i+1}}} - \frac{\qq^{i}}{\sqrt{\rr^{i}}}\right) \onev_{e^i}\right)\right)}\\
& =
	\sum\limits_{i \text{ is an \textsc{Update}}} 
	\sqrt{\mathcal{E}_{\rr^T}\left(\BB^\top \left(\frac{\qq^{i+1}}{\sqrt{\rr^{i+1}}} - \frac{\qq^{i}}{\sqrt{\rr^{i}}}\right) \onev_{e^i}\right)}\\
& \leq 
	\sum\limits_{i \text{ is an \textsc{Update}}} 
	\left(\sqrt{\mathcal{E}_{\rr^T}\left(\BB^\top \frac{\qq^{i+1}}{\sqrt{\rr^{i+1}}} \onev_{e^i}\right)}
	+ 
	\sqrt{\mathcal{E}_{\rr^T}\left(\BB^\top \frac{\qq^{i}}{\sqrt{\rr^{i}}} \onev_{e^i}\right)}\right)\\
& \leq 
	\max_i \left\|\frac{\rr^T}{\rr^i}\right\|_\infty^{1/2} 
	\sum\limits_{i \text{ is an \textsc{Update}}} 
	\left(\sqrt{\mathcal{E}_{\rr^{i+1}}\left(\BB^\top \frac{\qq^{i+1}}{\sqrt{\rr^{i+1}}} \onev_{e^i}\right)}
	+ 
	\sqrt{\mathcal{E}_{\rr^i}\left(\BB^\top \frac{\qq^{i}}{\sqrt{\rr^{i}}} \onev_{e^i}\right)}\right)\\
& \leq 
	2 \max_i \left\|\frac{\rr^T}{\rr^i}\right\|_\infty^{1/2} T\,.
\end{align*}
For the first one, we have
\begin{align*}
& \sqrt{\mathcal{E}_{\rr^T}\left(\dd_{Add}\right)} \\
& \leq 
\sum\limits_{\text{$i$ is an \textsc{AddTerminal}}}
\left|\pi_{v^i}^{C^{i+1},\rr^i}\left(\BB^\top \frac{\qq^{i}}{\sqrt{\rr^{i}}}\right)\right| \cdot \sqrt{\mathcal{E}_{\rr^T}\left(\onev_{v^i} - \vpi^{C^i,\rr^i}\left(\onev_{v^i}\right)\right)}\\
& \leq 
\left\|\frac{\rr^T}{\rr^i}\right\|_{\infty}^{1/2}
\sum\limits_{\text{$i$ is an \textsc{AddTerminal}}}
\left|\pi_{v^i}^{C^{i+1},\rr^i}\left(\BB^\top \frac{\qq^{i}}{\sqrt{\rr^{i}}}\right)\right| \cdot \sqrt{\mathcal{E}_{\rr^i}\left(\onev_{v^i} - \vpi^{C^i,\rr^i}\left(\onev_{v^i}\right)\right)}\\
& \leq 
\tO{\max_i \left\|\frac{\rr^T}{\rr^i}\right\|_{\infty}^{1/2}\beta^{-2}} \cdot T\,,
\end{align*}
where in the last inequality we used Lemma~\ref{lem:projection_change_energy}.
The desired statement now follows immediately.

\end{proof}

\subsection{Proof of Lemma~\ref{lem:locator}}
\label{sec:full_proof_lem_locator}

\begin{proof}
\noindent {\bf $\textsc{Initialize}(\ff)$}: 
We set $\ss^+ = \uu - \ff$, $\ss^- = \ff$, $\rr^0 = \rr = \frac{1}{(\ss^+)^2} + \frac{1}{(\ss^-)^2}$.

We first initialize a $\beta$-congestion reduction
subset $C$ based on Lemma~\ref{lem:cong_red}, which takes time
$\tO{m \beta^{-2}}$,
and a data structure $\textsc{DynamicSC}$ for maintaining the sparsified Schur Complement onto $C$,
as described in Appendix~\ref{sec:maintain_schur}, which
takes time $\tO{m \beta^{-4} \eps^{-4}}$. We also set $C^0 = C$.

Then, we generate an $\tO{\eps^{-2}}\times m$
sketching matrix $\QQ$ as in~(Lemma 5.1, \cite{gao2021fully} v2),
which takes time $\tO{m \eps^{-2}}$, and let its rows be $\qq^i$
for $i\in[\tO{\eps^{-2}}]$.

In order to compute the set of important edges, we use 
Lemma~\ref{lem:approx_effective_res}
after contracting $C$, which shows that we can compute all resistances
of the form $R_{eff}(C,u)$ for $u\in V\backslash C$ up to a factor of $2$ in $\tO{m}$.
From these, we can get $4$-approximate estimates of $R_{eff}(C,e)$ for 
$e\in E\backslash E(C)$, using the fact that 
\[ \min\{R_{eff}(C,u),R_{eff}(C,w)\} \approx_{2} R_{eff}(C,e) \,.\]

Then, in $O(m)$ time, we can easily compute a set of edges $S$
such that
\begin{align*}
\{e\ |\ e \text{ is $\frac{\eps\beta}{\alpha}$-important } \}
\subseteq S\subseteq \{e\ |\ e \text{ is $\frac{\eps \beta}{4\alpha}$-important}\}
\end{align*}

We also need to sample the random walks that will be used inside the demand projection
data structures. We use (Lemma 5.15,~\cite{gao2021fully} v2) to sample 
$h = \tO{\heps^{-4} \beta^{-6} + \heps^{-2} \beta^{-2} \gamma^{-2}}$
random walks for each $u\in V\backslash C$ and $e\in E\backslash E(C)$ with $u\in e$,
where we set $\gamma = \frac{\eps}{4\alpha}$ so that $S$ is a 
subset of $\gamma$-important edges. Note that, by Definition~\ref{def:locator},
a $\gamma$-important edge will always remain $\gamma$-important until
the \textsc{Locator} is re-initialized, as any edge's resistive
distance to $C$ can only decrease, and its own resistance is constant. Therefore
$S$ can be assumed to always be a subset of $\gamma$-important edges. 

The runtime to sample the set $\mathcal{P}$ of these random walks is 
\[ \tO{m h \beta^{-2}} 
= \tO{m (\heps^{-4} \beta^{-8} + \heps^{-2} \beta^{-4} \gamma^{-2})}
= \tO{m (\heps^{-4} \beta^{-8} + \heps^{-2} \eps^{-2}\alpha^2 \beta^{-4} )}\,.
\]

In order to be able to detect congested edges, we will initialize $\tO{\eps^{-2}}$
demand projection data structures, with the guarantees from Lemma~\ref{lem:ds}.
We will maintain 
an approximation to $\vpi^C(\BB^\top \frac{\qq_S^i}{\sqrt{\rr}})$ for all $i\in\tO{\eps^{-2}}$,
where $\qq^i$ are the rows of the sketching matrix that we have generated,
as well as 
$\vpi^{old} := \vpi^{C^0,\rr^0}\left(\BB^\top \frac{\pp^0}{\sqrt{\rr^0}}\right)$, where
$\pp^0 = \sqrt{\rr^0} g(\ss^0) = \frac{\frac{1}{\ss^{+,0}} - \frac{1}{\ss^{-,0}}}{\sqrt{\rr^0}}$.

Specifically, we call 
\[ \textsc{DP}^i.\textsc{Initialize}(C, \rr, \qq^i, S, \mathcal{P}) \]
for all $i\in[\tO{\eps^{-2}}]$, and also exactly
compute $\vpi^{old}$, which can be done by calling
\[ \textsc{DemandProjector}.\textsc{Initialize}(C, \rr, \pp, [m], \mathcal{P})\,. \]

The total runtime for this operation is dominated by the random walk generation, and is 
\[\tO{m (\heps^{-4} \beta^{-8} + \heps^{-2} \eps^{-2}\alpha^2 \beta^{-4} )}\,.\]

\noindent {\bf $\textsc{Update}(e,\ff)$}: 
We set $s_e^+ = u_e - f_e$, $s_e^- = f_e$, and $r_e = \frac{1}{(s_e^+)^2} + \frac{1}{(s_e^-)^2}$.
Then, we also set $p_e = \frac{\frac{1}{s_e^+} - \frac{1}{s_e^-}}{\sqrt{r_e}}$.

We distinguish two cases: 
\begin{itemize}
\item{$e\in E(C)$: 

In this case, we can simply call 
\[ \textsc{DynamicSC}.\textsc{Update}(e, r_e) \]
and
\[ \textsc{DP}^i.\textsc{Update}(e, \rr, \qq) \]
for all $i\in[\tO{\eps^{-2}}]$.
Note that we can do this as
$\textsc{DP}^i$ was initialized with resistances $\rr^0$
and $\rr^0 \approx_{\alpha} \rr$.
}
\item{$e\in E\backslash E(C)$:

We let $e = (u,w)$. We want to insert $u$ and $w$ into $C$, but for doing that
$\textsc{DP}^i$'s require constant factor estimates of the resistances
$R_{eff}(C,u)$ and $R_{eff}(C\cup\{u\},w)$. In order to get these estimates,
we will use $\textsc{DynamicSC}$.

We first call
\[ \textsc{DynamicSC}.\textsc{AddTerminal}(u)\,, \]
which takes time $\tO{\beta^{-2} \eps^{-2}}$ and returns 
$\tR_{eff}(C,u) \approx_2 R_{eff}(C,u)$. Given this estimate,
we can call
\[ \textsc{DP}^i.\textsc{AddTerminal}(u, \tR_{eff}(C,u))\,, \]
for all $i\in[\tO{\eps^{-2}}]$, each
of which takes time 
\[ \tO{\heps^{-4} \beta^{-8} + \heps^{-2} \beta^{-6} \gamma^{-2}}
= \tO{\heps^{-4} \beta^{-8} + \heps^{-2} \eps^{-2} \alpha^2 \beta^{-6}} \,.\]

Now, we can set $C = C\cup\{u\}$ and repeat the same process
for $w$.

Finally, to update the resistance, note that we now have 
$e\in E(C)$, so we apply the procedure from the first case.

}
\end{itemize}

Finally, if the total number of calls to 
$\textsc{DP}^i.\textsc{AddTerminal}$ for some fixed $i$
since the last call to $\cL.\textsc{BatchUpdate}(\emptyset)$
exceeds $\frac{\eps}{\heps \alpha^{1/2}}$ (note that 
the
number of calls is actually the same for all $i$), 
we call $\cL.\textsc{BatchUpdate}(\emptyset)$ in order to
re-initialize the demand projections.

We conclude that the total runtime is 
\[ \tO{m \frac{\heps \alpha^{1/2}}{\eps^{3}} 
+ \heps^{-4} \eps^{-2} \beta^{-8} +
\heps^{-2} \eps^{-4} \alpha^2 \beta^{-6}} \,,\]
where the first term comes from amortizing the calls
to $\cL.\textsc{BatchUpdate}(\emptyset)$, each of which,
as we will see, takes $\tO{m \eps^{-2}}$.

\noindent {\bf $\textsc{BatchUpdate}(Z,\ff)$}: 
First, for each $e\in Z$, we set
$s_e^+ = u_e - f_e$, $s_e^- = f_e$,
$r_e^0 = r_e = \frac{1}{(s_e^+)^2} + \frac{1}{(s_e^-)^2}$,
and $p_e^0 = p_e = \frac{\frac{1}{s_e^+} - \frac{1}{s_e^-}}{\sqrt{r_e}}$.

For each $e=(u,w)\in Z$, we call
\[ \textsc{DynamicSC}.\textsc{AddTerminal}(u)\]
and
\[ \textsc{DynamicSC}.\textsc{AddTerminal}(w)\]
(if $u$ and $w$ are not already in $C$), and then
we call
\[ \textsc{DynamicSC}.\textsc{Update}(e, r_e) \,. \]
Then, we set $C^0 = C = C\cup (\cup_{(u,w)\in Z}\{u,w\})$.
Additionally, we re-compute $\vpi^{old}$ based on the new
values of $C^0, \rr^0, \pp^0$.
All of this takes time $\tO{m + |Z| \beta^{-2} \eps^{-2}}$.

Now, to pass these updates to the $\textsc{DemandProjector}$s, we first
have to re-compute the set of important edges $S$ 
(with the newly updated resistances)
as any set such that
\begin{align*}
\{e\ |\ e \text{ is $\frac{\eps}{\alpha}$-important } \}
\subseteq S\subseteq \{e\ |\ e \text{ is $\frac{\eps}{4\alpha}$-important}\}\,.
\end{align*}
As we have already argued, this takes $\tO{m}$.

Now, finally, we re-initialize all the $\textsc{DemandProjector}$s by calling
\[ \textsc{DP}^i.\textsc{Initialize}(C, \rr, \qq, S, \mathcal{P})\,. \]
for all $i\in[\tO{\eps^{-2}}]$, where each call takes $\tO{m}$.

We conclude with a total runtime of 
\[ \tO{m\eps^{-2} + |Z| \beta^{-2} \eps^{-2}} \,. \]

\noindent {\bf $\textsc{Solve}()$}: This operation performs
the main task of the locator,  which is to detect congested edges. We will do that by using
the approximate demand projections that we have been maintaining.

We remind that the congestion vector we are trying to approximate to 
$O(\epsilon)$ additive accuracy is
\begin{align*}
\vrho^* = \delta \sqrt{\rr}g(\ss) - \delta \RR^{-1/2} \BB\LL^+ \BB^\top g(\ss)\,.
\end{align*}
We will first reduce the problem of finding the entries of 
$\vrho^*$ with magnitude $\geq \Omega(\eps)$, to the problem of computing an 
$O(\eps)$-additive approximation to
\[ v_i^* = \delta \cdot \left\langle \vpi^C\left(\BB^\top \frac{\qq_S^i}{\sqrt{\rr}}\right), \tSC^+ 
\vpi^{C^0,\rr^0}\left(\BB^\top \frac{\pp^0}{\sqrt{\rr^0}}\right)\right\rangle \]
for all $i\in[\tO{\eps^{-2}}]$,
where $\tSC$ is the approximate Schur complement maintained in $\textsc{DynamicSC}$.
Then, we will see how to approximate $v_i^*$ to additive accuracy 
$O(\eps)$ using the demand projection data structures.

First of all, note that, by definition of $g(\ss) = \frac{\frac{1}{\ss^+} - \frac{1}{\ss^-}}{\rr}$,
\[ \left\|\delta \sqrt{\rr} g(\ss)\right\|_\infty \leq \delta \leq \eps\,, \]
so this term can be ignored.

Using Lemma~\ref{lem:non-projected-demand-contrib}, we get that
\[ \delta \left\|\RR^{-1/2} \BB \LL^+ 
(\BB^\top g(\ss) - \vpi^{C}(\BB^\top g(\ss)))\right\|_\infty 
\leq \delta \cdot \tO{\beta^{-2}} \leq \eps / 2 \,. \]
This means that the entries of the vector 
\[ \delta \RR^{-1/2} \BB \LL^+ 
\vpi^{C}(\BB^\top g(\ss)) \]
that have magnitude $\leq \eps$ do not correspond to 
the $\Omega(\eps)$-congested edges that we are looking for.

Now, we set $T = |C \backslash C^0|$, where
$C^0$ was the congestion reduction subset during the last call to $\textsc{BatchUpdate}$,
and apply Lemma~\ref{lem:old_projection_approximate}. This shows that 
\[ \sqrt{\mathcal{E}_{\rr}\left(\vpi^{old} - \vpi^{C}\left(\BB^\top g(\ss)\right)\right)
} \leq \tO{\alpha^{1/2}\beta^{-2}} \cdot T \,.\]
Therefore, if we define
\[ \vrho = -\delta\RR^{-1/2}\BB \LL^+ \vpi^{old}\,, \] 
we conclude that
\[ \left\|\vrho - \vrho^*\right\|_\infty
\leq O(\eps) + 
\delta \left\|\RR^{-1/2} \BB \LL^+ 
\left(\vpi^{old} - \vpi^{C}\left(\BB^\top g(\ss)\right)\right) \right\|_\infty
\leq O(\eps) + \delta T \cdot \tO{\alpha^{1/2} \beta^{-2}}
\leq O(\eps)\,,
\]
where we used the fact that $T = \frac{\eps}{\heps \alpha^{1/2}} \leq 
\frac{\eps}{\delta \beta^{-2} \alpha^{1/2}}$.
Therefore it suffices to estimate $\vrho$ up to $O(\eps)$-additive accuracy.

Now, note that, by definition, no edge $e\in E\backslash S$ %
is $\eps / \alpha$-important with respect to $\rr^0$ and $C^0$.
By using Lemma~\ref{st_projection1_energy}, for each such edge we get
\begin{align*}
& \delta \left|\RR^{-1/2} \BB \LL^+ \vpi^{old}\right|_e\\
& \leq \delta \sqrt{\mathcal{E}_{\rr}\left(\vpi^{C^0}\left(\BB^\top \frac{\onev_e}{\sqrt{\rr}}\right)\right)}\sqrt{\mathcal{E}_{\rr}(\vpi^{old})}\\
& \leq \delta \alpha \sqrt{\mathcal{E}_{\rr^0}\left(\vpi^{C^0}\left(\BB^\top \frac{\onev_e}{\sqrt{\rr}}\right)\right)}\sqrt{\mathcal{E}_{\rr^0}(\vpi^{old})}\\
& \leq \delta \alpha \cdot \frac{\eps}{\alpha} \cdot O(\sqrt{m})\\
& = O(\eps) \,,
\end{align*}
where we also used the fact that
$\mathcal{E}_{\rr^0}(\vpi^{C^0,\rr^0}(g(\ss^0)))
\leq O(\mathcal{E}_{\rr^0}(g(\ss)))$.
Therefore it suffices to approximate
\[ \vrho' = \delta \II_S \RR^{-1/2} \BB \LL^+ \vpi^{old} \,.\]

Note that here we can replace 
$\LL^+$ by 
$\begin{pmatrix}-\LL_{FF}^{-1} \LL_{FC} \\ \II\end{pmatrix} \tSC^+$
where $\tSC \approx_{1+\eps} SC$ and only lose another additive
$\eps$ error, as 
\begin{align*}
& \delta \left|\left\langle 
\onev_e, \RR^{-1/2} \BB \LL^+ \vpi^{old}\right\rangle
- \delta \left\langle \onev_e, 
\RR^{-1/2} \BB \begin{pmatrix}-\LL_{FF}^{-1} \LL_{FC}\\ \II \end{pmatrix}\tSC^+ \vpi^{old}\right\rangle\right|\\
& 
= \delta \left|\left\langle \vpi^{C}\left(\BB^\top \frac{\onev_e}{\sqrt{\rr}}\right), 
\left(SC^+ - \tSC^+\right) \vpi^{old}\right\rangle\right|\\
& \leq O(\delta \eps \cdot \sqrt{m})\\
& \leq O(\eps) \,,
\end{align*}
where we used the fact that 
\[ (1-\eps) SC \preceq \tSC \preceq (1+\eps) SC \Rightarrow 
-O(\eps) \tSC^+ \preceq SC^+ - \tSC^+\preceq O(\eps) \tSC^+ \,.\]
and that
\begin{align*}
\sqrt{\mathcal{E}_{\rr}\left(\vpi^{old}\right)}
& \leq 
\sqrt{\mathcal{E}_{\rr}\left(\vpi^C\left(\BB^\top g(\ss)\right)\right)}
+ \sqrt{\mathcal{E}_{\rr}\left(\vpi^{old} - \vpi^C\left(\BB^\top g(\ss)\right)\right)}\\
& \leq O(m) + \tO{\alpha^{1/2} \beta^{-2}} \cdot T\\
& \leq O(m)\,,
\end{align*}
where we used the fact that $T = \frac{\eps}{\heps \alpha^{1/2}} 
\leq \frac{\eps}{\delta \beta^{-2}\alpha^{1/2}}
\leq \frac{\sqrt{m}}{\beta^{-2}\alpha^{1/2}}
$.

Now, we will use the sketching lemma (Lemma 5.1, \cite{gao2021fully} v2),
which shows that in order to find all entries of 
\[ \II_S \RR^{-1/2} \BB \begin{pmatrix}
-\LL_{FF} \LL_{FC} \\ \II\end{pmatrix}\tSC^+ \vpi^{old} \]
with magnitude $\Omega(\eps)$, it suffices to compute the inner
products
\begin{align*}
& \delta \left\langle \BB^\top \frac{\qq_S^i}{\sqrt{\rr}}, 
\begin{pmatrix}-\LL_{FF}^{-1} \LL_{FC} \\ \II \end{pmatrix}
\tSC^+ 
\vpi^{old}\right\rangle \\
& = \delta \left\langle \vpi^C\left(\BB^\top \frac{\qq_S^i}{\sqrt{\rr}}\right), 
\tSC^+ 
\vpi^{old} \right\rangle
\end{align*}
for $i\in[\tO{\eps^{-2}}]$,
up to additive accuracy
\[ \eps \cdot \left\|\delta \II_S \RR^{-1/2} \BB 
\begin{pmatrix}-\LL_{FF}^{-1} \LL_{FC} \\ \II \end{pmatrix}
\tSC^+ \vpi^{old} \right\|_2^{-1} 
\geq \Omega(\eps) \,,
\]
where we used the fact that 
\begin{align*}
& \left\|\delta \II_S \RR^{-1/2} \BB  \begin{pmatrix}-\LL_{FF}^{-1} \LL_{FC} \\ \II \end{pmatrix}
\tSC^+ \vpi^{old} \right\|_2^2\\
& = \delta^2 
\langle 
\tSC^+ 
\vpi^{old},
\begin{pmatrix}-\LL_{CF} \LL_{FF}^{-1} & \II \end{pmatrix}
\LL
\begin{pmatrix}-\LL_{FF}^{-1} \LL_{FC} \\ \II \end{pmatrix}
\tSC^+ \vpi^{old} \rangle\\
& = \delta^2 
\left\langle 
\tSC^+ 
\vpi^{old},
\begin{pmatrix}-\LL_{CF} \LL_{FF}^{-1} & \II \end{pmatrix}
\begin{pmatrix}\LL_{FF} & \LL_{FC} \\ \LL_{CF} & \LL_{CC}\end{pmatrix}
\begin{pmatrix}-\LL_{FF}^{-1} \LL_{FC} \\ \II \end{pmatrix}
\tSC^+ \vpi^{old} \right\rangle\\
& = \delta^2 
\left\langle 
\tSC^+ 
\vpi^{old},
SC
\tSC^+ \vpi^{old} \right\rangle\\
& \leq 2
\delta^2 
\left\langle 
\vpi^{old},
\tSC^+ \vpi^{old} \right\rangle\\
& \leq O(\delta^2 m)\\
& = O(1)\,.
\end{align*}

Now, for the second part of the proof, we would like to compute $\vv$ such that
$\left\|\vv-\vv^*\right\|_\infty \leq O(\eps)$, where we remind that
\[ \vv^* = \left\langle \vpi^C\left(\BB^\top \frac{\qq_S^i}{\sqrt{\rr}}\right), \tSC^+ \vpi^{old}
\right\rangle \,. \]
Note that we already have estimates $\tvpi^C\left(\BB^\top \frac{\qq_S^i}{\sqrt{\rr}}\right)$
given by $\textsc{DP}^i$ for all $i\in[\tO{\eps^{-2}}]$. We obtain these estimates by calling
\[ \textsc{DP}^i.\textsc{Output}() \]
each of which takes time $O(\beta m)$.
By the guarantees of Definition~\ref{def:demand_projector}, with high probability we have
\begin{align*}
    & \delta \left|
     \left\langle \tvpi^C\left(\BB^\top \frac{\qq^i_S}{\sqrt{\rr}}\right)
	 -\vpi^C\left(\BB^\top \frac{\qq^i_S}{\sqrt{\rr}}\right), 
    \tSC^+
	\vpi^{old}
    \right\rangle 
    \right| \leq \heps \sqrt{\alpha} T\,,
\end{align*}
where we used the fact that 
\[ E_{\rr}(\delta \cdot \tSC^+ \vpi^{old}) \leq O(1)\,. \]
Now, since by definition $\textsc{BatchUpdate}$ is called
every $\frac{\eps}{\heps\alpha^{1/2}}$ calls to $\textsc{Update}$,
We have $T \leq \frac{\eps}{\heps\alpha^{1/2}}$ and so
\begin{align*}
    & \delta \left|
     \left\langle \tvpi^C\left(\BB^\top \frac{\qq^i_S}{\sqrt{\rr}}\right)
	 -\vpi^C\left(\BB^\top \frac{\qq^i_S}{\sqrt{\rr}}\right), 
    \tSC^+
	\vpi^{old}
    \right\rangle 
    \right| \leq \eps\,.
\end{align*}
This means that, running the algorithm from (Lemma 5.1, \cite{gao2021fully} v2),
we can obtain an edge set of size $\tO{\eps^{-2}}$ that contains all edges
such that
$\left|\rho_e^*\right| \geq c \cdot \eps$ for some constant $c > 0$.
By rescaling $\eps$ to get the right constant, we obtain all edges
such that $\left|\rho_e^*\right| \geq \eps / 2$ with high probability.
The runtime is dominated by the time to get $\tSC$ and apply its inverse,
and is $\tO{\beta m \eps^{-2}}$.

\paragraph{Success probability} We will argue that $\cL$ uses $\textsc{DynamicSC}$ and the $\textsc{DP}^i$ as an oblivious 
adversary. First of all, note that no randomness is injected into the inputs of
$\textsc{DynamicSC}$, as they are all coming form the inputs of $\cL$.

Regarding $\textsc{DP}^i$, note that its only output is given by the call to
$\textsc{DP}^i.\textsc{Output}$. However, note that its output is only used 
to estimate the inner product
\[ \left\langle \vpi^C\left(\BB^\top \frac{\qq_S^i}{\sqrt{\rr}}\right), \tSC^+ \vpi_{old} \right\rangle\,, \]
from which we obtain the set of congested edges and we directly return it
from $\cL$. Thus, it does not influence the state of $\cL,\textsc{DynamicSC}$ or any future inputs.

\end{proof}

\section{Deferred Proofs from Section~\ref{sec:demand_projection}}
\subsection{Proof of Lemma~\ref{st_projection1}}
\label{proof_st_projection1}
\begin{proof}
Let $\mathcal{P}_v(u)$ be a random walk that starts from $u$ and stops when it hits $v$.
\begin{align*}
p_v^{C\cup\{v,w\}}(u) 
& = \Pr\left[\mathcal{P}_v(u)\cap C = \emptyset \text{ and } w\notin \mathcal{P}_v(u)\right]\\
& =
\Pr\left[\mathcal{P}_v(u)\cap C = \emptyset\right] \cdot
\Pr\left[w\notin \mathcal{P}_v(u) \ |\ \mathcal{P}_v(u)\cap C = \emptyset \right]\\
& =
p_v^{C\cup\{v\}}(u)
\cdot 
\Pr\left[w\notin \mathcal{P}_v(u) \ |\ \mathcal{P}_v(u)\cap C = \emptyset \right]
\end{align*}

Consider new resistances $\hr$, where $\hr_e = r_e$ for all $e\in E$ not incident to $C$ and $\hr_e = \infty$ for all $e\in E$ incident to $C$.
Also, let $\hp$ be the hitting probability function for these new resistances.
It is easy to see that
\[
\Pr\left[w\notin \mathcal{P}_v(u) \ |\ \mathcal{P}_v(u)\cap C = \emptyset \right] = \hp_v^{\{v,w\}}(u)\,.
\]
Therefore, we have
\begin{align*}
p_v^{C\cup\{v,w\}}(u) 
& =
p_v^{C\cup\{v\}}(u)
\cdot 
\hp_v^{\{v,w\}}(u)\,.
\end{align*}

Now we will bound $\hp_v^{\{v,w\}}(u)$.
Let $\psi$ be electrical potentials for pushing $1$ unit of flow from $v$ to $w$ with resistances $\hr$ and let $f$ be
the associated electrical flow.
We have that 
\begin{align}
\left|\psi_u - \psi_w\right| = |f_e| \hr_e \leq \hr_e = r_e \label{1}
\end{align}
(because $|f_e| \leq 1$ and $e$ is not incident to $C$)
and
\begin{align}
 \left|\psi_v - \psi_w\right| = \widehat{R}_{eff}(v,w) \geq R_{eff}(v,w) \label{2} %
\end{align}

Additionally, by well known facts that connect electrical potential embeddings with random walks, we have that
\[ \psi_u = \psi_w + \hp_v^{\{v,w\}}(u) (\psi_v - \psi_w)\,, \]
or equivalently
\[ 
\hp_v^{\{v,w\}}(u) 
= \frac{\psi_u - \psi_w}{\psi_v - \psi_w} 
\,. \]
Using (\ref{1}) and (\ref{2}), this immediately implies that
\begin{align*}
\hp_v^{\{v,w\}}(u) \leq \frac{r_e}{R_{eff}(v,w)}\,.
\end{align*}
So we have proved that 
\begin{align*}
p_v^{C\cup\{v,w\}}(u) \leq p_v^{C\cup\{v\}}(u) \frac{r_e}{R_{eff}(v,w)}
\end{align*}
and symmetrically
\begin{align*}
p_v^{C\cup\{v,u\}}(w) \leq p_v^{C\cup\{v\}}(w) \frac{r_e}{R_{eff}(v,u)}\,.
\end{align*}

Now, let's look at $\pi_v^{C\cup\{v\}}(B^\top \onev_e) = p_v^{C\cup\{v\}}(u) - p_v^{C\cup\{v\}}(w)$. Note that 
\[ p_v^{C\cup\{v\}}(u) = p_v^{C\cup\{v,w\}}(u) + p_w^{C\cup\{v,w\}}(u) p_v^{C\cup\{v\}}(w) \]
which we re-write as 
\[ p_v^{C\cup\{v\}}(u) - p_v^{C\cup\{v\}}(w) = p_v^{C\cup\{v,w\}}(u) - (1-p_w^{C\cup\{v,w\}}(u)) p_v^{C\cup\{v\}}(w)\leq p_v^{C\cup\{v,w\}}(u)\,. \]
Symmetrically,
\[ p_v^{C\cup\{v\}}(w) - p_v^{C\cup\{v\}}(u) \leq p_v^{C\cup\{v,u\}}(w)\,. \]
From these we conclude that
\begin{align*}
\left|\pi_v^{C\cup\{v\}}(B^\top \onev_e)\right| 
& = \left|p_v^{C\cup\{v\}}(u) - p_v^{C\cup\{v\}}(w)\right| \\
& \leq \max\left\{p_v^{C\cup\{v,w\}}(u), p_v^{C\cup\{v,u\}}(w)\right\}\\
& \leq \max\left\{p_v^{C\cup\{v\}}(u)\cdot \frac{r_e}{R_{eff}(v,w)}, p_v^{C\cup\{v\}}(w)\cdot\frac{r_e}{R_{eff}(v,u)}\right\}\\
& \leq (p_v^{C\cup\{v\}}(u) + p_v^{C\cup\{v\}}(w)) \cdot \max\left\{\frac{r_e}{R_{eff}(v,w)}, \frac{r_e}{R_{eff}(v,u)}\right\}\,,
\end{align*}
which, after dividing by $\sqrt{r_e}$ gives
\begin{align*}
\left|\pi_v^{C\cup\{v\}}(B^\top \onev_e)\right| 
& \leq (p_v^{C\cup\{v\}}(u) + p_v^{C\cup\{v\}}(w)) \cdot \max\left\{\frac{\sqrt{r_e}}{R_{eff}(v,w)}, \frac{\sqrt{r_e}}{R_{eff}(v,u)}\right\}\,.
\end{align*}
\end{proof}

\subsection{Proof of Lemma~\ref{estimate1}}
\label{proof_estimate1}
\begin{proof}
For each $u\in V \backslash C$
and $e\in S'$ with $u\in e$,
we generate 
$Z$ random walks $P^1(u),\dots,P^Z(u)$ from $u$
to $C\cup\{v\}$.
We set
\begin{align*}
\tpi_v^{C\cup\{v\}}\left(\BB^\top \frac{\qq_{S'}}{\sqrt{\rr}}\right)
&= \sum\limits_{e=(u,w)\in S'} \sum\limits_{z=1}^Z \frac{1}{Z} \frac{q_e}{\sqrt{r_e}} \left(1_{\{v \in P^z(u)\}} 
- 1_{\{v\in P^z(w)\}}\right) \\
&= \sum\limits_{e=(u,w)\in S'} \sum\limits_{z=1}^Z (X_{e,u,z} - X_{e,w,z})
\,,
\end{align*}
where we have set
$X_{e,u,z} = \frac{1}{Z} \frac{q_e}{\sqrt{r_e}} 1_{\{v\in P^z(u)\}} $
and
$ X_{e,w,z} = -\frac{1}{Z} \frac{q_e}{\sqrt{r_e}} 1_{\{v\in P^z(w)\}}$.

Note that 
$\underset{P^z(u)}{\mathbb{E}}[X_{e,u,z}] = \frac{1}{Z} \frac{q_e}{\sqrt{r_e}} p_v^{C\cup\{v\}}(u)$
and
$\underset{P^z(w)}{\mathbb{E}}[X_{e,w,z}] = -\frac{1}{Z} \frac{q_e}{\sqrt{r_e}} p_v^{C\cup\{v\}}(w)$.
This implies that our estimate is unbiased, as
\[ \mathbb{E}\left[\tpi_v^{C\cup\{v\}}\left(\BB^\top \frac{\qq_{S'}}{\sqrt{\rr}}\right)\right] = \sum\limits_{e=(u,w)\in S'} \frac{q_e}{\sqrt{r_e}} (p_v^{C\cup\{v\}}(u) - p_v^{C\cup\{v\}}(w)) = \pi_v^{C\cup\{v\}}\left(\BB^\top \frac{\qq_{S'}}{\sqrt{\rr}}\right) \,.\]
We now need to show that our estimate is concentrated around the mean.
To apply the concentration bound in Lemma~\ref{conc1}, we need the following bounds:
\[ \sum\limits_{e=(u,w)\in S'} 
\sum\limits_{z=1}^Z \left(|\mathbb{E}[X_{e,u,z}]| + |\mathbb{E}[X_{e,w,z}]|\right) 
= \sum\limits_{e=(u,w)\in S'} \frac{|q_e|}{\sqrt{r_e}} (p_v^{C\cup\{v\}}(u) + p_v^{C\cup\{v\}}(w)) := E\]
\[ \underset{\substack{e=(u,w)\in S'\\ z\in[Z]}}{\max} \max\{|X_{e,u,z}|, |X_{e,w,z}|\} \leq 
\max_{e\in S'} \frac{1}{Z\sqrt{r_e}} := M \,.\]
So now for any $t\in [0,E]$ we have
\begin{align*}
& \Pr\left[\left|\tpi_v^{C\cup\{v\}}\left(\BB^\top\frac{\qq_{S'}}{\sqrt{\rr}}\right) - \pi_v^{C\cup\{v\}}(\BB^\top \frac{\qq_{S'}}{\sqrt{\rr}})\right| > t\right] \\
& \leq 2 \exp\left(-\frac{t^2}{6EM}\right)\\
& = 2 \exp\left(-\frac{Z t^2}{6\sum\limits_{e=(u,w)\in S'} \frac{|q_e|}{\sqrt{r_e}} (p_v^{C\cup\{v\}}(u) + p_v^{C\cup\{v\}}(w))\max_{e\in S'} \frac{1}{\sqrt{r_e}} }\right)\\
& \leq 2 \exp\left(-\frac{Z t^2}{6\sum\limits_{e=(u,w)\in S'} (p_v^{C\cup\{v\}}(u) + p_v^{C\cup\{v\}}(w)) \max_{e\in S'} \frac{1}{r_e} }\right)\\
& \leq 2 \exp\left(-\frac{Z t^2 c^2 R_{eff}(C,v)}{6 \sum\limits_{e=(u,w)\in S'} (p_v^{C\cup\{v\}}(u) + p_v^{C\cup\{v\}}(w)) }\right)\\
& \leq 2 \exp\left(-Z t^2 c^2 R_{eff}(C,v) / \tO{\beta^{-2}}\right)\\
& \leq \frac{1}{n^{100}}\,,
\end{align*}
where the last inequality follows 
by setting $Z = \tO{\frac{\log n \log \frac{1}{\beta}}{\delta_1'^2}}$ and 
$t = \frac{\delta_1'}{\beta c \sqrt{R_{eff}(C,v)}}$. Note that we have used the fact that
$R_{eff}(C,v) \leq r_e / c^2$ for all $e\in S'$, as well as the congestion reduction property
(Definition~\ref{def:cong_red})
\[ \sum\limits_{e=(u,w)\in E\backslash E(C)} (p_v^{C\cup\{v\}}(u) + p_v^{C\cup\{v\}}(w)) \leq \tO{1/\beta^2} \,.\]
\end{proof}

\subsection{Proof of Lemma~\ref{estimate1_final}}
\label{proof_estimate1_final}
\begin{proof}
In order to compute 
$\tpi_v^{C\cup\{v\}}(\BB^\top \frac{\qq_{S}}{\sqrt{\rr}})$
we use Lemma~\ref{estimate1}
with demand $\BB^\top \frac{\qq_{S'}}{\sqrt{\rr}}$ and error parameter 
$\delta_1' > 0$, 
where 
\[ \{e\in S\ |\ R_{eff}(C,v) \leq r_e / (2c^2) \} \subseteq
S' \subseteq \{e\in S\ |\ R_{eff}(C,v) \leq r_e / c^2 \}\,, \]
and $c > 0$ will be defined later.
Note that 
such a set $S'$ can be trivially computed
given our effective resistance estimate
$\tR_{eff}(C,v) \approx_{2} R_{eff}(C,v)$.
However, algorithmically we do not directly compute $S'$, but instead find
its intersection with the edges from which a sampled random walk 
ends up at $v$. (Using the congestion reduction property of $C$, this can be done in $\tO{\delta_1'^{-2} \beta^{-2} \log n\log\frac{1}{\beta}}$ time 
just by going through all random walks that contain $v$.)

Now, Lemma~\ref{estimate1} guarantees that
\begin{align*}
\left|
\tpi_v^{C\cup\{v\}}\left(\BB^\top \frac{\qq_{S'}}{\sqrt{\rr}}\right)
- \pi_v^{C\cup\{v\}}\left(\BB^\top \frac{\qq_{S'}}{\sqrt{\rr}}\right)
\right| \leq \frac{\delta_1'}{\beta c\sqrt{R_{eff}(C,v)}}
\end{align*}
given access to $O(\delta'^{-2} \log n \log \frac{1}{\beta})$ random walks
for each $u\in V\backslash C$, $e\in S'$ with $u\in e$.

Then, we set $\tpi_v^{C\cup\{v\}}(\BB^\top \frac{\qq_S}{\sqrt{\rr}}) := \tpi_v^{C\cup\{v\}}(\BB^\top \frac{\qq_{S'}}{\sqrt{\rr}})$, and we have that
\begin{equation}
\begin{aligned}
& \left|
\tpi_v^{C\cup\{v\}}\left(\BB^\top \frac{\qq_S}{\sqrt{\rr}}\right)
- \pi_v^{C\cup\{v\}}\left(\BB^\top \frac{\qq_S}{\sqrt{\rr}}\right)
\right| \\
& \leq 
\left|\tpi_v^{C\cup\{v\}}\left(\BB^\top \frac{\qq_{S'}}{\sqrt{\rr}}\right) - \pi_v^{C\cup\{v\}}\left(\BB^\top \frac{\qq_{S'}}{\sqrt{\rr}}\right)\right|
+ \left|\pi_v^{C\cup\{v\}}\left(\BB^\top \frac{\qq_S - \qq_{S'}}{\sqrt{\rr}}\right)
\right| \\
& \leq \frac{\delta_1'}{\beta c \sqrt{R_{eff}(C,v)}} 
+ \left|\pi_v^{C\cup\{v\}}\left(\BB^\top \frac{\qq_S - \qq_{S'}}{\sqrt{\rr}}\right)\right|\,.
\end{aligned}
\label{approx1}
\end{equation}
Now, to bound the second term, we use Lemma~\ref{st_projection1}, which gives
\begin{align*}
\left|\pi_v^{C\cup\{v\}}\left(\BB^\top \frac{\qq_S - \qq_{S'}}{\sqrt{\rr}}\right)\right| 
& \leq \sum\limits_{e=(u,w)\in S\backslash S'} \left(p_v^{C\cup\{v\}}(u) + p_v^{C\cup\{v\}}(w)\right) 
\frac{\sqrt{r_e}}{R_{eff}(v,e)}\,.
\end{align*}
Now, note that for each $e\in S\backslash S'$, $e$ is close to $C$, but $v$ is far from $C$, so 
$R_{eff}(v,e)$ should be large. Specifically,
by Lemma~\ref{lem:multi_effective_resistance}
we have $R_{eff}(v,e) \geq \frac{1}{2} \min\left\{R_{eff}(v,u), R_{eff}(v,w)\right\}$,
and by the triangle inequality
\[ \min\{R_{eff}(v,u), R_{eff}(v,w) \}
\geq R_{eff}(C,v) - \max\{R_{eff}(C,u), R_{eff}(C,w)\}
\geq R_{eff}(C,v) - 2 R_{eff}(C,e) \,. \]
By the fact that $e$ is $\gamma$-important 
and that 
\[ e \notin S' \supseteq \{e\in S\ |\ R_{eff}(C,v) \leq r_e / (2c^2) \}\,, \]
we have $R_{eff}(C,e) \leq r_e / \gamma^2 \leq 2 c^2 R_{eff}(C,v) / \gamma^2$,
so
\begin{align*}
\frac{\sqrt{r_e}}{R_{eff}(v,e)} 
& \leq \frac{1}{1/2 - 2c^2 / \gamma^2} \frac{\sqrt{r_e}}{R_{eff}(C,v)}\\
& \leq \frac{c}{1/2 - 2c^2 / \gamma^2} \frac{1}{\sqrt{R_{eff}(C,v)}}\,.
\end{align*}
By using the congestion reduction property (Definition~\ref{def:cong_red}), we obtain
\begin{align}
\left|\pi_v^{C\cup\{v\}}\left(\BB^\top \frac{\qq_S - \qq_{S'}}{\sqrt{\rr}}\right)\right|
& \leq 
\frac{c}{1/2 - 2c^2 / \gamma^2}
	\frac{1}{\sqrt{R_{eff}(C,v)}} \tO{\frac{1}{\beta^2}} \,.
\label{eq:far_edges_projection}
\end{align}
Setting $c = \min\{\delta_1 / \tO{\beta^{-2}}, \gamma / 4\}$
and $\delta_1' = \beta c \cdot \delta_1 / 2$,
(\ref{approx1}) becomes
\begin{align*}
& \left|
\tpi_v^{C\cup\{v\}}\left(\BB^\top \frac{\qq_S}{\sqrt{\rr}}\right)
- \pi_v^{C\cup\{v\}}\left(\BB^\top \frac{\qq_S}{\sqrt{\rr}}\right)
\right| 
 \leq \frac{\delta_1}{\sqrt{R_{eff}(C,v)}} \,.
\end{align*}
Also, the number of random walks needed for each valid pair $(u,e)$
is 
\[ \tO{\delta_1'^{-2} \log n \log \frac{1}{\beta}} = 
\tO{\delta_1^{-2} \beta^{-2} c^{-2} \log n \log \frac{1}{\beta}} = 
\tO{\left(\delta_1^{-4} \beta^{-6} + \delta_1^{-2} \beta^{-2} \gamma^{-2}\right) \log n \log \frac{1}{\beta}}\]\,.

For the last part of the lemma, we let
$S'' = \{e\in S\ |\ R_{eff}(C,v) \leq r_e / \left(\gamma/4\right)^2\}$
and write
\begin{align*}
& \left| \pi_v^{C\cup\{v\}}\left(\BB^\top \frac{\qq_S}{\sqrt{\rr}}\right)\right| \\
& \leq 
\left|\pi_v^{C\cup\{v\}}\left(\BB^\top \frac{\qq_{S''}}{\sqrt{\rr}}\right)\right|
+ \left|\pi_v^{C\cup\{v\}}\left(\BB^\top \frac{\qq_S - \qq_{S''}}{\sqrt{\rr}}\right)
\right| \,.
\end{align*}
For the first term,
\begin{align*}
\left|\pi_v^{C\cup\{v\}}\left(\BB^\top \frac{\qq_{S''}}{\sqrt{\rr}}\right)\right| 
& \leq \sum\limits_{e=(u,w)\in S''} \left(p_v^{C\cup\{v\}}(u) + p_v^{C\cup\{v\}}(w)\right) \frac{1}{\sqrt{r_e}}\\
& \leq \frac{1}{(\gamma/4) \sqrt{R_{eff}(C,v)}} \tO{\frac{1}{\beta^2}}\,,
\end{align*}
and for the second term we have already proved in (\ref{eq:far_edges_projection}) 
(after replacing $c$ by $\gamma/4$) that
\begin{align*}
\left|\pi_v^{C\cup\{v\}}\left(\BB^\top \frac{\qq_S - \qq_{S''}}{\sqrt{\rr}}\right)\right| 
& \leq \frac{\gamma/4}{\sqrt{R_{eff}(C,v)}} \tO{\frac{1}{\beta^2}}\,.
\end{align*}
Putting these together,
we conclude that 
\begin{align*}
\left| \pi_v^{C\cup\{v\}}\left(\BB^\top \frac{\qq_S}{\sqrt{\rr}}\right)\right| 
\leq \frac{1}{\gamma \sqrt{R_{eff}(C,v)}}\cdot \tO{\frac{1}{\beta^2}} \,.
\end{align*}
\end{proof}

\subsection{Proof of Lemma~\ref{conc2}}
\label{proof_conc2}
\begin{proof}
For any $I\subseteq \mathbb{R}$, we define
$F_I = \{i\in[n]\ |\ |\ophi_i| \in I\}$.
For some $0 < a < b$ to be defined later, we partition $[n]$ as
\[ [n] = F_{I_0} \cup F_{I_1}\cup \dots\cup F_{I_K}\cup F_{I_{K+1}} \,,\]
where $I_0 = [0,a)$, $I_{K+1} = [b,\infty)$, and
$I_1,\dots,I_K$ is a partition of $[a,b)$ into $K = O(\log \frac{b}{a})$ intervals
such that for all $k\in[K]$ 
we have $\Phi_k := \underset{i\in F_k}{\max}\, |\ophi_i| \leq 2\cdot \underset{i\in F_k}{\min}\, |\ophi_i|$.

A union bound gives
\begin{align*}
 \Pr\left[\left|\langle \tvpi - \vpi, \ovphi \rangle\right| > t\right]
 \leq \sum\limits_{k=0}^{K+1} \Pr\left[\left|\sum\limits_{i\in I_k} (\tpi_i - \pi_i) \ophi_i\right| > t / (K+2)\right]\,.
\end{align*}
We first examine $I_0$ and $I_{K+1}$ separately. Note that
\begin{align*}
\left|\sum\limits_{i\in F_{I_0}} (\tpi_i - \pi_i) \ophi_i\right| \leq \|\tvpi - \vpi\|_1 a \leq 2 a
\end{align*}
and
\begin{align*}
\left|\sum\limits_{i\in F_{I_{K+1}}} (\tpi_i - \pi_i) \ophi_i\right| 
\leq \sum\limits_{i\in F_{I_{K+1}}} \tpi_i\left|\ophi_i\right| + 
\sum\limits_{i\in F_{I_{K+1}}} \pi_i \left|\ophi_i\right|
\leq 
\frac{1}{b} \sum\limits_{i\in F_{I_{K+1}}} |\tpi_i - \pi_i| \ophi_i^2
\leq \frac{1}{b}\sum\limits_{i\in F_{I_{K+1}}} \tpi_i |\ophi_i| + \frac{\left\|\ophi\right\|_{\vpi,2}^2}{b} 
\end{align*}
But note that by picking 
$b \geq \max\left\{\frac{(K+2) Var_{\vpi}(\ovphi)}{t}, \sqrt{Var_{\vpi}(\ovphi) \cdot n^{101}}\right\}$, 
we have $\frac{Var_{\vpi}(\ophi)}{b} \leq t / (K+2)$
and also for any $i\in F_{K+1}$ we have
$\pi_i \leq \frac{Var_{\vpi}(\ovphi)}{b^2} \leq \frac{1}{n^{101}}$.
This means that $\Pr[\tpi_i \neq 0] \leq \frac{1}{n^{101}}$, and so by union bound
\[ \Pr\left[\sum\limits_{i\in F_{I_{K+1}}} \tpi_i |\ophi_i| \neq 0\right] \leq \frac{1}{n^{100}}\,.\]

Now, we proceed to $F_1,\dots,F_K$. We draw $Z$ samples $x_1,\dots,x_Z$ from $\vpi$.
Then, we also define the following random variables for $z\in [Z]$ and $i\in[n]$:
\[ X_{z,i} = \begin{cases}1 & \text{ if $x_z = i$} \\ 0 & \text{ otherwise }\end{cases}\] for $i\in[n]$
and 
\[ Y_{z,k} = \frac{1}{Z} \sum\limits_{i\in F_k} X_{z,i} \ophi_i \]
This allows us to write $\sum\limits_{i\in F_k} \tpi_i \ophi_i = \sum\limits_{z=1}^Z Y_{z,k}$.

Fix $k\in[K]$.
We will apply Lemma~\ref{conc1} on the random variable $\sum\limits_{z=1}^Z Y_{z,k}$.
We first compute
\begin{align*}
\sum\limits_{z=1}^Z |\mathbb{E}[Y_{z,k}]| = \left|\sum\limits_{i\in F_k} \pi_i \ophi_i\right| \leq \sum\limits_{i\in F_k} \pi_i |\ophi_i|:= E_k
\end{align*}
and
\begin{align*}
\underset{z\in[Z]}{\max}\, |Y_{z,k}| \leq \frac{\Phi_k}{Z} := M_k\,.
\end{align*}
Therefore we immediately have $E_kM_k \leq \frac{2}{Z} \sum\limits_{i\in F_k} \pi_i \ophi_i^2 \leq \frac{2}{Z} \cdot Var_{\pi}(\ophi) $. By Lemma~\ref{conc1},
\begin{align*}
& \Pr\left[\left|\sum\limits_{z=1}^Z Y_{z,k} - \mathbb{E}\left[\sum\limits_{z=1}^Z Y_{z,k}\right]\right| > t / (K+2)\right] \\
& \leq 2 \exp\left(-\frac{t^2}{6E_kM_k (K+2)^2}\right)\\
& \leq 2 \exp\left(-\frac{Z t^2}{12\cdot Var_{\pi}(\ophi) (K+2)^2}\right)\,.\\
\end{align*}
Summarizing, and using the fact that $K = \tO{\log (n\cdot Var_{\vpi}(\ovphi) / t^2)}$, we get 
\[\Pr\left[\left|\langle \tvpi - \vpi, \ovphi\rangle\right| > t\right] \leq \tO{\frac{1}{n^{100}}} + 2\tO{\log \left(n \cdot Var_{\vpi}(\ovphi) / t^2\right)}\exp\left(- \frac{Z t^2}{12 \cdot\tO{Var_{\vpi}(\ovphi) \log^2 n}} \right)\,.\]
\end{proof}

\subsection{Proof of Lemma~\ref{lem:variance}}
\label{proof_lem_variance}
\begin{proof}
Let $S_0 = \emptyset$ and 
and for each $k\in\mathbb{N}$ let \[S_k = \{i\in[n]\backslash S_{k-1}\ :\ \phi_i^2 \leq 2^{k+1} R_{eff}(C,v)\}\,.\]
Fix some $k \geq 2$. Note that
$\phi_i^2 > 2^k R_{eff}(C,v)$ for all $i\in S_k$,
implying
$\frac{2^{k} R_{eff}(C,v)}{R_{eff}(S_k,v)} < 
E_{\rr}(\vphi) \leq
1$, and so $R_{eff}(S_k,v) > 2^k R_{eff}(C,v) \geq 4 R_{eff}(C,v)$.
As
\[ R_{eff}(C,v) \geq \frac{1}{4} \min\{R_{eff}(S_k,v), R_{eff}(C\backslash S_k,v)\} > \min\{R_{eff}(C,v), \frac{1}{4} R_{eff}(C\backslash S_k,v)\} \,,\]
we have
$R_{eff}(C\backslash S_k, v) < 4 R_{eff}(C,v)$.
This implies that $\left\|\pi_{S_k}^C(\onev_v)\right\|_1 \leq \frac{R_{eff}(C\backslash S_k, v)}{R_{eff}(S_k, v)} < \frac{1}{2^{k-2}}$.

So we conclude that $Var_{\vpi}(\vphi) = \sum\limits_{i\in S_k} \pi_i \phi_i^2 \leq \frac{1}{2^{k-2}} \cdot 2^{k+1} R_{eff}(C,v) = 8 R_{eff}(C,v)$.
\end{proof}

\subsection{Proof of Lemma~\ref{lem:insert1}}
\label{proof_lem_insert1}

\begin{proof}
The first part of the statement is given by 
applying Lemma~\ref{estimate1_final}, and we see that it requires
$\tO{\delta_1^{-4} \beta^{-6} + \delta_1^{-2} \beta^{-2} \gamma^{-2}}$ 
random walks for each $u\in V\backslash C$ and $e\in E\backslash E(C)$ with $u\in e$.

For the second part we use the fact that the change in the demand projection after inserting $v$ into $C$ is given by
\[
\vpi^{C\cup \{v\}}\left(\BB^\top \frac{\qq_S}{\sqrt{\rr}}\right)
-
\vpi^{C}\left(\BB^\top \frac{\qq_S}{\sqrt{\rr}}\right)
=
\pi_v^{C\cup \{v\}}\left(\BB^\top \frac{\qq_S}{\sqrt{\rr}}\right)\cdot (\onev_v - \vpi^C(\onev_v))\,,
\]
and therefore we can estimate this update via
\[
\tpi_v^{C\cup \{v\}}\left(\BB^\top \frac{\qq_S}{\sqrt{\rr}}\right)\cdot (\onev_v - \tvpi^C(\onev_v))\,.
\]
where $\tpi_v^{C\cup \{v\}}\left(\BB^\top \frac{\qq_S}{\sqrt{\rr}}\right)$
is the estimate we computed using Lemma~\ref{estimate1_final} 
and $\tvpi^C(\onev_v)$ is obtained by applying Lemma~\ref{estimate2}.

Let us show that this estimation indeed introduces only a small amount of error. 
For any $\vphi$, such that $E_r(\vphi) \leq 1$, we can write
\begin{align*}
&\left| \left\langle   
\tpi_v^{C\cup \{v\}}\left(\BB^\top \frac{\qq_S}{\sqrt{\rr}}\right)\cdot (\onev_v - \tvpi^C(\onev_v))
-
\pi_v^{C\cup \{v\}}\left(\BB^\top \frac{\qq_S}{\sqrt{\rr}}\right)\cdot (\onev_v - \vpi^C(\onev_v))
,
\vphi
 \right\rangle \right|
 \\
 &\leq
 \left| \left\langle   
\pi_v^{C\cup \{v\}}\left(\BB^\top \frac{\qq_S}{\sqrt{\rr}}\right)\cdot (\vpi^C(\onev_v) - \tvpi^C(\onev_v))
,
\vphi  \right\rangle \right|\\
 &+
 \left| \left\langle   
\left( \tpi_v^{C\cup \{v\}}\left(\BB^\top \frac{\qq_S}{\sqrt{\rr}}\right)
-
\pi_v^{C\cup \{v\}}\left(\BB^\top \frac{\qq_S}{\sqrt{\rr}}\right)\right)\cdot (\onev_v - \vpi^C(\onev_v))
,
\vphi  \right\rangle \right|\\
&+
 \left| \left\langle   
\left(\tpi_v^{C\cup \{v\}}\left(\BB^\top \frac{\qq_S}{\sqrt{\rr}}\right)-  \pi_v^{C\cup \{v\}}\left(\BB^\top \frac{\qq_S}{\sqrt{\rr}}\right) \right)\cdot (\vpi^C(\onev_v) - \tvpi^C(\onev_v))
,
\vphi  \right\rangle \right|\,.
\end{align*}
At this point we can bound these quantities using Lemmas~\ref{estimate1_final} and~\ref{estimate2}. It is important to notice that they require that $S$ is a set of $\gamma$-important edges, for some parameter $\gamma$.
Our congestion reduction subset $C$ keeps increasing due to vertex insertions. This, however, means that effective resistances between any vertex in $V\setminus C$ and $C$ can only decrease, and therefore the set of important edges can only increase. Thus we are still in a valid position to apply these lemmas.

Using $E_{\rr}(\vphi)\leq 1$, which allows us to write: 
\[
\langle \onev_v - \vpi^C(\onev_v)), \vphi \rangle \leq \mathcal{E}_{\rr}\left( \onev_v - \vpi^C(\onev_v) \right) = R_{eff}(v,C)\,,
\]
we can continue to upper bound the error by:
\begin{align*}
&\frac{\tO{\gamma^{-1} \beta^{-2}}}{\sqrt{R_{eff}(C,v)}}\cdot \delta_2 \sqrt{R_{eff}(C,v)}
	+ \frac{\delta_1}{\sqrt{R_{eff}(C,v)}} \cdot \sqrt{R_{eff}(C,v)}
	+ \frac{\delta_1}{\sqrt{R_{eff}(C,v)}} \cdot \delta_2 \sqrt{R_{eff}(C,v)}\\
&= \delta_2 \cdot \tO{\gamma^{-1} \beta^{-2}} + \delta_1 + \delta_1\delta_2 \,.
\end{align*}
Setting $\delta_1 = \heps / 2$ and $\delta_2 = {\heps \beta^2 \gamma}/\tO{1}$, we conclude that w.h.p. each operation introduces at most $\heps$ additive error in the maintained estimate for $\left\langle \tvpi^C(\BB^\top \frac{\qq_S}{\sqrt{\rr}}),\vphi \right\rangle$.

Per Lemma~\ref{estimate1_final}, estimating one coordinate of the demand projection requires 
\[
\tO{\delta_1^{-4} \beta^{-6} + \delta_1^{-2} \beta^{-2} \gamma^{-2}} 
 = \tO{\heps^{-4} \beta^{-6} + \heps^{-2} \beta^{-2} \gamma^{-2}}\
\]
  random walks,
and estimating $\tvpi^C(\BB^\top \frac{\qq_S}{\sqrt{\rr}})$, per Lemma~\ref{estimate2}, requires 
\[
 \tO{\delta_2^{-2} } = \tO{\heps^{-2} \beta^{-4} \gamma^{-2}}
\]
random walks. This concludes the proof.

\end{proof}

\section{The \textsc{Checker} Data Structure}
\label{sec:checker}

\begin{theorem}[Theorem 3, \cite{gao2021fully}]
There is a \textsc{Checker} data structure supporting the following operations with the given runtimes
against oblivious adversaries, for parameters $0 < \bc,\eps < 1$ such that
$\bc \geq \tOm{\eps^{-1/2}/m^{1/4}}$.
\begin{itemize}
\item{\textsc{Initialize}$(\ff, \eps, \bc)$:
Initializes the data structure with slacks 
$\ss^+ = \uu - \ff$,
$\ss^- = \ff$,
and resistances $\rr = \frac{1}{(\ss^+)^2} + \frac{1}{(\ss^-)^2}$.
Runtime: $\tO{m\bc^{-4} \eps^{-4}}$.
}
\item{\textsc{Update}$(e,\ff')$:
Set $s_e^+ = u_e - f_e'$, $s_e^- = f_e'$, and $r_e = \frac{1}{(s_e^+)^2} + \frac{1}{(s_e^-)^2}$. 
Runtime: Amortized $\tO{\bc^{-2} \eps^{-2}}$.
}
\item{\textsc{TemporaryUpdate}$(e,\ff')$:
Set $s_e^+ = u_e - f_e'$, $s_e^- = f_e'$, and $r_e = \frac{1}{(s_e^+)^2} + \frac{1}{(s_e^-)^2}$. 
Runtime: Worst case $\tO{(K \bc^{-2} \eps^{-2})^2}$, where $K$
is the number of $\textsc{TemporaryUpdate}$s that have not been rolled back
using $\textsc{Rollback}$. All $\textsc{TemporaryUpdate}$s should be rolled back
before the next call to $\textsc{Update}$.
}
\item{\textsc{Rollback}$()$: Rolls back the last $\textsc{TemporaryUpdate}$ if it
exists. The runtime is the same as the original operation.
}
\item{\textsc{Check}$(e, \vpi_{old})$: 
Returns $\tf_e$ such that $\sqrt{r_e}\left|\tf_e - \tf_e^*\right| \leq \eps$, where
\[ \tff^* = \delta g(\ss) - \delta \RR^{-1} \BB \left(\BB^\top \RR^{-1} \BB \right)^+ \BB^\top g(\ss)\,, \] 
for $\delta = 1/\sqrt{m}$. Additionally, a vector $\vpi_{old}$ that is supported
on $C$ such that
\[ \mathcal{E}_{\rr}\left(
 \vpi_{old} - \vpi^{C} \left(\BB^\top g(\ss)\right) \right) \leq \eps^2 m / 4 \]
is provided, where $C$ is the vertex set of the dynamic sparsifier 
in the $\textsc{DynamicSC}$ that is maintained internally.
Runtime: Worst case $\tO{\left(\bc m + (K\bc^{-2} \eps^{-2})^2\right) \eps^{-2} }$, where $K$
is the number of $\textsc{TemporaryUpdate}$s that have not been rolled back.
Additionally, the output of \textsc{Check}$(e)$ is independent of any previous calls to
\textsc{Check}.
}
\end{itemize}
Finally, all calls to $\textsc{Check}$ return valid outputs with high probability.
The total number of \textsc{Update}s and \textsc{TemporaryUpdate}s that have not been
rolled back should always be $O(\bc m)$.
\label{thm:checker}
\end{theorem}

This theorem is from~\cite{gao2021fully}. The only difference is in the guarantee of $\textsc{Check}$. We will now show how it can be implemented.
Let 
\[ \tff^* = \delta g(\ss) - \delta \RR^{-1} \BB \LL^+ \BB^\top g(\ss)\,. \] 
Let $\textsc{DynamicSC}$ be the underlying Schur complement data structure.
We first add the endpoints $u,w$ of $e$ as terminals by calling 
\[ \textsc{DynamicSC}.\textsc{TemporaryAddTerminals}(\{u,w\}) \]
so that the new Schur complement is on the vertex set $C' = C\cup\{u,w\}$.
Then, we set
\[ \vphi = -\tSC^+ \vpi_{old}\]
and $\tf_e = (\phi_u - \phi_w) / \sqrt{r_e}$, where $\tSC$ is the output of 
$\textsc{DynamicSC}.\tSC()$. Equivalently, note that
\[ \tf_e = \delta \cdot \onev_e^\top \RR^{-1} \BB \LL^+ \vpi_{old}\,.\]
We will show that $\sqrt{r_e}\left|\tf_e - \tf_e^*\right| \leq \eps$.

We write
\begin{align*}
&  \sqrt{r_e} \left|\tf_e - \tf_e^*\right| \\
& \leq \left\|\delta \sqrt{\rr} g(\ss)\right\|_\infty
+ \left\|\delta \RR^{-1/2} \BB \LL^+ \left(\BB^\top g(\ss) - \vpi^{C'}\left(\BB^\top g(\ss)\right)\right)\right\|_\infty\\
& + \left\|\delta \RR^{-1/2} \BB \LL^+ \left(
\vpi^{C'}\left(\BB^\top g(\ss)\right)
- \vpi^C\left(\BB^\top g(\ss)\right)
\right)\right\|_\infty
+ \left\|\delta \RR^{-1/2} \BB \LL^+ \left(\vpi^C\left(\BB^\top g(\ss)\right) - \vpi_{old}\right)\right\|_\infty\,.
\end{align*}
For the first term, 
\[ \left\|\delta \sqrt{\rr} g(\ss)\right\|_\infty = \left\|\delta \frac{\frac{1}{\ss^+} - \frac{1}{\ss^-}}{\sqrt{\frac{1}{(\ss^+)^2} + \frac{1}{(\ss^-)^2}}}\right\|_\infty \leq \delta \leq \eps / 10\,.\]
Now, by the fact that $C'$ is a $\bc$-congestion reduction subset by definition in $\textsc{DynamicSC}$, 
Lemma~\ref{lem:non-projected-demand-contrib} immediately implies that the second term is 
$\leq \delta \cdot \tO{\bc^{-2}} \leq \eps/10$.

For the third term, we apply Lemma~\ref{lem:old_projection_approximate}, which shows that
\begin{align*}
& \left\|\delta \RR^{-1/2} \BB \LL^+ \left(
\vpi^{C'}\left(\BB^\top g(\ss)\right)
- \vpi^C\left(\BB^\top g(\ss)\right)
\right)\right\|_\infty\\
& \leq 
\delta \cdot \sqrt{\mathcal{E}_{\rr}\left(
\vpi^{C'}\left(\BB^\top g(\ss)\right)
- \vpi^C\left(\BB^\top g(\ss)\right)
\right)}\\
& \leq \delta \cdot \tO{\bc^{-2}}\\
& \leq \eps / 10\,,
\end{align*}
as the resistances don't change and we only have two terminal insertions from $C$ to $C'$.

Finally, the fourth term is 
\[\left\|\delta \RR^{-1/2} \BB \LL^+ \left(\vpi^C\left(\BB^\top g(\ss)\right) - \vpi_{old}\right)\right\|_\infty
\leq \delta \cdot \sqrt{\mathcal{E}_{\rr}
\left(\vpi^C\left(\BB^\top g(\ss)\right) - \vpi_{old}\right)}
\leq \delta \cdot \eps \sqrt{m} / 2
= \eps / 2\,.
\]

We conclude that $\sqrt{r_e}\left|\tf_e - \tf_e^*\right| \leq \eps$.
Finally, we call $\textsc{DynamicSC}.\textsc{Rollback}$ to undo the terminal insertions.

The runtime of this operation is dominated by the call to 
$\textsc{DynamicSC}.\tSC()$, which takes time 
$\tO{(\bc m + (K\bc^{-2} \eps^{-2})^2)\eps^{-2}}$.

\bibliographystyle{alpha}
\bibliography{main}

\end{document}